\documentclass[12pt]{article}

\usepackage{ae,aecompl}

\usepackage{amsmath}
\usepackage{amssymb} 
\usepackage{amsthm} 
\usepackage{bm} 
\usepackage{commath} 
\usepackage{geometry}
\usepackage{graphicx} 
\usepackage{mathrsfs} 
\usepackage{mleftright} 
\mleftright 
\usepackage[authoryear]{natbib}
\usepackage{paralist} 
\usepackage{setspace}
\usepackage{verbatim} 

\usepackage[backref=page, hyperindex=true]{hyperref}
\hypersetup{%
    breaklinks=true,
    colorlinks=true,
    linkcolor=magenta,
    urlcolor=blue,
    citecolor=blue,
    pdfstartview={FitH},
    unicode=true
 }
\usepackage{bookmark} 

\let \backreforig \backref
\renewcommand*{\backref}[1]{[\backreforig{#1}]}

\theoremstyle{plain}
\newtheorem{assumption}{Assumption}
\numberwithin{assumption}{section}
\theoremstyle{plain}
\newtheorem{lem}{Lemma}
\numberwithin{lem}{section}
\theoremstyle{plain}
\newtheorem{thm}{Theorem}
\numberwithin{thm}{section}
\theoremstyle{definition}
\newtheorem{rem}{Remark}
\numberwithin{rem}{section}

\def\E{\mathrm{E}}

\def\bbG{\mathbb{G}}

\def\bbP{\mathbb{P}}
\def\N{\mathbb{N}}
\def\R{\mathbb{R}}
\def\var{\mathrm{var}}

\def\bzero{\boldsymbol{0}}
\def\bone{\boldsymbol{1}}
\def\be{\boldsymbol{e}}
\def\bg{\boldsymbol{g}}
\def\bG{\boldsymbol{G}}
\def\bh{\boldsymbol{h}}
\def\bfH{\mathbf{H}}
\def\bfJ{\mathbf{J}}
\def\bfL{\mathbf{L}}
\def\bfR{\mathbf{R}}
\def\bu{\boldsymbol{u}}

\def\bfV{\mathbf{V}}
\def\bw{\boldsymbol{w}}
\def\bW{\boldsymbol{W}}
\def\bx{\boldsymbol{x}}
\def\bX{\boldsymbol{X}}
\def\tbX{\widetilde{\boldsymbol{X}}}
\def\by{\boldsymbol{y}}
\def\bY{\boldsymbol{Y}}
\def\tbY{\widetilde{\boldsymbol{Y}}}
\def\bz{\boldsymbol{z}}
\def\bZ{\boldsymbol{Z}}
\def\tbZ{\widetilde{\boldsymbol{Z}}}

\def\cA{\mathcal{A}}
\def\cB{\mathcal{B}}
\def\cC{\mathcal{C}}
\def\cE{\mathcal{E}}
\def\cF{\mathcal{F}}
\def\cG{\mathcal{G}}
\def\cH{\mathcal{H}}
\def\cN{\mathcal{N}}
\def\cR{\mathcal{R}}
\def\cW{\mathcal{W}}
\def\cZ{\mathcal{Z}}

\def\bdelta{\boldsymbol{\delta}}
\def\bepsilon{\boldsymbol{\varepsilon}}
\def\bpsi{\boldsymbol{\psi}}
\def\btheta{\boldsymbol{\theta}}
\def\tbtheta{\widetilde{\boldsymbol{\theta}}}
\def\bvtheta{\boldsymbol{\vartheta}}
\def\bGamma{\boldsymbol{\Gamma}}
\def\bSigma{\boldsymbol{\Sigma}}

\def\odelta{\overline{\delta}}

\newcommand*{\argmin}{\operatornamewithlimits{argmin}\limits}

\usepackage[affil-it]{authblk} 

\begin{document}

\title{Asymptotic Variance Theory for Trimmed Least Squares and Trimmed Least
Absolute Deviations in Censored Panel Models with Fixed Effects\footnote{We
thank Matias Cattaneo, Jin Hahn, Michael Jansson, and Whitney Newey for helpful
comments and suggestions.}}

\author[1]{Denis Chetverikov}
\author[2,4]{Jesper R.-V.~S{\o}rensen}
\author[3,4]{Bo Honor{\'e}}
\affil[1]{\small{University of California, Los Angeles}}
\affil[2]{\small{University of Copenhagen}}
\affil[3]{\small{Princeton University}}
\affil[4]{\small{Aarhus Center for Econometrics (ACE)}}

\date{\today}

\maketitle

\begin{abstract}
We study inference using trimmed least squares (TLS) and trimmed least
absolute deviations (TLAD) estimators of \citet{honore_trimmed_1992} in censored
two-period panel-data models with fixed effects. We show that the published
asymptotic variance formulas rely on additional regularity conditions that are
not fully stated in the original analysis. For TLS, the published Hessian formula requires that the regressor-difference index vanish only when the regressor difference itself is zero, a restriction not explicitly stated in the original paper and violated, for instance, with a zero parameter vector.
We derive the correct Hessian, establish asymptotic normality without imposing this restriction, and obtain a consistent plug-in variance estimator. We also show that the Hessian estimator proposed in \citet{honore_trimmed_1992} {\em is} actually consistent for the {\em correct} TLS asymptotic
variance. For TLAD, we show that the published variance formula omits a
conditional-probability term and that asymptotic normality requires additional
continuity conditions. Under these conditions, we derive the corrected
asymptotic variance and provide a tuning-parameter-free bootstrap variance
estimator.

\end{abstract}

\section{Introduction}

Inference for censored panel-data models with fixed effects often relies on
trimmed estimators whose objective functions are convex but nonsmooth.
\citet{honore_trimmed_1992} introduced two influential examples, the trimmed
least squares (TLS) and trimmed least absolute deviations (TLAD) estimators, and
derived their asymptotic distributions for a two-period censored regression
model with fixed effects. These estimators remain natural benchmark procedures
in this setting, both because they avoid parametric restrictions on the fixed
effects and because they provide tractable ways to exploit the structure of
two-period panel data under censoring. Valid inference for such estimators,
however, depends delicately on the form of the Hessian of the expected loss, and
the nonsmoothness created by trimming and censoring makes that Hessian less
straightforward than the published formulas suggest.

This paper revisits the asymptotic variance theory for the TLS and TLAD
estimators of \citet{honore_trimmed_1992}. We show that the published asymptotic
variance formulas rely on additional regularity conditions that are not fully
reflected in the original theorem statements. For TLS, the problem is that
censoring creates boundary points at which the score of the trimmed square loss
is nondifferentiable with positive probability. At such points, differentiation
under expectation need not be valid, so the Hessian of the expected loss can
differ from the expression obtained by differentiating inside the expectation.
For TLAD, the problem is that the expected loss can fail to be twice
differentiable when the conditional distribution of the model errors is not
sufficiently smooth. Our results identify how these issues enter the asymptotic
variance analysis for TLS and TLAD, derive the corrected variance formulas, and
clarify the scope of valid legacy inference.

Our starting point is the two-period censored regression model in which the
outcome variable $Y_\tau$ in period $\tau\in\{1,2\}$ is generated as
\begin{equation}\label{eq:Outcomes}
Y_\tau=\max\{0,Y_\tau^\ast\},\qquad
Y_\tau^\ast=\alpha+\bX_\tau^\top\btheta_0+\varepsilon_\tau,
\end{equation}
where $\bX_\tau\in\R^K$ is a vector of time-varying regressors, $\btheta_0\in
\R^K$ is the parameter of interest, $\alpha$ is an unobserved fixed effect, and
$\varepsilon_\tau$ is an unobserved error term. Let
$\{(Y_{i1},\bX_{i1},Y_{i2},\bX_{i2})\}_{i=1}^n$ be a random sample from the
distribution of $(Y_1,\bX_1,Y_2,\bX_2)$. \citet{honore_trimmed_1992} proposed
M-estimators of $\btheta_0$ based on a trimmed loss function
$m_\Xi:\R\times[0,\infty)\times[0,\infty)\to\R$ defined by
\begin{equation}\label{eq:TrimmedGenericLoss}
  m_{\Xi}(t,\by):=\begin{cases}
  \Xi(y_{1})-\left(y_{2}+t\right)\xi(y_{1}), &t\leqslant-y_2,\\
  \Xi(y_{1}-y_{2}-t), &t\in\left(-y_{2},y_{1}\right),\\
  \Xi(-y_{2})-\left(t-y_{1}\right)\xi(-y_{2}), &t\geqslant y_{1},
  \end{cases}
\end{equation}
where $\Xi(\cdot)$ represents either the absolute loss $\left|\cdot\right|$ or
the (one-half) square loss
$\textstyle\frac{1}{2}(\cdot)^2$,\footnote{\citet[Assumption
L1]{honore_pairwise_1994} provide a list of conditions allowing other choices of
$\Xi$.} and $\xi(\cdot)$ is its derivative.\footnote{In the case of the absolute
loss $\Xi(\cdot)=\left|\cdot\right|$, we set $\xi(0) = 0$.} These choices yield
the TLAD and TLS estimators, respectively, defined for the corresponding
$\Xi(\cdot)$ as any solution
\begin{align}\label{eq:TrimmedEstimator}
  \widehat{\btheta}_{\Xi} &\in \argmin_{\btheta\in\R^K}\cbr[3]{\frac{1}{n}\sum_{i=1}^n m_{\Xi}\big(\Delta\bX_i^\top\btheta,\bY_i\big)},
\end{align}
where we introduced the shorthands $\Delta\bX_i := \bX_{i1} - \bX_{i2}$ and
$\bY_i:=(Y_{i1},Y_{i2})$. Under conditional exchangeability of
$(\varepsilon_1,\varepsilon_2)$ given $(\bX_1,\bX_2,\alpha)$ and other
regularity conditions, \citet{honore_trimmed_1992} established consistency and
asymptotic normality of both estimators.

Our first set of results concerns the TLS estimator. We show that the Hessian
formula appearing in \citet{honore_trimmed_1992} is valid only if
\begin{equation}\label{eq: Bo condition}
\bbP\del[1]{\Delta\bX^\top\btheta_0 = 0\text{ and }\Delta\bX\neq\bzero} = 0,
\end{equation}
a condition not explicitly stated in the original paper. This condition is
satisfied, for example, when the index $\Delta\bX^\top\btheta_0$ has a
continuous distribution, but it fails when all components of
$\btheta_0$ are equal to zero unless $\Delta\bX = \bzero$ almost surely. When condition \eqref{eq: Bo condition} is not satisfied, the TLS score is nondifferentiable with positive probability, and the published Hessian formula need not equal the derivative of the expected score. We derive
the correct Hessian, provide an explicit representation of it, establish
asymptotic normality without excluding mass at zero, and obtain a consistent
plug-in estimator of the corresponding asymptotic variance.

The TLS analysis also yields a second result. Although the population Hessian
formula in \citet{honore_trimmed_1992} is not generally correct, the Hessian
estimator proposed there remains consistent for the correct TLS asymptotic
variance under the maintained assumptions. Thus, the paper does not only
identify where the published population formula fails; it also clarifies which
parts of the original inference procedure remain usable.

Our second set of results concerns the TLAD estimator. We show, first, that the
published asymptotic variance formula omits a conditional-probability term. More
importantly, we show that asymptotic normality requires additional continuity
conditions to ensure existence of the Hessian of the expected loss. We provide a
counterexample satisfying the conditions stated in \citet{honore_trimmed_1992}
for which that Hessian does not exist. Under additional continuity conditions,
we derive the corrected TLAD Hessian and corresponding asymptotic variance
formula. We then provide a tuning-parameter-free bootstrap variance estimator
for TLAD.

Taken together, these results clarify the role of boundary nondifferentiability
in the asymptotic variance theory for censored trimmed estimators. More
specifically, they show when legacy inference remains valid, when corrected
variance formulas are required, and how to implement inference under the
corrected theory.

As a by-product, our TLS analysis also yields a plug-in perspective on the
asymptotic variance of the cross-sectional trimmed least squares estimator
studied by \citet{honore_pairwise_1994}, thus avoiding the choice of tuning
parameters required by the numerical-differentiation-based variance estimator
proposed there.

The remainder of the paper is organized as follows. Section~\ref{sec:TLS}
revisits the TLS estimator, beginning with the asymptotic normality result in
\citet{honore_trimmed_1992}, then identifying the source of the problem,
deriving the correct Hessian, establishing asymptotic normality under the
corrected theory, and discussing variance estimation, including the status of
the legacy Hessian estimator. Section~\ref{sec:TLAD} carries out the analogous
analysis for TLAD, derives the corrected Hessian and asymptotic variance
formula, and develops bootstrap-based variance estimation. Proofs for the TLS results are in the appendix, while proofs for the TLAD results are in the supplemental appendix, which also contains technical lemmas establishing a generalized dominated convergence theorem, the existence of conditional PDFs,
and related measurability issues.

\paragraph{Notation.} 
For a positive integer $k$, we write $\left[k\right]:=\{1,\ldots,k\}$. We use
$\left\|\cdot\right\|_2$ to denote the Euclidean norm. For any (suitably
differentiable) real-valued function $f$, whose first argument is a scalar, we
let $\dot{f}_1$ and $\ddot{f}_{11}$ denote the first- and second-order partial
derivatives of $f$ with respect to its first argument, respectively. For a
possibly vector-valued differentiable function $\bg$, we let $\nabla \bg$
denote the Jacobian of $\bg$. If $\bg$ is itself the gradient mapping
corresponding to some function $f$, then we write $\nabla^2 f$ for the Hessian
of $f$. The indicator $\mathbf{1}\{A\}$ equals one (zero) if the logical
statement $A$ is true (false). We use comma-separated events to denote the
intersection of events, e.g.~$\mathbf 1\{A,B\} = \mathbf 1\{A\cap B\}$. Unless
otherwise stated, all limits are understood as the sample size $n$ grows
without bound, holding everything else fixed. The symbols $\rightsquigarrow$,
$\to_{\text{a.s.}}$, and $\to_{\bbP}$ denote convergences in distribution,
almost surely, and in probability, respectively. We write $o_{\bbP}(1)$,
$o_{\mathrm{a.s.}}(1)$, and $o_{L^1}(1)$ to denote sequences of random
variables that converge to zero in probability, almost surely, and in
$L^1(\bbP)$, respectively.

\section{Trimmed Least Squares}\label{sec:TLS} In this section, we focus on the
TLS estimator. We abbreviate the trimmed square loss, $m_{\Xi}$ in
\eqref{eq:TrimmedGenericLoss} with $\Xi=\textstyle{\frac{1}{2}}(\cdot)^2$, by
$m^{\texttt{tls}}$, which takes the explicit form
\begin{equation}\label{eq:TrimmedSqLoss}
m^{\texttt{tls}}(t,\by)=\frac{1}{2}\cdot\begin{cases}
y_{1}^2-2\left(y_{2}+t\right)y_{1}, &t\leqslant-y_2,\\
\left(y_{1}-y_{2}-t\right)^2, &t\in\left(-y_{2},y_{1}\right),\\
y_{2}^2+2\left(t-y_{1}\right)y_{2}, &t\geqslant y_{1}.
\end{cases}
\end{equation}
This loss is continuously differentiable in its first argument with partial derivative given by
\begin{equation}\label{eq:TrimmedSqLossDerivative}
\dot{m}^{\texttt{tls}}_1(t,\by)=\begin{cases}
  -y_{1}, &t\leqslant-y_{2},\\
  y_{2}-y_{1}+t, &t\in(-y_{2},y_{1}),\\
  y_{2}, &t\geqslant y_{1}.
  \end{cases}
\end{equation}
To facilitate our presentation below, note that it follows from
\eqref{eq:TrimmedSqLossDerivative} that $\dot{m}^{\texttt{tls}}_1(\cdot,\by)$ is
Lipschitz continuous on $\R$ with Lipschitz constant one, uniformly in
$\by\in[0,\infty)\times[0,\infty)$. However, it fails to be differentiable at
$-y_2$ and $y_1$ (unless both $y_1$ and $y_2$ are zero), implying that the
trimmed square loss is not necessarily twice differentiable. Specifically, the
points of second-order \emph{non}-differentiability of
$m^{\texttt{tls}}(\cdot,\by)$ are captured by set-valued function
$N:[0,\infty)\times[0,\infty)\rightrightarrows\R$, defined by
\begin{equation}\label{eq:TrimmedSqLossNonTwiceDiff}
N(\by):=\left\{t\in\R\middle| \ddot{m}_{11}^{\texttt{tls}}(t,\by)\text{ does not exist}\right\}=
\begin{cases}
\emptyset, &y_{1}+y_{2}=0,\\
\{-y_{2},y_{1}\}, &y_{1}+y_{2}>0.
\end{cases}
\end{equation}

\subsection{Assumptions for Trimmed Least Squares}\label{sec:Assumptions-TLS}
For notational convenience, abbreviate the TLS estimator $\widehat\btheta_\Xi$
in \eqref{eq:TrimmedEstimator} with
$\Xi=\textstyle{\frac{1}{2}}\left(\cdot\right)^2$ by
$\widehat{\btheta}^{\texttt{tls}}$. Also, denote $\bW:=(\bX_1,\bX_2,\alpha)$ and
let $\cW\subseteq\R^{K}\times \mathbb R^K\times\mathbb R$ be the support of
$\bW$, with $\bw := (\bx_1,\bx_2,a)$ standing for a generic element of $\cW$. In
addition, denote $\bepsilon:=(\varepsilon_1,\varepsilon_2)$ and let $\mathcal E
:= \mathbb R\times\mathbb R$, with $\be = (e_1,e_2)$ standing for a generic
element of $\mathcal E$. \citet{honore_trimmed_1992} derived the asymptotic
distribution of the TLS estimator under the following
assumptions.\footnote{Assumptions
\ref{assu:Non-Degeneracy-TLS}--\ref{assu:RankRegressors-TLS}
are from Assumptions S.2, M.3, E.1, E.3 and R.1
, respectively, in \citet{honore_trimmed_1992}.}
\begin{assumption}[\textbf{Non-Degeneracy}]\label{assu:Non-Degeneracy-TLS}
The probability $\bbP(Y_1>0, Y_2>0)$ is strictly positive.
\end{assumption}
\begin{assumption}[\textbf{Integrability}]\label{assu:MomentConditions-TLS}
All of the following expectations are finite:
\[
 \E[\|\bX_1\|_2^4],\;\E[\|\bX_2\|_2^4],\;\E[\alpha^2\|\Delta\bX\|_2^4],\;
 \E[\varepsilon_1^2\|\Delta\bX\|_2^4],\;\text{and}\;
 \E[\varepsilon_2^2\|\Delta\bX\|_2^4].
\]
\end{assumption} 
\begin{assumption}[\textbf{Continuity}]\label{assu:continuity tls} 
The conditional distribution of $(\varepsilon_1,\varepsilon_2)$ given 
$\bW$ is absolutely continuous with respect to the Lebesgue measure.
\end{assumption}
\begin{assumption}[\textbf{Exchangeability}]\label{assu:Exchangeability-TLS}
Conditional on $\bW$, 
$\varepsilon_1$ and $\varepsilon_2$ are exchangeable.
\end{assumption} 

\begin{assumption}[\textbf{Rank of Regressors}]\label{assu:RankRegressors-TLS}
There is no proper linear subspace of $\R^K$ containing the random variable
$\mathbf{1}\{\bbP\left(Y_1>0, Y_2>0\middle|\bX_1,\bX_2\right)>0\}\Delta\bX$ with probability one.
\end{assumption} 

Let $(\bw,\be)\mapsto f_{\bepsilon\mid\bw}(\be)$, mapping $\mathcal
W\times\mathcal E$ to $[0,\infty)$, be a version of the joint probability
density function (PDF) of the pair $\bepsilon = (\varepsilon_1,\varepsilon_2)$
conditional on $\bW = \bw$ that is measurable in $(\bw,\be)$ and is such that
$f_{\bepsilon\mid\bw}(e_1,e_2) = f_{\bepsilon\mid\bw}(e_2,e_1)$ for all
$\bw\in\mathcal W$ and $\be = (e_1,e_2)\in\mathcal E$. As we
show in Appendix \ref{sec: measurability appendix}, the function
$(\bw,\be)\mapsto f_{\bepsilon\mid\bw}(\be)$ does exist under Assumptions
\ref{assu:continuity tls} and \ref{assu:Exchangeability-TLS}. Also, let
$(\bw,\be)\mapsto F_{\bepsilon\mid\bw}(\be) =
\int_{-\infty}^{e_1}\int_{-\infty}^{e_2}f_{\bepsilon\mid\bw}(u_1,u_2)\dif u_2\dif u_1$
be the corresponding version of the joint cumulative distribution function (CDF)
of $(\varepsilon_1,\varepsilon_2)$ conditional on $\bW = \bw$, and let
$(\bw,e)\mapsto F_{\varepsilon\mid\bw}(e) = \lim_{u\to\infty}
F_{\bepsilon\mid\bw}(e,u)$ be the corresponding version of the common marginal
CDF of $\varepsilon_1$ and $\varepsilon_2$ conditional on $\bW = \bw$.


\subsection{Asymptotic Normality Result in Honor{\'e} (1992)} To state the TLS normality
result in \citet{honore_trimmed_1992}, introduce the $K\times K$ matrices
\[
  \bfV_0^{\tt{tls}}:=\E\left[\dot{m}_1^{\tt{tls}}(\Delta\bX^\top\btheta_0,\bY)^2\Delta\bX\Delta\bX^\top\right]
\]
and
\begin{equation}\label{eq:VarianceSandwichBread}
\bGamma_0^{\tt{tls}}:=\E\left[\mathbf{1}\cbr[0]{-Y_2<\Delta\bX^\top\btheta_0<Y_1}\Delta\bX\Delta\bX^\top\right].\footnote{In
\citet{honore_trimmed_1992}, these matrices are denoted $V_4$ and $\Gamma_4$,
respectively.}
\end{equation}
\citet[Theorem 2(iv)]{honore_trimmed_1992} states that if Assumptions
\ref{assu:Non-Degeneracy-TLS}--\ref{assu:RankRegressors-TLS} hold, the
expectations involved in defining the matrices $\bfV_0^{\tt{tls}}$ and
$\bGamma_0^{\tt{tls}}$ exist (in $\R^{K\times K}$), and both matrices are of full rank, then
\begin{equation}\label{eq:HonoreAsymptoticNormality-TLS}
\sqrt{n}\del[1]{\widehat\btheta^{\texttt{tls}}-\btheta_0}\rightsquigarrow\cN\left(\mathbf{0},(\bGamma_0^{\tt{tls}})^{-1}\bfV_0^{\tt{tls}}(\bGamma_0^{\tt{tls}})^{-1}\right)\;\text{in}\;\R^K,
\end{equation}
with $\cN(\boldsymbol{\mu},\boldsymbol{\Sigma})$ denoting the normal
distribution with mean $\boldsymbol{\mu}$ and covariance matrix
$\boldsymbol{\Sigma}$.

Among several steps, the proof of this result in
\citet{honore_trimmed_1992} involves:
\begin{enumerate}
  \item Arguing that the gradient of the expected loss
  $\btheta\mapsto\E[m^{\tt{tls}}(\Delta\bX^\top\btheta,\bY)]$ is equal to the
  expected value of the model score [see \eqref{eq:ExpectedScore-TLS} below].
  \item Establishing differentiability of the expected value of the model score
  at the true parameter  $\btheta_0$.
\end{enumerate}
In our notation, the second task amounts to arguing that the function
  $\bG:\R^K\to\R^K$ defined by
\begin{equation}\label{eq:ExpectedScore-TLS}
  \bG(\btheta):=\E\left[\frac{\partial}{\partial\btheta}m^{\tt{tls}}\del[0]{\Delta\bX^\top\btheta,\bY}\right]
  =\E\left[\dot{m}_1^\texttt{tls}\del[0]{\Delta\bX^\top\btheta,\bY}\Delta\bX\right],\quad\btheta\in\R^K,
\end{equation}
is differentiable at $\btheta=\btheta_0$. Its Jacobian $\nabla\bG(\btheta_0)\in\R^{K\times K}$ is then the Hessian of the expected loss.\footnote{In \citet{honore_trimmed_1992}, the function $\bG$
is denoted $G_4$.} To carry out this task, \citet{honore_trimmed_1992} invoked the LDCT to interchange the
order of differentiation and expectation, and the matrix $\bGamma_0^{\tt{tls}}$
is the result of that calculation. In the next subsection, however, we will show by example that
$\bGamma_0^{\tt{tls}}$ is not always equal to the Jacobian of $\bG$ at
$\btheta_0$.

\subsection{Counterexample}\label{sec:Counterexample} We here provide a DGP
which satisfies Assumptions
\ref{assu:Non-Degeneracy-TLS}--\ref{assu:RankRegressors-TLS}, for which the
function $\bG$ \emph{is} differentiable at $\btheta=\btheta_0$, but the
resulting Jacobian $\nabla\bG(\btheta_0)$ differs from $\bGamma_0^{\tt{tls}}$.
To this end, let $K=1$, $\alpha\equiv0$, $\btheta_0=0$, $\bX_1\equiv2$ and
$\bX_2\equiv1$, and let $\varepsilon_1$ and $\varepsilon_2$ be independent of
each other and standard normally distributed. Then the outcome variables
$Y_{\tau}=\max\{0,\varepsilon_{\tau}\},\tau\in\{1,2\}$, are independent and
identically distributed as standard normals censored from below at zero. Case by
case inspection shows that Assumptions
\ref{assu:Non-Degeneracy-TLS}--\ref{assu:RankRegressors-TLS} are satisfied and
that (the here scalar) $\bfV_0^{\tt{tls}}$ and $\bGamma_0^{\tt{tls}}$ are both
well-defined and of full rank (non-zero).

In this example DGP, the regressor difference $\Delta \bX = \bX_1 - \bX_2$ is
identically one, so the (here scalar-valued) function $\bG$ simplifies to
$G(\theta)=\E[\dot{m}^{\texttt{tls}}_1(\theta,\bY)],\theta\in\R$. It is then straightforward to check that the function $G$ is given by
\[
G(\theta) =
\begin{cases}
  \theta+\int_{0}^{-\theta}\Phi(u)\dif u, &\theta<0,\\
  0, &\theta=0,\\
  \theta-\int_{0}^{\theta}\Phi(u)\dif u, &\theta>0,
  \end{cases}
\]
where $\Phi$ is the standard normal CDF. For example, if $\theta < 0$, then
letting $\varphi$ denote the standard normal PDF,
\begin{align*}
\E[\dot{m}^{\texttt{tls}}_1(\theta,\bY)]
& = \E[-Y_1 + (Y_2 + \theta)\mathbf{1}\{Y_2 + \theta > 0\}] 
= \theta - \E[(Y_1 + \theta)\mathbf{1}\{Y_1 + \theta \leqslant 0\}] \\
& = \theta - \int_{0}^{-\theta}u\varphi(u)\dif u - \theta\Phi(-\theta)
= \theta + \int_0^{-\theta}\Phi(u)\dif u,
\end{align*}
where the first equality follows from \eqref{eq:TrimmedSqLossDerivative}, the
second from (the here) identical distributions of $Y_1$ and $Y_2$, the third
from noting that $Y_1$ is a standard normal random variable truncated from
below at zero, and the fourth from integration by parts.

The graph of $G$ is shown in Figure
\ref{fig:ExpectedLoss}.
\begin{figure}
\caption{Graph of $\theta\mapsto G(\theta)$ in the
counterexample\label{fig:ExpectedLoss}}
\centering
\includegraphics{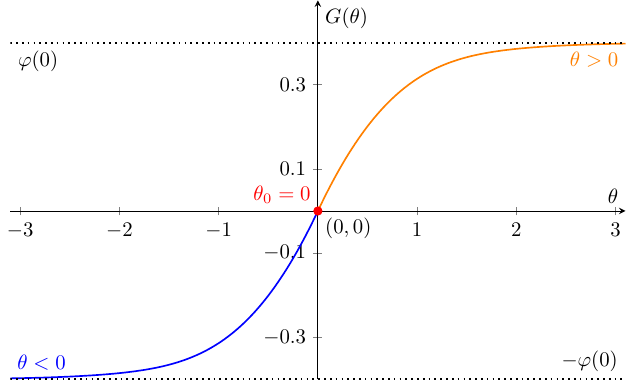}
\end{figure}
As indicated by the figure, and immediately follows analytically, $G$
\emph{is} differentiable at $\theta_0(=0)$, with derivative $\dot{G}(0)=1 -
\Phi(0)=\textstyle{\frac{1}{2}}$. However, since $Y_1$ and $Y_2$ are
independently distributed as standard normals censored from below at zero, (the
here scalar) $\bGamma_0^{\tt{tls}}$ equals
\[
\Gamma_0^{\tt{tls}}=\E[\mathbf{1}\{Y_1>0>-Y_2\}]=\bbP(Y_1>0)\cdot\bbP(Y_2>0)=\frac{1}{4},
\]
showing that $\dot{G}(0)\neq\Gamma_0^{\tt{tls}}$.

How does this discrepancy arise? When arguing differentiability of the function
$\bG(\btheta) =
\E\sbr[1]{\dot{m}_1^\texttt{tls}\del[0]{\Delta\bX^\top\btheta,\bY}\Delta\bX}$ at
$\btheta = \btheta_0$, \citet{honore_trimmed_1992} called upon the LDCT to
interchange the order of differentiation and expectation. However, as captured
by $N(\by)$ in \eqref{eq:TrimmedSqLossNonTwiceDiff}, the function
$\dot{m}^{\texttt{tls}}_1(\cdot,\by)$ fails to be differentiable at $-y_2$ and
$y_1$, unless both $y_1$ and $y_2$ are zero. In turn, in a model with censoring,
one may see exactly one outcome equal to zero with (strictly) positive
probability. Thus, whenever the inner product $\Delta\bX^\top\btheta_0$ places
mass at zero, the (random) function $t\mapsto \dot{m}_1^\texttt{tls}(t,\bY)$
fails to be differentiable at $t=\Delta\bX^{\top}\btheta_0$ with (strictly)
positive probability, invalidating the use of the LDCT.


\subsection{Extended Asymptotic Normality Result}\label{sec:DiffOfG} In this
subsection, we modify the argument in \citet{honore_trimmed_1992} to allow the
distribution of $\Delta\bX^\top\btheta_0$ to place mass at zero. We show that
Honor{\'e}'s assumptions, corresponding to our Assumptions
\ref{assu:Non-Degeneracy-TLS}--\ref{assu:RankRegressors-TLS}, imply
differentiability of $\bG$ at $\btheta=\btheta_0$. We also derive an explicit
expression for the Jacobian $\nabla\bG(\btheta_0)$, which coincides with the
Hessian of the expected loss. Our calculations further show that
$\nabla\bG(\btheta_0)$ generally differs from $\bGamma_0^{\tt{tls}}$, unless
mass at zero in the distribution of $\Delta\bX^\top\btheta_0$ is ruled out a
priori. By deriving equivalent representations of this Hessian, we obtain easily
interpretable necessary and sufficient conditions for its invertibility. We then
state an asymptotic normality result that allows $\Delta\bX^\top\btheta_0$ to
have mass at zero.
\begin{thm}[\textbf{Trimmed Least Squares Hessian
Existence}]\label{thm:JacobianExistence-TLS} Under Assumptions
\ref{assu:Non-Degeneracy-TLS}--\ref{assu:RankRegressors-TLS}, the function
$\bG:\R^K\to\R^K$ defined by \eqref{eq:ExpectedScore-TLS} is differentiable at
$\btheta=\btheta_0$ with the resulting matrix
$\bfJ^{\tt{tls}}_0:=\nabla\bG(\btheta_0)$ given by 
\begin{equation}\label{eq:Jacobian}
\bfJ^{\tt{tls}}_0 =\E\sbr[2]{\del[2]{1-F_{\varepsilon\mid\bW}\del[1]{-\alpha-\min\cbr[0]{\bX_1^\top\btheta_0,\bX_2^\top\btheta_0}}}\Delta\bX\Delta\bX^\top}.
\end{equation}
Equivalently,
\begin{align}
\bfJ^{\tt{tls}}_0 =\E\Big[\Big(
  &\mathbf{1}\{Y_1>0\}\del[1]{\mathbf{1}\{\Delta\bX^\top\btheta_0<0\}+\tfrac{1}{2}\mathbf{1}\{\Delta\bX^\top\btheta_0=0\}}\notag\\ 
  +&\mathbf{1}\{Y_2>0\}\del[1]{\mathbf{1}\{\Delta\bX^\top\btheta_0>0\}+\tfrac{1}{2}\mathbf{1}\{\Delta\bX^\top\btheta_0=0\}} \Big)\Delta\bX\Delta\bX^\top\Big].\label{eq:JacobianMidpoint}
\end{align}
\end{thm}
The expressions in \eqref{eq:Jacobian} and \eqref{eq:JacobianMidpoint}
make it clear that $(Y_1,\bX_1)$ and $(Y_2,\bX_2)$ enter the Hessian
$\bfJ^{\tt{tls}}_0$ in a symmetric manner, so that the (time) labeling is
irrelevant. The version in \eqref{eq:JacobianMidpoint} facilitates plug-in
estimation of $\bfJ^{\tt{tls}}_0$, which we cover in Section
\ref{sec:TLSAsymptoticNormalityAndVarianceEstimation}. In turn, the version in
\eqref{eq:Jacobian} facilitates comparison with Honor{\'e}'s
$\bGamma_0^{\tt{tls}}$. To this end, iterate expectations to write the latter as
\[
 \bGamma_0^{\tt{tls}}=\E\sbr[2]{\E\sbr[1]{\mathbf{1}\cbr[0]{-Y_2<\Delta\bX^\top\btheta_0<Y_1}\mid\bW}\Delta\bX\Delta\bX^\top}.
\]
We next expand on the inner expectation. Rewrite the inner
expectation indicator as
\[
 1-\mathbf{1}\{ Y_{1}\leqslant\Delta\bX^\top\btheta_0\} -\mathbf{1}\{ Y_{2}\leqslant-\Delta\bX^\top\btheta_0\} +\mathbf{1}\{ Y_{1}\leqslant \Delta\bX^\top\btheta_0\} \mathbf{1}\{ Y_{2}\leqslant-\Delta\bX^\top\btheta_0\}.
\]
Because $Y_{1}$ and $Y_{2}$ are non-negative, the only way both indicators
appearing in the previous display can be turned on is if
$\Delta\bX^\top\btheta_0=0$. Again using non-negativity, we deduce
\[
  \mathbf{1}\{ Y_{1}\leqslant \Delta\bX^\top\btheta_0\} \mathbf{1}\{ Y_{2}\leqslant-\Delta\bX^\top\btheta_0\}
  =\mathbf{1}\{ \Delta\bX^\top\btheta_0=0\} \mathbf{1}\{ Y_{1}=0\} \mathbf{1}\{ Y_{2}=0\},
\]
from which it follows that $\mathbf{1}\{-Y_2<\Delta\bX^\top\btheta_0<Y_1\}$ can be written as
\begin{align*}
  &1-\mathbf{1}\{ Y_{1}\leqslant\Delta\bX^\top\btheta_0\} -\mathbf{1}\{ Y_{2}\leqslant-\Delta\bX^\top\btheta_0\} + \mathbf{1}\{ \Delta\bX^\top\btheta_0=0\} \mathbf{1}\{ Y_{1}=0\} \mathbf{1}\{ Y_{2}=0\}.
\end{align*}
Taking conditional expectations, the inner expectation in $\bGamma_0^{\tt{tls}}$
works out as
\begin{align*}
  &\E\sbr[1]{\mathbf{1}\{-Y_2<\Delta\bX^\top\btheta_0<Y_1\}\mid\bW}\\
  &\qquad =\begin{cases}
  1-F_{\varepsilon\mid\bW}\del[0]{V_1}, & \Delta\bX^\top\btheta_0<0,\\
  1-F_{\varepsilon\mid\bW}\del[0]{V_1}-F_{\varepsilon\mid\bW}\del[0]{V_2}+F_{\bepsilon\mid\bW}\del[0]{V_1,V_2}, & \Delta\bX^\top\btheta_0=0,\\
  1-F_{\varepsilon\mid\bW}\del[0]{V_2}, & \Delta\bX^\top\btheta_0>0,
  \end{cases}
\end{align*}
where (due to space constraints) we introduced the shorthand notations
$V_1:=-(\alpha+\bX_1^\top\btheta_0)$ and
$V_2:=-(\alpha+\bX_2^\top\btheta_0)$.
Since $V_2-V_1=\bX_1^\top\btheta_0-\bX_2^\top\btheta_0=\Delta\bX^\top\btheta_0$,
we have $V_1\gtreqqless V_2$ if
$\Delta\bX^\top\btheta_0\lesseqqgtr0$. With this newly introduced shorthand
notation, the inner expectation in our Hessian \eqref{eq:Jacobian} can be
written as
$
1-F_{\varepsilon\mid\bW}\del[0]{\max\cbr{V_1,V_2}},
$
which agrees with the previous display whenever $\Delta\bX^\top\btheta_0\neq0$.
Hence, if this case can be ruled out
[i.e.,~$\bbP(\Delta\bX^\top\btheta_0=0)=0$], then Honor{\'e}'s
$\bGamma_0^{\tt{tls}}$ \emph{is} equal to $\bfJ^{\tt{tls}}$.

However, when $\Delta\bX^\top\btheta_0=0$, so that $V_1$ and $V_2$ take a common
value (here:~$V$), the inner expectations differ by
\begin{align*}
 &\sbr[1]{1-F_{\varepsilon\mid\bW}\del[0]{V}}-\E\sbr[1]{\mathbf{1}\{ -Y_{2}<\Delta\bX^{\top}\btheta_{0}<Y_{1}\} \mid\bW}\\
 &\qquad =F_{\varepsilon\mid\bW}\del[0]{V}-F_{\bepsilon\mid\bW}\del[0]{V,V}=\lim_{u\to\infty}F_{\bepsilon\mid\bW}\del[0]{V,u}-F_{\bepsilon\mid\bW}\del[0]{V,V},
\end{align*}
which is non-negative by monotone increasingness of CDFs. In the Section
\ref{sec:Counterexample} counterexample, we have both $V_{1}\equiv0$ and
$V_{2}\equiv0$, and the (conditional) joint CDF factors into the product of the
(standard normal) marginals,
$F_{\bepsilon\mid\bW}(u_{1},u_{2})=\Phi(u_{1})\Phi(u_{2})$. Since there
$\bbP(\Delta\bX^\top\btheta_0=0)=1$, the difference in the previous display
captures the only relevant case. The difference is then
$\Phi(0)-\Phi(0)\Phi(0)=\frac{1}{4}$, which is precisely the previously
demonstrated discrepancy.

\begin{rem}
[\textbf{Alternative Hessian Expressions}]\label{rem:JacobianAlt1} As
established in the proof of Theorem \ref{thm:JacobianExistence-TLS}, the Hessian
$\bfJ^{\tt{tls}}_0$ can also be expressed using either of the following
expressions:
\begin{align}
  \bfJ^{\tt{tls}}_0&=\E\sbr[2]{\del[2]{\mathbf{1}\{Y_1>0\}\mathbf{1}\{\Delta\bX^\top\btheta_0\leqslant0\}+\mathbf{1}\{Y_2>0\}\mathbf{1}\{\Delta\bX^\top\btheta_0>0\}}\Delta\bX\Delta\bX^\top},\label{eq:JacobianAlt1} \\ 
  \bfJ^{\tt{tls}}_0&=\E\sbr[2]{\del[2]{\mathbf{1}\{Y_1>0\}\mathbf{1}\{\Delta\bX^\top\btheta_0<0\}+\mathbf{1}\{Y_2>0\}\mathbf{1}\{\Delta\bX^\top\btheta_0\geqslant0\}}\Delta\bX\Delta\bX^\top}\label{eq:JacobianAlt2}.
\end{align}
The version in \eqref{eq:JacobianMidpoint} is the average of the
\eqref{eq:JacobianAlt1} and \eqref{eq:JacobianAlt2} right-hand
sides.\hfill$\diamondsuit$
\end{rem}

\begin{rem}[\textbf{Relaxing
Exchangeability}]\label{rem:RelaxingExchangeability-TLS} Inspecting the proof of
Theorem \ref{thm:JacobianExistence-TLS} (specifically, the proof of Lemma
\ref{lem:ExpectedTrimmedSqLossDerivative}), we see that it does not actually
require full conditional exchangeability of $\varepsilon_1$ and $\varepsilon_2$
given $\bW$. For the conclusions of Theorem \ref{thm:JacobianExistence-TLS} to
hold, Assumption \ref{assu:Exchangeability-TLS} can be replaced with the weaker
condition of conditional \emph{stationarity}: conditional on $\bW$,
$\varepsilon_1$ and $\varepsilon_2$ are identically distributed. Some of the
generalizations made in \citet[Section 7.1]{arellano2001panel} for Tobit-type
models with fixed effects similarly only require the weaker assumption of
conditional stationarity. When their $\psi(\cdot)$ and $\xi(\cdot)$ are both the
identity mapping on the real line, the resulting estimator is precisely the
\citet{honore_trimmed_1992} TLS estimator studied here, but motivated
differently. \citet[p.~3274]{arellano2001panel} observe that ``[i]t follows from
standard results about extremum estimators that the resulting estimator will be
consistent and $\sqrt n$ asymptotically normal'', but do not comment on the form
of the limit variance components. The Hessian expressions in Theorem
\ref{thm:JacobianExistence-TLS} and Remark \ref{rem:JacobianAlt1} are therefore
also relevant for this \citet{arellano2001panel} generalization of the TLS
estimator.\hfill$\diamondsuit$
\end{rem}

To state the asymptotic normality result, we need to make sure that the Hessian
$\bfJ^{\tt{tls}}_0$ is invertible. To this end, we will impose the following
condition:

\begin{assumption}[\textbf{Trimmed Least Squares Score Jacobian
Invertibility}]\label{assu:JacobianInvertibility-TLS} There is no proper
linear subspace of $\R^K$ containing the random variable
$(\mathbf{1}\{Y_1>0\}\mathbf{1}\{\Delta\bX^\top\btheta_0\leqslant0\}+\mathbf{1}\{Y_2>0\}\mathbf{1}\{\Delta\bX^\top\btheta_0>0\})\Delta\bX$
with probability one.
\end{assumption}

This assumption is easy to interpret and implies via
\eqref{eq:JacobianAlt1} that the matrix $\bfJ^{\tt{tls}}_0$ is invertible.
Alternatively, we could state an assumption based on the expression for
$\bfJ^{\tt{tls}}_0$ in \eqref{eq:JacobianAlt2}. Strictly speaking, Assumption
\ref{assu:JacobianInvertibility-TLS} does not appear in
\citet{honore_trimmed_1992} but the only purpose of this assumption is to ensure
the invertibility of the Hessian $\bfJ^{\tt{tls}}_0$, and \citet[Theorem
2(iv)]{honore_trimmed_1992} did impose invertibility of the corresponding
Hessian $\bGamma_0^{\tt{tls}}$.

Using our new understanding of the Hessian $\bfJ^{\tt{tls}}_0$, we next state an
extended asymptotic normality result for the TLS estimator that does not require
the inner product $\Delta\bX^\top\btheta_0$ to have no mass at zero.
\begin{thm}[\textbf{Asymptotic Normality of Trimmed Least
Squares}]\label{thm:AsymptoticNormality-TLS} Let Assumptions
\ref{assu:Non-Degeneracy-TLS}--\ref{assu:JacobianInvertibility-TLS} hold, and
suppose that the expectations involved in defining the matrix
$\bfV_0^{\tt{tls}}$ exist (in $\R^{K\times K}$), and that this matrix is of full
rank. Then the TLS estimator satisfies
\[
\sqrt{n}\del[1]{\widehat\btheta^{\tt{tls}}-\btheta_0}\rightsquigarrow\mathcal{N}\left(\mathbf{0},(\bfJ_0^{\tt{tls}})^{-1}\bfV_0^{\tt{tls}}(\bfJ_0^{\tt{tls}})^{-1}\right)\;\text{in}\;\R^K.
\]
\end{thm}
The asymptotic normality result in Theorem \ref{thm:AsymptoticNormality-TLS}
differs from \citet[Theorem 2(iv)]{honore_trimmed_1992} in the sense that it
replaces the matrix $\bGamma_0^{\tt{tls}}$ in
\eqref{eq:HonoreAsymptoticNormality-TLS} by $\bfJ_0^{\tt{tls}}$. Whenever
$\Delta\bX^\top\btheta_0\neq 0$ with probability one, the two matrices coincide.
However, the matrices are in general different if $\Delta\bX^\top\btheta_0 = 0$
with (strictly) positive probability.

\subsection{Asymptotic Variance
Estimation}\label{sec:TLSAsymptoticNormalityAndVarianceEstimation} For the
asymptotic normality to yield a practical approximation, we need to consistently
estimate the asymptotic variance components. For $\bfV_0^{\tt{tls}}$, we use the
plug-in estimator from \citet{honore_trimmed_1992}:
\[
  \widehat{\bfV}^{\tt{tls}}:=\frac{1}{n}\sum_{i=1}^n \dot{m}_1^{\tt{tls}}(\Delta\bX_i^\top\widehat\btheta^{\tt{tls}},\bY_i)^2\Delta\bX_i\Delta\bX_i^\top.\footnote{In \citet{honore_trimmed_1992}, this estimator is denoted $\hat{V}_4$.}
\]
For the Hessian $\bfJ_0^{\tt{tls}}$, we apply the analogy principle to
\eqref{eq:JacobianMidpoint} to arrive at
\begin{align}
  \widehat\bfJ^{\tt{tls}}=\frac{1}{n}\sum_{i=1}^n\bigg(
    &\mathbf{1}\{Y_{i1}>0\}\del[2]{\mathbf{1}\{\Delta\bX_i^\top\widehat\btheta^{\tt{tls}}<0\}+\frac{1}{2}\mathbf{1}\{\Delta\bX_i^\top\widehat\btheta^{\tt{tls}}=0\}}\notag\\ 
    +&\mathbf{1}\{Y_{i2}>0\}\del[2]{\mathbf{1}\{\Delta\bX_i^\top\widehat\btheta^{\tt{tls}}>0\}+\frac{1}{2}\mathbf{1}\{\Delta\bX_i^\top\widehat\btheta^{\tt{tls}}=0\}} \bigg)\Delta\bX_i\Delta\bX_i^\top.\label{eq:JacobianMidpointEstimator}
\end{align}
Alternatively, in order to estimate the Hessian $\bfJ_0^{\tt{tls}}$, we could
use the equivalent expressions in \eqref{eq:JacobianAlt1} and
\eqref{eq:JacobianAlt2}.

These estimators are (strongly) consistent under the assumptions of Theorem
\ref{thm:AsymptoticNormality-TLS}:
\begin{thm}[\textbf{Plug-in Variance Estimator Consistency for
TLS}]\label{thm:VarianceConsistency-TLS} Let the assumptions of Theorem
\ref{thm:AsymptoticNormality-TLS} hold. Then
$\widehat\bfV^{\tt{tls}}\to_{\mathrm{a.s.}}\bfV_0^{\tt{tls}}$ and
$\widehat\bfJ^{\tt{tls}}\to_{\mathrm{a.s.}}\bfJ_0^{\tt{tls}}$.
\end{thm}
Theorem \ref{thm:VarianceConsistency-TLS} implies that the limit
variance $(\bfJ_0^{\tt{tls}})^{-1}\bfV_0^{\tt{tls}}(\bfJ_0^{\tt{tls}})^{-1}$ is
(strongly) consistently estimated by
$(\widehat\bfJ^{\tt{tls}})^{-1}\widehat\bfV^{\tt{tls}}(\widehat\bfJ^{\tt{tls}})^{-1}$,
which facilitates hypothesis testing and the construction of confidence
intervals.

\begin{rem}[\textbf{Comparison with \citet{honore2000estimation}}]
\citet[Section 2.1]{honore2000estimation} discuss estimation of the censored
regression model with fixed effects, allowing the number of time periods $T_i$
to exceed two and/or be individual specific (unbalanced panel data). While
\emph{ibid.}~(Section 2.1) covers a whole class of estimators, in the special
case of their $\xi(\cdot)$ being the identity mapping on the real line, the
estimator in their (6) becomes a TLS estimator based on all pairs of
time periods. In the balanced case $(T_i=T\geqslant2)$, their estimator is (in
our notation) any minimizer of
\[
  \btheta\mapsto \frac{1}{n}\sum_{i=1}^n \sum_{1\leqslant \tau < \tau' \leqslant T} 
 {m}^{\tt{tls}}\del[1]{(\bX_{i\tau}-\bX_{i\tau'})^\top\btheta,(Y_{i\tau},Y_{i\tau'})},
\]
which naturally generalizes the two-period TLS estimator studied in this paper.
Appropriately extending our assumptions to their multi-period case, one can
establish ($\sqrt n$-)consistency and asymptotic normality of their multi-period
TLS estimator and reuse the argument leading to Theorem
\ref{thm:JacobianExistence-TLS} to show that the Hessian of the expected loss in
the multi-period setting is given by an expression that naturally generalizes
\eqref{eq:JacobianMidpoint} to accommodate more than one pair of time periods.
One can then construct a (strongly) consistent estimator of this Hessian by
applying the analogy principle, as we did in
\eqref{eq:JacobianMidpointEstimator} for the $T=2$ case. \hfill$\diamondsuit$
\end{rem}

\begin{rem}[\textbf{A Plug-In Estimator for
\citet{honore_pairwise_1994}}]\label{rem:ComparisonWithHonorePowellJacobian}
Consider the {\em cross-sectional} censored regression model $Y =
\max\{0,Y^*\}$, where $Y^* = \bX^{\top}\btheta_0 + \varepsilon$ and
$\varepsilon$ is independent of $\bX$. This model was studied in
\citet{honore_pairwise_1994}, who developed, among other things, the TLS
estimator of the vector of parameters $\btheta_0$ in this model, proved its
($\sqrt{n}$-)consistency, and derived the corresponding asymptotic normality result. On
\emph{ibid.}~(p.~260), they gave an expression for the Hessian of the expected
loss, which enters the asymptotic variance formula. In our notation, their
Hessian takes the following form:
\begin{align}
\E\sbr[2]{\del[2]{1-F_{\varepsilon}\del[1]{-\min\cbr[0]{\bX_1^\top\btheta_0,\bX_2^\top\btheta_0}}}\Delta\bX\Delta\bX^\top},\label{eq: hessian cross section}
\end{align}
where the pairs $(Y_1,\bX_1)$ and $(Y_2,\bX_2)$ represent independent units of
observation and $F_{\varepsilon}$ is the CDF of $\varepsilon$. On \emph{ibid.}
(p.~261), they also proposed a generic numerical derivative estimator of this
Hessian. Due to numerical differentiation, however, their estimator introduces a
tuning parameter through the choice of the stepsize. 

Our point in this remark is to show that one can actually estimate the Hessian
in \eqref{eq: hessian cross section} without introducing extra tuning
parameters. To see that, let  $\bZ_i := (Y_i,\bX_i)$, $i\in\{1,2,\dots,n\}$, be
a random sample from the distribution of $\bZ:=(Y,\bX)$, and observe that our
panel-data Hessian $\bfJ^{\tt{tls}}_0$ in \eqref{eq:Jacobian} reduces to the
cross-sectional Hessian in \eqref{eq: hessian cross section} upon setting
$\alpha\equiv0$ and imposing independence of $\varepsilon_1$ and $\varepsilon_2$
from $\bW$ in the former expression. Theorem \ref{thm:JacobianExistence-TLS}
thus reveals, via \eqref{eq:JacobianMidpoint}, that the cross-sectional Hessian
in \eqref{eq: hessian cross section} is equal to $\E[h(\bZ_1,\bZ_2;\btheta_0)]$,
where the function $h:\R^{1+K}\times\R^{1+K}\times\R^K\to\R^{K\times K}$ is
given by
\begin{align*}
h(\bz_1,\bz_2;\btheta)
  :=\Big(&\mathbf{1}\{y_1>0\}\del[1]{\mathbf{1}\{\Delta\bx^\top\btheta<0\}+\tfrac{1}{2}\mathbf{1}\{\Delta\bx^\top\btheta=0\}}\\ 
  +&\mathbf{1}\{y_2>0\}\del[1]{\mathbf{1}\{\Delta\bx^\top\btheta>0\}+\tfrac{1}{2}\mathbf{1}\{\Delta\bx^\top\btheta=0\}}\Big)\Delta\bx\Delta\bx^\top,
\end{align*}
and we denote $\bz_1:=(y_1,\bx_1)$, $\bz_2:=(y_2,\bx_2)$, and $\Delta\bx :=
\bx_1 - \bx_2$. Let $\widehat{\btheta}$ be the \citet{honore_pairwise_1994} TLS
estimator, i.e.~their estimator with
$\Xi\left(\cdot\right)=\textstyle{\frac{1}{2}}\left(\cdot\right)^2$. Then a
natural estimator of the cross-sectional Hessian in \eqref{eq: hessian cross
section} is of the plug-in form
\begin{align}
\frac{1}{n(n-1)}\sum_{1\leqslant i\neq j\leqslant n}h(\bZ_i,\bZ_j;\widehat\btheta),\label{eq:HonorePowellPlugInEstimator}
\end{align}
which is an approximate second-order U-statistic. Note that this estimator
involves no tuning parameter, and is therefore free of the stepsize problem
arising due to numerical differentiation. The consistency of this plug-in
estimator can be established using an argument paralleling the one used in the
proof of Theorem \ref{thm:VarianceConsistency-TLS}. The main difference is that, instead of establishing a uniform law of large
numbers involving \emph{simple} averages (or, equivalently, an empirical
process), to accommodate the expression in
\eqref{eq:HonorePowellPlugInEstimator}, one must now work with
\emph{generalized} averages (leading to a second-order U-process). To this end,
we refer the reader to \citet{nolan1987_u_processes_rates}.\hfill$\diamondsuit$
\end{rem}


\begin{rem}[\textbf{Comparison with
\citet{powell1984least,powell1986symmetrically}}]\label{rem:ComparisonWithPowell} 

As noted in \citet{honore_trimmed_1992}, the rationale behind the TLS and TLAD
estimators considered in this paper can be viewed as a bivariate generalization
of the idea behind the \citet{powell1986symmetrically} symmetrically trimmed LS
estimators for censored Tobit models without fixed effects. The approach adopted
in \citet{powell1986symmetrically} is, in turn, closely related to the censored
LAD (CLAD) estimator proposed in \citet{powell1984least}---a key paper in the
censored regression literature.
In our notation, \citet{powell1984least} models the conditional median of a
non-negative response variable $Y$ given covariates $\bX$ as
$\mathrm{Med}(Y|\bX)=\max\{0,\bX^\top\btheta_0\}$. To ensure asymptotic
normality of the CLAD estimator, \citet[Assumption R.2]{powell1984least} rules
out regressors that are orthogonal to $\btheta_0$ with positive probability. The
possibility of such regressors is precisely the cause of difficulty in our
analysis, cf.~the discussion following Theorem
\ref{thm:JacobianExistence-TLS}.\hfill$\diamondsuit$
\end{rem}

\subsection{Consistency of the Honor{\'e} (1992) Hessian
Estimator}\label{sec:ConsistencyHonoreHessianEstimator-TLS} In
\citet{honore_trimmed_1992} the matrix $\bGamma_0^{\tt{tls}}$ in
\eqref{eq:VarianceSandwichBread} is estimated by its empirical analogue
\begin{equation}\label{eq:Honore1992HessianEstimator}
  \widehat\bGamma^{\tt{H92}}:=\frac{1}{n}
  \sum_{i=1}^{n}\bone\{-Y_{i2}<\Delta\bX_{i}^{\top}\widehat{\btheta}^{\tt{tls}}<Y_{i1}\}
  \Delta\bX_{i}\Delta\bX_{i}^{\top}.\footnote{In
\cite{honore_trimmed_1992}, this estimator is denoted by $\widehat{\Gamma}_4$.}
\end{equation}
The TLS sandwich variance consistency results in \citet[Theorem
3]{honore_trimmed_1992} involve showing that $\widehat\bGamma^{\tt{H92}}$ is a
consistent estimator of $\bGamma_0^{\tt{tls}}$. As we have shown by example in
Section \ref{sec:Counterexample}, $\bGamma_0^{\tt{tls}}$ is not equal to the
Hessian $\bfJ_0^{\tt{tls}}$ of the expected loss in general. As
$\widehat\bGamma^{\tt{H92}}$ is targeting $\bGamma_0^{\tt{tls}}$, one would
expect that $\widehat\bGamma^{\tt{H92}}$ is not a consistent estimator of the
TLS Hessian $\bfJ_0^{\tt{tls}}$ in general. Surprisingly, as we establish below,
$\widehat\bGamma^{\tt{H92}}$ nevertheless converges to the TLS Hessian
$\bfJ_0^{\tt{tls}}$.
\begin{thm}[\textbf{Consistency of the \citet{honore_trimmed_1992} Hessian
Estimator for TLS}]\label{thm:ConsistencyHonoreHessianEstimator-TLS} Let the
assumptions of Theorem \ref{thm:AsymptoticNormality-TLS} hold. Then
$\widehat\bGamma^{\tt{H92}}=\bfJ_0^{\tt{tls}}+o_{L^1}(1)+o_{\mathrm{a.s.}}(1)=\bfJ_0^{\tt{tls}}+o_{\bbP}(1)$.
\end{thm}
This counterintuitive result can be explained as follows. Decompose
$\widehat\bGamma^{\tt{H92}}$ into two parts:
\begin{align*}
\widehat{\mathbf{\Gamma}}^{\tt{H92}} & =\frac{1}{n}\sum_{i=1}^{n}\left(\bone\{-Y_{i2}<\Delta\bX_{i}^{\top}\widehat{\btheta}^{\tt{tls}}<0\}+\bone\{0<\Delta\bX_{i}^{\top}\widehat{\btheta}^{\tt{tls}}<Y_{i1}\}\right)\Delta\bX_{i}\Delta\bX_{i}^{\top}\tag{=:\ensuremath{\widehat{\bfL}^{\tt{H92}}}}\\
 & \quad+\frac{1}{n}\sum_{i=1}^{n}\bone\{Y_{i1}>0,Y_{i2}>0,\Delta\bX_{i}^{\top}\widehat{\btheta}^{\tt{tls}}=0\}\Delta\bX_{i}\Delta\bX_{i}^{\top}\tag{=:\ensuremath{\widehat{\bfR}^{\tt{H92}}}}.
\end{align*}
Although $\widehat\bfL^{\tt{H92}}$ and $\widehat\bfJ^{\tt{tls}}$ are not
identical term by term, the leading component $\widehat\bfL^{\tt{H92}}$ is
closely related to our Hessian estimator in
\eqref{eq:JacobianMidpointEstimator}, which is the empirical analogue of
$\bfJ_0^{\tt{tls}}$. The proof of Theorem
\ref{thm:ConsistencyHonoreHessianEstimator-TLS} shows that
$\widehat\bfL^{\tt{H92}}$ converges to $\bfJ_0^{\tt{tls}}$, while the remainder
term $\widehat\bfR^{\tt{H92}}$ converges to the zero matrix. Key to these
findings is our extended asymptotic normality result in Theorem
\ref{thm:AsymptoticNormality-TLS}. Specifically, even if
$\Delta\bX^\top\btheta_0$ has mass at zero, the absolutely continuous limiting
distribution of $\sqrt n(\widehat\btheta^{\tt{tls}}-\btheta_0)$ implies that,
for any fixed $\bx\in\R^K\backslash\{\bzero\}$,
$\bbP(\bx^\top\widehat\btheta^{\tt{tls}}=0)\to0$. We stress that Theorem
\ref{thm:ConsistencyHonoreHessianEstimator-TLS} is \emph{not} a simple
consequence of consistency of $\widehat\btheta^{\tt{tls}}$, established in
\citet[Theorem 1(iv)]{honore_trimmed_1992}. Indeed, even with a uniform law of
large numbers allowing us to replace sample averages in
$\widehat{\bGamma}^{\tt{H92}}$ by their population counterparts (see the proof
of Theorem \ref{thm:ConsistencyHonoreHessianEstimator-TLS}), the resulting map
$\btheta\mapsto\bGamma(\btheta)$ need not be continuous at $\btheta_0$.
Consequently, a continuous mapping argument need not apply. In the
counterexample of Section \ref{sec:Counterexample}, one finds
$\Gamma(\theta)=[1-\Phi(|\theta|)]\bone\{\theta\neq0\}+\tfrac{1}{4}\bone\{\theta=0\}$,
which has a jump discontinuity at $\theta=0=\theta_0$. The jump size
$(\tfrac{1}{4})$ is exactly the difference between $\bfJ_0^{\tt{tls}}$ and
$\bGamma_0^{\tt{tls}}$ in that example.

A by-product of Theorem \ref{thm:ConsistencyHonoreHessianEstimator-TLS}
is that, under our maintained assumptions, the TLS Hessian consistency
statement in \citet[Theorem 3]{honore_trimmed_1992} does not hold
without additional regularity conditions ensuring continuity at $\btheta_0$.

\section{Trimmed Least Absolute Deviations}\label{sec:TLAD} In this section, we
focus on the TLAD estimator. We abbreviate the trimmed absolute loss, $m_{\Xi}$
in \eqref{eq:TrimmedGenericLoss} with $\Xi=\left|\cdot\right|$, by
$m^{\texttt{tlad}}$, which takes the form
\begin{equation}\label{eq:TrimmedAbsLoss}
m^{\texttt{tlad}}(t,\by)=
\begin{cases}
  |y_{1}|-\left(t+y_{2}\right)\mathrm{sgn}(y_{1}), & t\leqslant-y_{2},\\
  |y_{1}-y_{2}-t|, & t\in\left(-y_{2},y_{1}\right),\\
  |-y_{2}|-\left(t-y_{1}\right)\mathrm{sgn}(-y_{2}), & t\geqslant y_{1}.
\end{cases}
\end{equation}

\subsection{Assumptions for Trimmed Least Absolute
Deviations}\label{sec:Assumptions-TLAD} For notational convenience, abbreviate
the TLAD estimator $\widehat\btheta_\Xi$ in \eqref{eq:TrimmedEstimator} with
$\Xi=\left|\cdot\right|$ by $\widehat{\btheta}^{\texttt{tlad}}$. Also, let
$\bW$, $\mathcal W$, $\bw$, $\bepsilon$, $\mathcal E$, and $\be$ be the same as
in Section \ref{sec:TLS}. Consider the following
assumptions.\footnote{Assumptions
\ref{assu:Non-Degeneracy-TLAD}--\ref{assu:RankRegressors-TLAD} are from
Assumptions S.2, M.2, E.1, E.3, E.4 and R.1, respectively, in
\citet{honore_trimmed_1992}.}
\begin{assumption}[\textbf{Non-Degeneracy}]\label{assu:Non-Degeneracy-TLAD} The
  probability $\bbP(Y_1>0, Y_2>0)$ is strictly positive. \end{assumption} 
\begin{assumption}[\textbf{Integrability}]\label{assu:MomentConditions-TLAD} All
  of the following expectations are finite:
  \[
    \E[\|\bX_1\|_2^2],\;\E[\|\bX_2\|_2^2],\;\E[\|\alpha\Delta\bX\|_2],\;
    \E[\|\varepsilon_1\Delta\bX\|_2]\quad\text{and}\quad
    \E[\|\varepsilon_2\Delta\bX\|_2].
  \]
\end{assumption} 
\begin{assumption}[\textbf{Continuity}]\label{assu:AbsoluteContinuity-TLAD} The
  conditional distribution of $(\varepsilon_1,\varepsilon_2)$ given $\bW$ is
  absolutely continuous with respect to the Lebesgue measure. \end{assumption}
\begin{assumption}[\textbf{Exchangeability}]\label{assu:Exchangeability-TLAD}
  Conditional on $\bW$, $\varepsilon_{1}$ and $\varepsilon_{2}$ are
  exchangeable. \end{assumption}
Since Assumptions \ref{assu:AbsoluteContinuity-TLAD} and
\ref{assu:Exchangeability-TLAD} are the same as Assumptions \ref{assu:continuity
tls} and \ref{assu:Exchangeability-TLS}, following the reasoning in Section
\ref{sec:Assumptions-TLS}, there exists a function $(\bw,\be)\mapsto
f_{\bepsilon\mid\bw}(\be)$, mapping $\mathcal W\times\mathcal E$ to
$[0,\infty)$, that is a version of the PDF of the pair $\bepsilon =
(\varepsilon_1,\varepsilon_2)$ conditional on $\bW = \bw$, which is measurable
in $(\bw,\be)$, and is such that $f_{\bepsilon\mid\bw}(e_1,e_2) =
f_{\bepsilon\mid\bw}(e_2,e_1)$ for all $\bw\in\mathcal W$ and $\be =
(e_1,e_2)\in\mathcal E$. Also, let $(\bw,e)\mapsto f_{\varepsilon\mid\bw}(e) =
\int_{\mathbb R}f_{\bepsilon\mid\bw}(e,u)\dif u$ be the corresponding version of
the common marginal PDF of $\varepsilon_1$ and $\varepsilon_2$ conditional on
$\bW = \bw$ and let $(\bw,e)\mapsto f_{\varepsilon_1 - \varepsilon_2\mid\bw}(e)
= \int_{\mathbb R} f_{\bepsilon\mid\bw}(u+e,u)\dif u$ be the corresponding
version of the PDF of the difference $\varepsilon_1 - \varepsilon_2$ conditional
on $\bW = \bw$.

\begin{assumption}[\textbf{Regularity}]\label{assu:Regularity-TLAD} There is a
constant $C\in(0,\infty)$ such that
$\sup_{e\in\R}f_{\varepsilon_{1}-\varepsilon_{2}\mid\bW}(e)\leqslant C$ and
$\sup_{e\in\R}f_{\varepsilon\mid\bW}(e)\leqslant C$ with probability one.
\end{assumption} 
\begin{assumption}[\textbf{Rank of Regressors}]\label{assu:RankRegressors-TLAD}
There is no proper linear subspace of $\R^K$ containing the random variable
$\mathbf{1}\{\bbP\left(Y_1>0, Y_2>0\middle|\bX_1,\bX_2\right)>0\}\Delta\bX$ with
probability one. \end{assumption} 

\begin{assumption}[\textbf{Continuity, II}]\label{assu:extra continuity} The
functions $\be\mapsto f_{\bepsilon\mid \bW}(\be)$, $e\mapsto
f_{\varepsilon\mid\bW}(e)$, and $e\mapsto f_{\varepsilon_1 -
\varepsilon_2\mid\bW}(e)$ are continuous with probability one. 
\end{assumption}

Assumptions
\ref{assu:Non-Degeneracy-TLAD}--\ref{assu:RankRegressors-TLAD} are the same as
the corresponding assumptions in \citet{honore_trimmed_1992}. Assumption
\ref{assu:extra continuity}, however, is not present in
\citet{honore_trimmed_1992}. We consider this assumption because the asymptotic
normality result in \citet{honore_trimmed_1992} may not hold without it. In
particular, in Section \ref{sec:Counterexample 2} below we provide a DGP that
satisfies Assumptions
\ref{assu:Non-Degeneracy-TLAD}--\ref{assu:RankRegressors-TLAD} and is such that
the Hessian of the expected loss [see \eqref{eq:PopulationLossFunction-TLAD}
below], which appears in the asymptotic variance formula in
\citet{honore_trimmed_1992}, does not exist. We note also that Assumption
\ref{assu:extra continuity} may be stronger than necessary. For example, it
seems possible to obtain the asymptotic normality result assuming continuity of
the functions in Assumption \ref{assu:extra continuity} only on their respective
supports but we opt for a stronger than necessary conditions for clarity of the
argument.

Note also that Assumptions \ref{assu:Non-Degeneracy-TLAD},
\ref{assu:AbsoluteContinuity-TLAD}, \ref{assu:Exchangeability-TLAD}, and
\ref{assu:RankRegressors-TLAD} are the same as the corresponding assumptions for
the TLS estimator (Assumptions \ref{assu:Non-Degeneracy-TLS},
\ref{assu:continuity tls}, \ref{assu:Exchangeability-TLS}, and
\ref{assu:RankRegressors-TLS}, respectively). The TLAD moment conditions in
Assumption \ref{assu:MomentConditions-TLAD} are weaker than those for the TLS
estimator (Assumption \ref{assu:MomentConditions-TLS}), which reflects the fact
that the trimmed absolute loss is less sensitive to outliers than the trimmed
square loss. Assumptions \ref{assu:Regularity-TLAD} and \ref{assu:extra
continuity}, requiring certain PDFs to be bounded and continuous, do not have
analogs in the case of the TLS estimator.

\subsection{Asymptotic Normality in Honor{\'e} (1992)} To state the TLAD normality
result in \citet{honore_trimmed_1992}, introduce the $K\times K$ matrices\footnote{In
\citet{honore_trimmed_1992}, these matrices are denoted $V_3$ and $\Gamma_3$,
respectively. Honor{\'e}'s $\Gamma_3$ is stated in terms of probabilities and
densities of the censored $Y_1$ and $Y_2$, but it is clear from the underlying
proof that he means the latent variables $Y_1^\ast$ and $Y_2^\ast$,
respectively.}
\begin{equation}\label{eq:VarianceSandwichMeat-TLAD}
\left.\begin{aligned}
\bfV_0^{\tt{tlad}}:=
\E\Big[\Big(&\bone\cbr[1]{Y_1>0}\bone\cbr[1]{Y_1-Y_2>\Delta\bX^\top\btheta_0}\\
+&\bone\cbr[1]{Y_2>0}\bone\cbr[1]{Y_1-Y_2<\Delta\bX^\top\btheta_0}\Big)\Delta\bX\Delta\bX^\top\Big]
\end{aligned}\right\}
\end{equation}
and
\begin{equation}\label{eq:VarianceSandwichBread-TLAD}
  \left.\begin{aligned}
  \bGamma_0^{\tt{tlad}}
    &:=\E\Big[\Big(2f_{Y_1^\ast-Y_2^\ast\mid\bW,Y_1^\ast>0,Y_2^\ast>0}\del[1]{\Delta\bX^\top\btheta_0} \\ 
    &\qquad+\mathbf{1}\cbr[1]{\Delta\bX^\top\btheta_0\geqslant0}\bbP\left(Y_2^\ast\leqslant0\middle|\bW\right)f_{Y^\ast_1\mid\bW,Y_2^\ast\leqslant0}\del[1]{\Delta\bX^\top\btheta_0} \\ 
    &\qquad+\mathbf{1}\cbr[1]{\Delta\bX^\top\btheta_0<0}\bbP\left(Y_1^\ast\leqslant0\middle|\bW\right)f_{Y^\ast_2\mid\bW,Y_1^\ast\leqslant0}\del[1]{-\Delta\bX^\top\btheta_0}\Big)\Delta\bX\Delta\bX^\top\Big],
  \end{aligned}\right\}
\end{equation}
where $(\bw,e) \mapsto
f_{Y_1^\ast-Y_2^\ast\mid\bw,Y_1^\ast>0,Y_2^\ast>0}\left(e\right)$ is the
conditional PDF of $Y_1^\ast-Y_2^\ast$ given $\bW=\bw$ and $\{Y_1^\ast>0\}\cap
\{Y_2^\ast>0\}$, $(\bw,e) \mapsto
f_{Y^\ast_1\mid\bw,Y_2^\ast\leqslant0}\left(e\right)$ is the conditional PDF of
$Y_1^\ast$ given $\bW=\bw$ and $Y_2^\ast\leqslant0$, and $(\bw,e)\mapsto
f_{Y^\ast_2\mid\bw,Y_1^\ast\leqslant0}\left(e\right)$ is the conditional PDF of
$Y_2^\ast$ given $\bW=\bw$ and $Y_1^\ast\leqslant0$.\footnote{Here and below, we
assume that the versions of all conditional PDFs and CDFs including latent
outcomes $Y_1^*$ and $Y_2^*$ are obtained (fixed) by combining the conditional
PDF $(\bw,\be)\mapsto f_{\bepsilon\mid\bw}(\be)$ with \eqref{eq:Outcomes}.}

\citet[Theorem 2(iii)]{honore_trimmed_1992} states that if Assumptions
\ref{assu:Non-Degeneracy-TLAD}--\ref{assu:RankRegressors-TLAD} hold, the
expectations involved in defining the matrices $\bfV_0^{\tt{tlad}}$ and
$\bGamma_0^{\tt{tlad}}$ exist (in $\R^{K\times K}$), and both matrices are of
full rank, then
\begin{equation}\label{eq:AsymptoticNormality-TLAD}
  \sqrt{n}\del[1]{\widehat\btheta^{\texttt{tlad}}-\btheta_0}\rightsquigarrow\cN\left(\mathbf{0},(\bGamma_0^{\tt{tlad}})^{-1}\bfV_0^{\tt{tlad}}(\bGamma_0^{\tt{tlad}})^{-1}\right)\;\text{in}\;\R^K.
\end{equation}

Among several steps, the proof of this result in \citet{honore_trimmed_1992}
includes establishing the existence of the Hessian of the expected loss [see
\eqref{eq:PopulationLossFunction-TLAD} below] at the true parameter value
$\btheta_0$. In our notation, this task corresponds to arguing that the function
$L:\R^K\to\R$ defined by
\begin{equation}\label{eq:PopulationLossFunction-TLAD}
  L(\btheta):=\E[m^{\texttt{tlad}}(\Delta\bX^{\top}\btheta,\bY) - m^{\texttt{tlad}}(0,\bY)],\quad\btheta\in\R^{K},
\end{equation}
is twice differentiable at $\btheta=\btheta_0$.\footnote{Note that because the
function $m^{\texttt{tlad}}(t,\by)$ is Lipschitz continuous in its first
argument, it follows from Assumption \ref{assu:MomentConditions-TLAD} that the
expectation in \eqref{eq:PopulationLossFunction-TLAD} is well-defined. Also note
that, as in \citet{honore_trimmed_1992}, we work with the expected loss
\emph{difference} $\E[m^{\texttt{tlad}}(\Delta\bX^{\top}\btheta,\bY) -
m^{\texttt{tlad}}(0,\bY)]$ instead of
$\E[m^{\texttt{tlad}}(\Delta\bX^{\top}\btheta,\bY)]$ because the former loss
gives results under weaker regularity conditions.} To this end, Honor{\'e} used
the LDCT to differentiate once under the expectation, and then argued
differentiability at $\btheta_0$ of the resulting function to arrive at
$\bGamma_0^{\tt{tlad}}$ in \eqref{eq:VarianceSandwichBread-TLAD}. In the next
subsection, however, we will show by example that Assumptions
\ref{assu:Non-Degeneracy-TLAD}--\ref{assu:RankRegressors-TLAD} are not
sufficient to ensure the existence of the Hessian.

\subsection{Counterexample}\label{sec:Counterexample 2}

As in the Section \ref{sec:Counterexample} counterexample, let $K=1$,
$\alpha\equiv0$, $\btheta_0 = 0$, $\bX_1\equiv2$ and $\bX_2\equiv1$. Also, to
define the distribution of the pair $(\varepsilon_1,\varepsilon_2)$, let
$r:\mathbb R\to \mathbb R$ be a continuous function such that (i) $r(t)\geqslant 0$
for all $t\in\mathbb R$, (ii) $\int_{\mathbb R} r(t)\dif t = 1$, and (iii) $r(t)=0$
if $t\leqslant 1$ or $t\geqslant 3$. In addition, let $\mathscr E := \{0,2,4,\dots\}$,
$\mathscr O = \{1,3,5,\dots\}$, and let $\tilde h:[0,1]\to\{0,1\}$ be the function
defined by
$$
\tilde{h}(t)=\begin{cases}
1, & t\in(2^{-(k+1)},2^{-k}]\text{ for }k\in \mathscr E,\\
0, & t\in(2^{-(k+1)},2^{-k}]\text{ for }k\in \mathscr O,\\
0, & t=0.
\end{cases}
$$
Moreover, let $h:\mathbb R \to\mathbb R$ be the function defined by $h(t) =
3\tilde h(|t|) / 4$ for $t\in[-1,1]$ and $0$ otherwise. Note that $h(t)\geqslant0$
for all $t\in\mathbb R$ and
$$
\int_{\mathbb R} h(t)\dif t = \frac{3}{4}\int_{-1}^1 \tilde h(\left|t\right|)\dif t = \frac{3}{2}\int_{0}^1 \tilde h(t)\dif t = \frac{3}{2}\sum_{k\in\mathscr E}\frac{1}{2^{k+1}} = \frac{3}{2}\cdot \frac{1}{2}\cdot\frac{1}{1-\frac{1}{4}} = 1.
$$
Thus, the function $f\colon\mathbb R^2 \to \mathbb R$ defined by
$$
f(e_1,e_2) = 2h(e_1 - e_2)r(e_1 + e_2),\quad e_1,e_2\in\mathbb R,
$$
satisfies $f(e_1,e_2)\geqslant 0$ for all $e_1,e_2\in\mathbb R$ and
\begin{align*}
\int_{\R^2}f(\be)\dif \be
& =2\int_{\R}\int_{\R}h(e_1-e_2)r(e_1+e_2)\dif e_1 \dif e_2\\
&=2\int_{\R}\int_{\R}h(s)r(s+2e_2)\dif s \dif e_2
 =\int_{\R}\int_{\R}h(s)r(t)\dif s \dif t=1.
\end{align*}
Hence, $f$ is the PDF of a certain distribution on $\mathbb R^2$. Let
$(\varepsilon_1,\varepsilon_2)$ be a pair of random variables sampled from this
distribution. Because of the symmetry of the function $f$, the random variables
$\varepsilon_1$ and $\varepsilon_2$ are then exchangeable, and their common PDF
is
$$
f_{\varepsilon}(t) = 2\int_{\mathbb R}h(t - e_2)r(t+e_2)\dif e_2 = 2\int_{\mathbb R} h(u)r(u+2t)\dif u,\quad t\in\mathbb R,
$$
which implies that $0<\varepsilon_1<2$ and $0<\varepsilon_2<2$ with probability one. Moreover, the PDF of the difference $\varepsilon_1 - \varepsilon_2$ is
\begin{align*}
f_{\varepsilon_1 - \varepsilon_2}(t) 
& = \int_{\mathbb R}f(e_2 + t,e_2)\dif e_2 
 = 2\int_{\mathbb R}h(t)r(2e_2 + t)\dif e_2 = \int_{\mathbb R}h(t)r(s)\dif s = h(t),\quad t\in\mathbb R.
\end{align*}
It is then straightforward to check that Assumptions
\ref{assu:Non-Degeneracy-TLAD}--\ref{assu:RankRegressors-TLAD} are all
satisfied. Note, however, that Assumption \ref{assu:extra continuity} is {\em
not} satisfied because the function $f$ is not continuous.

Now, observe that since $\mathbb P(Y_1>0, Y_2>0) = 1$ by construction, a
calculation shows that the expected loss function $L$ in
\eqref{eq:PopulationLossFunction-TLAD} simplifies to
$$
L(\theta) = \E\left[m^{\texttt{tlad}}(\theta,\bY) - m^{\texttt{tlad}}(0,\bY)\right] = \E\left[|Y_1 - Y_2 - \theta| - |Y_1 - Y_2|\right],\quad \theta\in\mathbb R. 
$$
Since the integrand here is convex in $\theta$ and non-differentiable in
$\theta$ only at $Y_1 - Y_2$, which, for a given $\theta$, happens with
probability zero, it follows from \citet[Proposition
2.3]{bertsekas1973stochastic} that $L$ is (everywhere) differentiable with
derivative
$$
\dot{L}(\theta) = \E\left[2\cdot\mathbf{1}\{Y_1 - Y_2 \leqslant \theta\} - 1\right] = 2F_{\varepsilon_1 - \varepsilon_2}(\theta) - 1,
$$
where $F_{\varepsilon_1 - \varepsilon_2}$ denotes the CDF of the difference
$\varepsilon_1 - \varepsilon_2$. We claim that $F_{\varepsilon_1 -
\varepsilon_2}$ is non-differentiable at zero $(=\theta_0)$. Indeed, if $t_k =
2^{-k}$ for $k\in\mathscr E$, then
$$
\frac{F_{\varepsilon_1 - \varepsilon_2}(t_k) - F_{\varepsilon_1 - \varepsilon_2}(0)}{t_k} = 2^k \int_0^{2^{-k}} h(t)\dif t  = \frac{3}{4}\cdot 2^k \cdot \sum_{l\in\mathscr E : l\geqslant k}\frac{1}{2^{l+1}} = \frac{1}{2}
$$
and if $t_k = 2^{-k}$ for $k\in\mathscr O$, then
$$
\frac{F_{\varepsilon_1 - \varepsilon_2}(t_k) - F_{\varepsilon_1 - \varepsilon_2}(0)}{t_k} = 2^k \int_0^{2^{-k}} h(t)\dif t  = \frac{3}{4}\cdot 2^k \cdot \sum_{l\in\mathscr E : l\geqslant k+1}\frac{1}{2^{l+1}} = \frac{1}{4},
$$
which implies that the limit of $(F_{\varepsilon_1 - \varepsilon_2}(t) -
F_{\varepsilon_1 - \varepsilon_2}(0))/t$ as $t\to 0$ does not exist. Thus, the
function $\dot{L}$ is non-differentiable at zero ($=\theta_0$), implying that
the Hessian of $L$ at $\theta_0$ does not exist.

One can show that the TLAD estimator is still $\sqrt{n}$-consistent but not
(asymptotically) normal. Specifically, in Figure
\ref{fig:ECDF_Scaled_Estimates}, we show the (kernel) PDF of the scaled TLAD
estimates $\sqrt{n}\widehat{\theta}^{\tt{tlad}}$ based on $10,000$ Monte Carlo
samples of size $n=50,000$ using the DGP defined earlier in this
section.\footnote{The kernel density was created using the \texttt{R} package
\texttt{ggplot2} with \texttt{geom\_density}. We use a Gaussian kernel and the
\citet[Equation (3.31)]{silverman_density_1986} rule-of-thumb bandwidth (both
\texttt{geom\_density }defaults).} To facilitate comparison, we include the PDF
resulting from fitting a Gaussian distribution to the Monte Carlo dataset of
scaled estimates using maximum likelihood. The PDF of the scaled estimates looks
far from Gaussian. We therefore conclude that establishing the asymptotic
normality result for the TLAD estimator requires more than Assumptions
\ref{assu:Non-Degeneracy-TLAD}--\ref{assu:RankRegressors-TLAD}. We do so by
imposing the continuity conditions in Assumption \ref{assu:extra continuity}. 

\begin{figure}\caption{\normalsize{PDFs of Scaled TLAD Estimates (Solid) and Best Normal
Approximation (Dashed)}}
  \centering
  \includegraphics[width=\textwidth, trim={0 0 0 0.75cm}, clip]{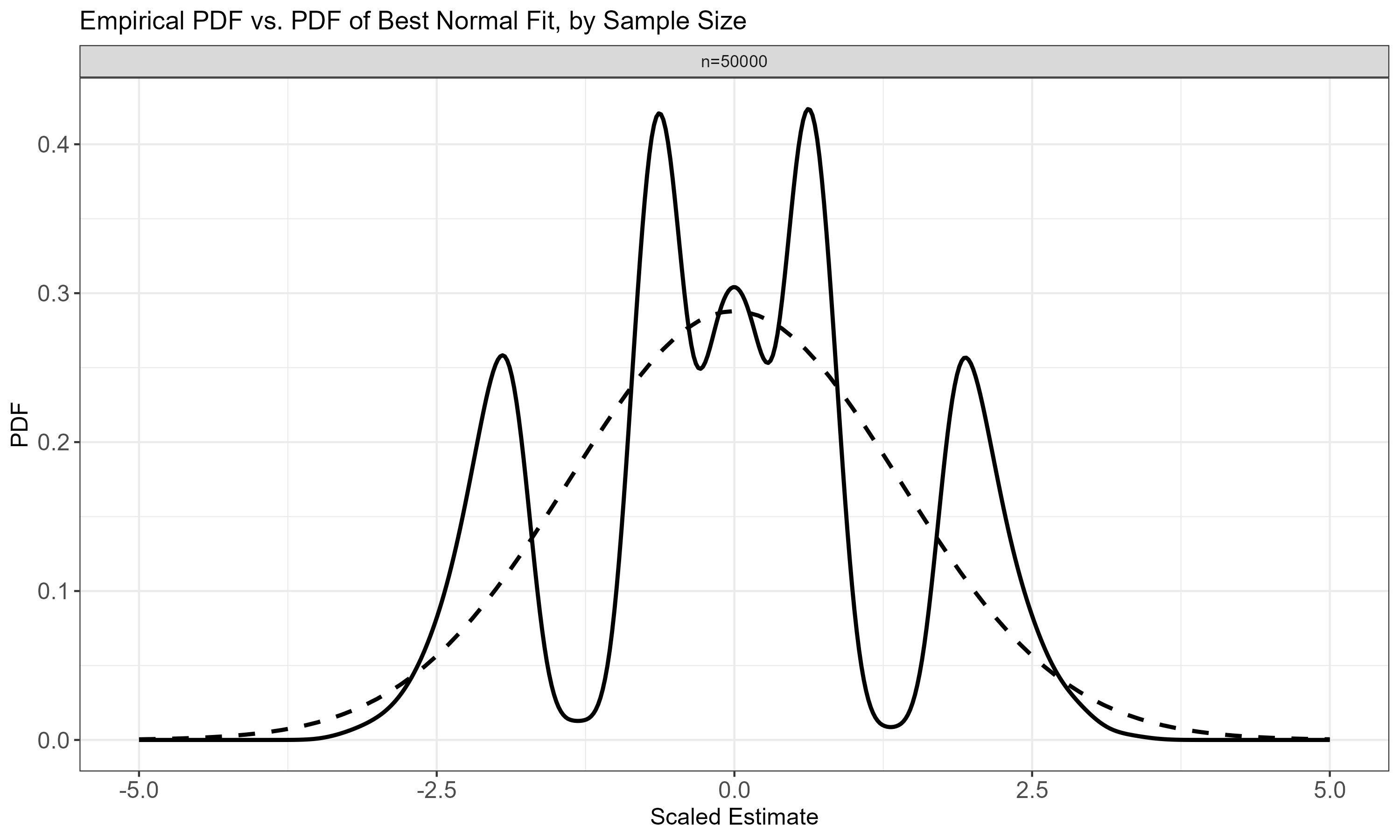}
  \label{fig:ECDF_Scaled_Estimates}
\end{figure}

\subsection{Extended Asymptotic Normality
Result}\label{sec:HessianExistence-TLAD} In this subsection, we provide a
Hessian existence argument for the TLAD estimator using  Assumptions
\ref{assu:Non-Degeneracy-TLAD}--\ref{assu:extra continuity}. We also show that,
unless the pair of latent outcomes $(Y_1^\ast,Y_2^\ast)$ belongs to the first
quadrant with conditional probability given $\bW$ equal to one almost surely,
$\bGamma_0^{\tt{tlad}}$ will \emph{not} be the Hessian of the expected loss
(because of an apparent typographical error). 
We then state the asymptotic normality result under Assumptions
\ref{assu:Non-Degeneracy-TLAD}--\ref{assu:extra continuity} with the corrected Hessian.

To state the existence theorem, let $(\bw,\by)\mapsto
f_{\bY^\ast\mid\bw}\left(\by\right)$ denote the joint PDF of the latent outcomes
$Y_1^\ast$ and $Y_2^\ast$ conditional on $\bW=\bw$ for all $\bw\in\cW$.

\begin{thm}[\textbf{Trimmed Least Absolute Deviations Hessian
Existence}]\label{thm:LossHessianExistence-TLAD} Under Assumptions
\ref{assu:Non-Degeneracy-TLAD}--\ref{assu:extra continuity}, the function
$L:\R^K\to\R$ defined by \eqref{eq:PopulationLossFunction-TLAD} is twice
differentiable at $\btheta=\btheta_0$. Its Hessian matrix $\bfH_0^{\tt{tlad}} :=
\nabla^2L(\btheta_0)$ is given by
\begin{equation}\label{eq:LossHessian-TLAD}
  \left.\begin{aligned}
    \bfH_0^{\tt{tlad}}
    &=\E\Bigg[\bigg(2\int_{0}^{+\infty}f_{\bY^{\ast}\mid\bW}\del[1]{z+\max\cbr[0]{0,\Delta\bX^\top\btheta_0},z-\min\cbr[0]{0,\Delta\bX^\top\btheta_0}}\dif z\\
    & \quad\qquad+\bone\{\Delta\bX^\top\btheta_0\geqslant0\}\int_{-\infty}^{0}f_{\bY^{\ast}\mid\bW}\del[1]{\Delta\bX^\top\btheta_0,z}\dif z\\
    & \quad\qquad+\bone\{\Delta\bX^\top\btheta_0<0\}\int_{-\infty}^{0}f_{\bY^{\ast}\mid\bW}\del[1]{z,-\Delta\bX^\top\btheta_0}\dif z\bigg)\Delta\bX\Delta\bX^\top\Bigg].
  \end{aligned}\right\}
\end{equation}
\end{thm}

\begin{rem}[\textbf{Comparison with \citet{honore_trimmed_1992}}] Comparing
$\bfH_0^{\tt{tlad}}$ in \eqref{eq:LossHessian-TLAD} with the matrix
$\bGamma_0^{\tt{tlad}}$ in \eqref{eq:VarianceSandwichBread-TLAD}, and writing
out the definitions of the conditional PDFs involved, we see that the latter two
contributions to each expression are equal. (To align the two expressions, we
here interpret ``undefined'' times ``zero'' as ``zero.'') However, comparing the
\emph{first} terms of each matrix, we see that the conditional PDF
$f_{Y_1^\ast-Y_2^\ast\mid\bW,Y_1^\ast>0,Y_2^\ast>0}(\Delta\bX^\top\btheta_0)$ in
Honor{\'e}'s $\bGamma_0^{\tt{tlad}}$ is missing multiplication by the
conditional probability $\bbP\left(Y_1^\ast>0, Y_2^\ast>0\middle|\bW\right)$.
Hence, unless the pair of latent outcomes $(Y_1^\ast,Y_2^\ast)$ belongs to the
first quadrant with conditional probability given $\bW$ equal to one almost
surely, $\bGamma_0^{\tt{tlad}}$ will \emph{not} be the Hessian of the expected
loss. Of course, if $(Y_1^\ast,Y_2^\ast)$ resides in the first quadrant with
probability one, then the model involves no censoring.\hfill$\diamondsuit$
\end{rem}

We now revise the asymptotic normality statement in
\eqref{eq:AsymptoticNormality-TLAD} for the TLAD estimator based on our new
understanding of the Hessian $\bfH_0^{\tt{tlad}}$.
\begin{thm}[\textbf{Asymptotic Normality of Trimmed Least Absolute
Deviations}]\label{thm:AsymptoticNormality-TLAD} Let Assumptions
\ref{assu:Non-Degeneracy-TLAD}--\ref{assu:extra continuity} hold, and suppose
that the expectations involved in defining the matrix $\bfV_0^{\tt{tlad}}$ exist
(in $\R^{K\times K}$), and that both matrices $\bfV_0^{\tt{tlad}}$ and
$\bfH_0^{\tt{tlad}}$ are of full rank. Then
\begin{equation}\label{eq:AsymptoticNormality-TLAD-2}
  \sqrt{n}\del[1]{\widehat\btheta^{\tt{tlad}}-\btheta_0}\rightsquigarrow\cN\left(\mathbf{0},(\bfH_0^{\tt{tlad}})^{-1}\bfV_0^{\tt{tlad}}(\bfH_0^{\tt{tlad}})^{-1}\right)\;\text{in}\;\R^K.
\end{equation}
\end{thm}

\begin{rem}[\textbf{Comparison with \citet{honore_pairwise_1994}}] We return to
the comparison with \citet{honore_pairwise_1994}, who considered censored
regression with cross-sectional data as discussed in Remark
\ref{rem:ComparisonWithHonorePowellJacobian}. On \emph{ibid.}~(p.~260), the
authors gave an expression for the Hessian of the expected loss at the true
parameter value $\btheta_0$ when using the trimmed absolute loss. Their
expression is (in our notation)
\begin{align*}
  &\E\Bigg[\bigg(2\int_{0}^{+\infty}f_{\varepsilon}\del[1]{z-\min\cbr[1]{\bX_1^\top\btheta_0,\bX_2^\top\btheta_0}}^2\dif z\\
  &\qquad+f_{\varepsilon}\del[1]{-\min\cbr[1]{\bX_1^\top\btheta_0,\bX_2^\top\btheta_0}}F_{\varepsilon}\del[1]{-\min\cbr[1]{\bX_1^\top\btheta_0,\bX_2^\top\btheta_0}}\bigg)
    \Delta\bX\Delta\bX^\top\Bigg],
\end{align*}
with the pairs $(Y_1,\bX_1)$ and $(Y_2,\bX_2)$ representing independent units,
where $f_{\varepsilon}$ and $F_{\varepsilon}$ denote the PDF of $\varepsilon$
and the CDF of $\varepsilon$, respectively. Upon setting $\alpha\equiv0$ in our
panel-data setting, and taking the model errors $(\varepsilon_1,\varepsilon_2)$
to be independent and identically distributed, the conditional PDF of
$(Y_1^\ast,Y_2^\ast)$ given $\bW=\bw$ factors as
$
  f_{Y_1^\ast,Y_2^\ast\mid\bw}\left(y_1^\ast,y_2^\ast\right)=f_{\varepsilon}\del[1]{y_1^\ast-\bx_1^\top\btheta_0}f_{\varepsilon}\del[1]{y_2^\ast-\bx_2^\top\btheta_0}.
$
A straightforward calculation then shows that our panel-data Hessian in
\eqref{eq:LossHessian-TLAD} reduces to the cross-sectional Hessian from
\citet{honore_pairwise_1994} in the previous display.\hfill$\diamondsuit$
\end{rem}

\subsection{Asymptotic Variance Estimation}\label{sec:LADAsymptoticNormalityAndVarianceEstimation}

We now discuss estimation of the asymptotic variance
$\bSigma_0^{\tt{tlad}}:=(\bfH_0^{\tt{tlad}})^{-1}\bfV_0^{\tt{tlad}}(\bfH_0^{\tt{tlad}})^{-1}$
in \eqref{eq:AsymptoticNormality-TLAD-2}. Since the formula for
$\bfH_0^{\tt{tlad}}$ given in \eqref{eq:LossHessian-TLAD} includes a
nonparametric conditional PDF, we propose a bootstrap estimator instead of a
plug-in estimator. Also, since a \emph{na{\"i}ve} bootstrap variance estimator
may fail to be consistent for LAD-type estimators without additional
integrability conditions \citep{BF81,GPSB84}, we instead construct a
quantile-based robust bootstrap variance estimator that avoids reliance on
bootstrap second moments. 

The \emph{robust} bootstrap variance estimator exploits the fact that quantiles
of linear combinations of a normal distribution scale with their standard
deviations. We therefore estimate bootstrap quantiles of suitably centered
parameter combinations and invert this scaling to recover the corresponding
covariance entries. This construction replaces unstable bootstrap variance
estimation based on second moments by identification of covariance entries
through the asymptotic normal geometry of the estimator.

To describe the robust bootstrap variance estimator in more detail, consider a
bootstrap sample
$\{(\widetilde{Y}_{i1},\tbX_{i1},\widetilde{Y}_{i2},\tbX_{i2})\}_{i=1}^n$.
obtained by sampling $n$ observation indices with replacement from
$\{1,2,\dots,n\}$ and extracting the corresponding observations from the
original sample. Also, let 
\begin{equation}\label{eq: bootstrap}
  \tbtheta^{\tt{tlad}} := (\widetilde{\theta}^{\tt{tlad}}_1,\dots,\widetilde{\theta}^{\tt{tlad}}_K)^\top \in \argmin_{\btheta\in\R^K}\cbr[3]{\frac{1}{n}\sum_{i=1}^n m^{\tt{tlad}}\big((\Delta\tbX_i)^\top\btheta,\tbY_i\big)}
\end{equation}
be the bootstrap analog of $\widehat\btheta^{\tt{tlad}} =
(\widehat\theta^{\tt{tlad}}_1,\dots,\widehat\theta^{\tt{tlad}}_K)^\top$, where
we introduced shorthands $\Delta\tbX_i := \tbX_{i1} - \tbX_{i2}$ and
$\tbY_i:=(\widetilde{Y}_{i1},\widetilde{Y}_{i2})$. Identification of the
covariance entries relies on the identities $\var(Z_k)=\Sigma^{\tt{tlad}}_{0kk}$
and $\var(Z_j+Z_k)=\Sigma^{\tt{tlad}}_{0jj} +\Sigma^{\tt{tlad}}_{0kk}
+2\Sigma^{\tt{tlad}}_{0jk}$ for a normal vector $\bZ
\sim\cN(\bzero,\bSigma_0^{\tt{tlad}})$. Accordingly, we recover variances and
covariances from bootstrap quantiles of the corresponding linear combinations.
Motivated by these identities, we estimate the diagonal and off-diagonal entries
of $\bSigma_0^{\tt{tlad}}=[\Sigma^{\tt{tlad}}_{0jk}]_{j,k=1}^K$ separately. For
the \emph{diagonal} elements, for each $k\in[K]$, we set
$\widehat\Sigma_{kk}^{\tt{tlad}}:=[\widehat{q}_{0.9,k}/\Phi^{-1}(0.95)]^2$,
where $\widehat{q}_{0.9,k}$ is the $0.9$ quantile of the conditional
distribution of $\sqrt n|\widetilde\theta^{\tt{tlad}}_k -
\widehat\theta^{\tt{tlad}}_k|$ given the original data and $\Phi^{-1}(0.95)$ is
the $0.95$ quantile of the standard normal distribution. For the
\emph{off-diagonal} elements, for each $(j,k)\in[K]\times[K]$ such that $j\neq
k$, we set $\widehat \Sigma_{jk}^{\tt{tlad}}:=
([\widehat{q}_{0.9,j,k}/\Phi^{-1}(0.95)]^2 - \widehat\Sigma_{jj}^{\tt{tlad}} -
\widehat\Sigma_{kk}^{\tt{tlad}})/2$, where $\widehat{q}_{0.9,j,k}$ is the $0.9$
quantile of the conditional distribution of $\sqrt
n|\widetilde\theta^{\tt{tlad}}_j + \widetilde\theta^{\tt{tlad}}_k -
\widehat\theta^{\tt{tlad}}_j - \widehat\theta^{\tt{tlad}}_k|$ given the original
data.\footnote{As the robust bootstrap covariance estimator
$\widehat\bSigma^{\tt{tlad}}$ is recovered entrywise from quantiles of linear
combinations, it need not be positive semidefinite in finite samples. One can
construct a positive semidefinite variance estimator by projecting
$\widehat\bSigma^{\tt{tlad}}$ onto the set of positive semidefinite matrices.}
In the following theorem, we prove consistency of
$\widehat\bSigma^{\tt{tlad}}:=[\widehat\Sigma_{jk}^{\tt{tlad}}]_{j,k=1}^K$ under
the assumptions of Theorem \ref{thm:AsymptoticNormality-TLAD}.

\begin{thm}[\textbf{Consistency of the Robust Bootstrap Variance Estimator for
TLAD}]\label{thm: bootstrap consistency lad} Let the assumptions of Theorem
\ref{thm:AsymptoticNormality-TLAD} hold. Then
$\widehat\bSigma^{\tt{tlad}}\to_{\bbP}\bSigma_0^{\tt{tlad}}$.
\end{thm}

\begin{rem}[\textbf{Reducing Computational Burden via \citet{HH17}}]
Note that calculating the quantiles $\widehat{q}_{0.9,k}$ and
$\widehat{q}_{0.9,j,k}$ requires solving the $K$-dimensional optimization
problem \eqref{eq: bootstrap} for many draws of the bootstrap sample. To reduce
computational burden, one may instead employ the alternative bootstrap procedure
of \citet{HH17}, which replaces repeated $K$-dimensional optimization by a
sequence of one-dimensional problems. This approach is particularly attractive
in our setting because the ``meat'' matrix $\bfV_0^{\tt{tlad}}$ in
\eqref{eq:VarianceSandwichMeat-TLAD} admits a consistent plug-in estimator. We
refer the reader to the original paper for further details.\hfill$\diamondsuit$
\end{rem}

\begin{rem}[\textbf{Bootstrap Inference for Convex Pairwise Differencing Estimators}]
A related recent paper is \citet{cattaneo2025robust}, which develops
distribution theory and bootstrap-based inference for a broad class of convex
pairwise differencing estimators. Their framework includes trimmed absolute loss
as a special case. Our analysis in Theorem \ref{thm: bootstrap consistency lad}
is complementary in that it considers a panel setting, where the trimmed loss is
used to eliminate an individual-specific fixed effect and no bandwidth choice is
required. In addition, while \citet{cattaneo2025robust} establish bootstrap
validity for the distributional approximation, Theorem \ref{thm: bootstrap
consistency lad} establishes consistency of a robust bootstrap variance
estimator for the asymptotic variance of the TLAD estimator under the
assumptions of Theorem \ref{thm:AsymptoticNormality-TLAD}.\hfill$\diamondsuit$
\end{rem}

\begin{rem}[\textbf{Conservative Inference}]
The robust bootstrap variance estimator $\widehat\bSigma^{\tt{tlad}}$ is
consistent for $\bSigma_0^{\tt{tlad}}$ under the assumptions of Theorem
\ref{thm:AsymptoticNormality-TLAD}. However, the construction of
$\widehat\bSigma^{\tt{tlad}}$ is based on a particular choice of quantile level
(i.e., $0.9$ for the absolute differences, corresponding to $0.95$ for the
standard normal quantiles). This choice is not essential: the consistency
argument in the proof of Theorem \ref{thm: bootstrap consistency lad} continues
to apply for any fixed interior quantile. We use the particular quantile level
above as a convenient way to avoid reliance on bootstrap second moments, which
may be unstable for LAD-type estimators without additional integrability
conditions.

For comparison, \citet{hahn2021bootstrap} show that, under their conditions,
bootstrap variance estimators based on second moments, although not necessarily
consistent in LAD-type settings, can still be asymptotically conservative for
inference on fixed linear combinations. Thus, if one is willing to forgo
consistency of the variance estimator itself and instead target conservative
inference for linear combinations of the parameter vector, the na{\"i}ve
bootstrap provides an alternative route. This conservative-inference perspective
is distinct from the goal pursued here, namely consistent estimation of the full
asymptotic covariance matrix.\hfill$\diamondsuit$
\end{rem}

\section*{Funding}
Honor{\'e} and S{\o}rensen are grateful for research support from the Aarhus
Center for Econometrics (ACE) funded by the Danish National Research Foundation
grant number DNRF186. Honor{\'e} also acknowledges support from the Gregory
C.~Chow Econometric Research Program at Princeton University.

\begin{singlespace}
  \bibliographystyle{ecta}
  \bibliography{bo}   

\begin{thebibliography}{22}
\newcommand{\enquote}[1]{``#1''}
\expandafter\ifx\csname natexlab\endcsname\relax\def\natexlab#1{#1}\fi

\bibitem[\protect\citeauthoryear{Arcones and Gin{\'e}}{Arcones and
  Gin{\'e}}{1992}]{AG92}
\textsc{Arcones, M. and E.~Gin{\'e}} (1992): \enquote{On the bootstrap of
  M-estimators and other statistical functionals,} in \emph{Exloring the Limits
  of Bootstrap}, New York: Wiley, 13--47.

\bibitem[\protect\citeauthoryear{Arellano and Honor{\'e}}{Arellano and
  Honor{\'e}}{2001}]{arellano2001panel}
\textsc{Arellano, M. and B.~Honor{\'e}} (2001): \enquote{Panel data models:
  some recent developments,} in \emph{Handbook of econometrics}, Elsevier,
  vol.~5, 3229--3296.

\bibitem[\protect\citeauthoryear{Bertsekas}{Bertsekas}{1973}]{bertsekas1973stochastic}
\textsc{Bertsekas, D.~P.} (1973): \enquote{Stochastic optimization problems
  with nondifferentiable cost functionals,} \emph{Journal of Optimization
  Theory and Applications}, 12, 218--231.

\bibitem[\protect\citeauthoryear{Bickel and Freedman}{Bickel and
  Freedman}{1981}]{BF81}
\textsc{Bickel, P. and D.~Freedman} (1981): \enquote{Some asymptotic theory for
  the bootstrap,} \emph{Annals of Statistics}, 1196--1217.

\bibitem[\protect\citeauthoryear{Billingsley}{Billingsley}{1995}]{billingsley1995probability}
\textsc{Billingsley, P.} (1995): \emph{Probability and Measure}, John Wiley \&
  Sons.

\bibitem[\protect\citeauthoryear{Cattaneo, Jansson, and Nagasawa}{Cattaneo
  et~al.}{2025}]{cattaneo2025robust}
\textsc{Cattaneo, M.~D., M.~Jansson, and K.~Nagasawa} (2025): \enquote{Robust
  Inference for Convex Pairwise Difference Estimators,} \emph{arXiv preprint
  arXiv:2510.05991}.

\bibitem[\protect\citeauthoryear{Dudley}{Dudley}{2004}]{D04}
\textsc{Dudley, R.} (2004): \emph{Real Analysis and Probability}, Cambridge
  University Press.

\bibitem[\protect\citeauthoryear{Ghosh, Parr, Singh, and Babu}{Ghosh
  et~al.}{1984}]{GPSB84}
\textsc{Ghosh, M., W.~Parr, K.~Singh, and J.~Babu} (1984): \enquote{A note on
  bootstrapping the sample median,} \emph{Annals of Statistics}, 1130--1135.

\bibitem[\protect\citeauthoryear{Hahn}{Hahn}{1995}]{H95}
\textsc{Hahn, J.} (1995): \enquote{Bootstrapping quantile regression
  estimators,} \emph{Econometric Theory}, 105--121.

\bibitem[\protect\citeauthoryear{Hahn and Liao}{Hahn and
  Liao}{2021}]{hahn2021bootstrap}
\textsc{Hahn, J. and Z.~Liao} (2021): \enquote{Bootstrap standard error
  estimates and inference,} \emph{Econometrica}, 89, 1963--1977.

\bibitem[\protect\citeauthoryear{Honor{\'e} and Hu}{Honor{\'e} and
  Hu}{2017}]{HH17}
\textsc{Honor{\'e}, B. and L.~Hu} (2017): \enquote{Poor (wo)man's bootstrap,}
  \emph{Econometrica}, 1277--1301.

\bibitem[\protect\citeauthoryear{Honor{\'e}}{Honor{\'e}}{1992}]{honore_trimmed_1992}
\textsc{Honor{\'e}, B.~E.} (1992): \enquote{Trimmed {LAD} and least squares
  estimation of truncated and censored regression models with fixed effects,}
  \emph{Econometrica}, 533--565.

\bibitem[\protect\citeauthoryear{Honore and Kyriazidou}{Honore and
  Kyriazidou}{2000}]{honore2000estimation}
\textsc{Honore, B.~E. and E.~Kyriazidou} (2000): \enquote{Estimation of
  Tobit-type models with individual specific effects,} \emph{Econometric
  Reviews}, 19, 341--366.

\bibitem[\protect\citeauthoryear{Honor{\'e} and Powell}{Honor{\'e} and
  Powell}{1994}]{honore_pairwise_1994}
\textsc{Honor{\'e}, B.~E. and J.~L. Powell} (1994): \enquote{Pairwise
  difference estimators of censored and truncated regression models,}
  \emph{Journal of Econometrics}, 64, 241--278.

\bibitem[\protect\citeauthoryear{Kosorok}{Kosorok}{2008}]{kosorok2008introduction}
\textsc{Kosorok, M.~R.} (2008): \emph{Introduction to empirical processes and
  semiparametric inference}, Springer.

\bibitem[\protect\citeauthoryear{Nolan and Pollard}{Nolan and
  Pollard}{1987}]{nolan1987_u_processes_rates}
\textsc{Nolan, D. and D.~Pollard} (1987): \enquote{U-processes: rates of
  convergence,} \emph{The Annals of Statistics}, 780--799.

\bibitem[\protect\citeauthoryear{Pakes and Pollard}{Pakes and
  Pollard}{1989}]{pakes1989simulation}
\textsc{Pakes, A. and D.~Pollard} (1989): \enquote{Simulation and the
  asymptotics of optimization estimators,} \emph{Econometrica}, 1027--1057.

\bibitem[\protect\citeauthoryear{Powell}{Powell}{1984}]{powell1984least}
\textsc{Powell, J.~L.} (1984): \enquote{Least absolute deviations estimation
  for the censored regression model,} \emph{Journal of Econometrics}, 25,
  303--325.

\bibitem[\protect\citeauthoryear{Powell}{Powell}{1986}]{powell1986symmetrically}
---\hspace{-.1pt}---\hspace{-.1pt}--- (1986): \enquote{Symmetrically trimmed
  least squares estimation for Tobit models,} \emph{Econometrica}, 1435--1460.

\bibitem[\protect\citeauthoryear{Royden and Fitzpatrick}{Royden and
  Fitzpatrick}{2023}]{royden2023real}
\textsc{Royden, H.~L. and P.~M. Fitzpatrick} (2023): \emph{Real Analysis},
  Prentice-Hall, 5th ed.

\bibitem[\protect\citeauthoryear{Silverman}{Silverman}{1986}]{silverman_density_1986}
\textsc{Silverman, B.~W.} (1986): \emph{Density {Estimation} for {Statistics}
  and {Data} {Analysis}}, vol.~26 of \emph{Monographs on {Statistics} and
  {Applied} {Probability}}, London: Chapman \& Hall.

\bibitem[\protect\citeauthoryear{Van~der Vaart and Wellner}{Van~der Vaart and
  Wellner}{1996}]{vdVW1996weak}
\textsc{Van~der Vaart, A.~W. and J.~A. Wellner} (1996): \emph{Weak Convergence
  and Empirical Processes: With Applications to Statistics}, Springer.

\end{thebibliography}
\end{singlespace}

\appendix
\part*{Appendix}
\addcontentsline{toc}{part}{Appendix}


\paragraph*{Notation.} Throughout the main and supplemental appendices, we use
the following additional notation. We write $\N=\{1,2,\dotsc\}$ for the natural
numbers and abbreviate $[k]:=\{1,2,\dotsc,k\}$ for $k\in\N$. We let $\cB_d$ and
$\lambda_d$ denote the Borel subsets of $\R^d$ and the Lebesgue measure on
$\R^d$, respectively, and abbreviate $\cB:=\cB_1$ and $\lambda:=\lambda_1$. For
any (Lebesgue integrable) function $f$, we use $\int_{\R^d}f\dif\lambda_d$,
$\int_{\R^d}f(\bu)\lambda_d(\dif\bu)$ and $\int_{\R^d}f(\bu)\dif\bu$
interchangeably to denote integration against the Lebesgue measure on $\R^d$. We
denote the underlying probability space for the data as $(\Omega,\cA,\bbP)$. 

\section{Proofs for Trimmed Least Squares (Section
\ref{sec:TLS})}\label{sec:ProofsTLS}

\subsection{Proof of Theorem \ref{thm:JacobianExistence-TLS}} 

To prove Theorem \ref{thm:JacobianExistence-TLS}, we first show that the
conditional expectation of the trimmed square loss derivative is essentially
well defined under our assumptions. The proof of the following lemma can be
found at the end of this section. To state the lemma, recall that
$f_{\bepsilon\mid\bw}(\cdot)$ denotes the version of the conditional PDF of
$\bepsilon$ given $\bW=\bw$ provided by Lemma \ref{lem: regular conditional
probability}.
\begin{lem}[\textbf{Well-Definedness of Expected Trimmed Square Loss Derivative}]\label{lem:WellDefinedExpectedTrimmedSqLossDerivative}
Let Assumptions
\ref{assu:Non-Degeneracy-TLS}--\ref{assu:RankRegressors-TLS} hold.
Then there is a subset $\cW'\subseteq\cW$ such that $\bbP(\bW\in\cW')=1$ and 
for all $t\in\R$ and $\bw=(\bx_1,\bx_2,a)\in\cW'$, 
\[
 \int_{\cE}\left|\dot{m}_{1}^{\tt{tls}}\del[1]{t,\del[0]{\max\{0,a+\bx_1^\top\btheta_0+e_1\},\max\{0,a+\bx_2^\top\btheta_0+e_2\}}}\right|f_{\bepsilon\mid\bw}(\be)\dif\be < \infty.
\]
\end{lem}
Fixing $\cW'$ provided by Lemma
\ref{lem:WellDefinedExpectedTrimmedSqLossDerivative}, we define a measurable
function $M^{\tt{tls}}:\R\times\cW\to\R$ by
\begin{align}
M^{\tt{tls}}(t,\bw)
&:=\E\left[\dot{m}_{1}^{\tt{tls}}\del[0]{t,\bY}\middle|\bW=\bw\right]\label{eq:ExpectedLossDerivativeDefn},
\end{align}
if $\bw\in\cW'$ and (arbitrarily) set $M^{\tt{tls}}(t,\bw):=0$ if
$\bw\in\cW\backslash\cW'$, and the conditional expectation is understood as an
integral against $f_{\bepsilon\mid\bw}(\cdot)$. With this notation, we can
express $\bG$ as
\begin{align*}
\bG(\btheta) & =\E\sbr[2]{\E\left[\dot{m}_{1}^{\tt{tls}}\del[0]{\Delta\bX^{\top}\btheta,\bY}\middle|\bW\right]\Delta\bX}=\E\sbr[1]{M^{\tt{tls}}\del[0]{\Delta\bX^{\top}\btheta,\bW}\Delta\bX}.
\end{align*}
The differentiability properties of $\bG$ (at $\btheta_0$) will by and large
follow from those of $M^{\tt{tls}}(\cdot,\bw)$. To show Theorem
\ref{thm:JacobianExistence-TLS}, we further rely on the following two lemmas,
the proofs of which can be found at the end of this section. To state the
lemmas, recall that $f_{\varepsilon\mid\bw}(\cdot)$ is the common marginal PDF
of $\varepsilon_1$ and $\varepsilon_2$ conditional on $\bW = \bw$ for
$\bw\in\mathcal W$, and that $F_{\varepsilon\mid\bw}(\cdot)$ is the
corresponding CDF.

\begin{lem}[\textbf{Expected Trimmed Square Loss
Derivative}]\label{lem:ExpectedTrimmedSqLossDerivative} Let Assumptions
\ref{assu:Non-Degeneracy-TLS}--\ref{assu:RankRegressors-TLS} hold.
Then, for all $t\in\R$ and $\bw=(\bx_1,\bx_2,a)\in\cW'$, the function
$M^{\tt{tls}}$ defined in \eqref{eq:ExpectedLossDerivativeDefn} satisfies
\begin{equation}\label{eq:ExpectedLossDerivative}
M^{\tt{tls}}(t,\bw)
=
\mu_{Y_{2}-Y_{1}}(\bw)
+
\begin{cases}
t+\int_{v_{2}(\bw)}^{v_{2}(\bw)-t}F_{\varepsilon\mid\bw}(u)\dif u, & t<0,\\
0, & t=0,\\
t-\int_{v_{1}(\bw)}^{v_{1}(\bw)+t}F_{\varepsilon\mid\bw}(u)\dif u, & t>0,
\end{cases}
\end{equation}
where $\mu_{Y_{2}-Y_{1}}(\bw):=\E[Y_2-Y_1\mid\bW=\bw]$ and
$v_\tau(\bw):=-(a+\bx_\tau^\top\btheta_0),\tau\in\{1,2\}$. 
\end{lem}

\begin{lem}[\textbf{Differentiability of Expected Trimmed Square Loss
Derivative}]\label{lem:DifferentiabilityOfExpectedTrimmedSqLossDerivative} Let
Assumptions \ref{assu:Non-Degeneracy-TLS}--\ref{assu:RankRegressors-TLS} hold
and fix $\bw = (\bx_1,\bx_2,a)\in\cW'$. 
Then:
\begin{enumerate}[(1)]
\item\label{enu:MLipschitz} $M^{\tt{tls}}(\cdot,\bw)$ is Lipschitz continuous on $\R$
with Lipschitz constant equal to one.
\item\label{enu:MDiffAwayFromZero} $M^{\tt{tls}}(\cdot,\bw)$ is differentiable
at $t\neq0$ with derivative given by
\begin{equation}\label{eq:DerivativeExpectedLoss}
\dot{M}_1^{\tt{tls}}(t,\bw)=\begin{cases}
1-F_{\varepsilon\mid\bw}(v_{2}(\bw)-t), & t<0,\\
1-F_{\varepsilon\mid\bw}(v_{1}(\bw)+t), & t>0.
\end{cases}
\end{equation}
\item\label{enu:MDirDiffAtZero} $M^{\tt{tls}}(\cdot,\bw)$ is semi-differentiable at $t=0$
with left and right derivatives given by
\begin{equation}\label{eq:LeftAndRightDerivativesExpectedLossAtZero}
  \dot{M}_{1-}^{\tt{tls}}(0,\bw)	=1-F_{\varepsilon\mid\bw}(v_{2}(\bw))\quad\text{and}\quad \dot{M}_{1+}^{\tt{tls}}(0,\bw)=1-F_{\varepsilon\mid\bw}(v_{1}(\bw)),
\end{equation}
respectively.
\item\label{enu:MDiffatxtheta}
$M^{\tt{tls}}(\cdot,\bw)$ is differentiable at $t=\Delta\bx^\top\btheta_0$ with
derivative given by
\begin{equation}\label{eq:DerivativeExpectedLossAtxtheta0}
\dot{M}_1^{\tt{tls}}\del[0]{\Delta\bx^\top\btheta_0,\bw}
=1-F_{\varepsilon\mid\bw}\del[1]{\max\cbr[0]{v_1(\bw),v_2(\bw)}}.
\end{equation}
\end{enumerate}
\end{lem}

\begin{proof}[\sc{Proof of Theorem \ref{thm:JacobianExistence-TLS}}]
For notational convenience, abbreviate $M^{\tt{tls}}(\cdot,\bw)$ by
$M(\cdot,\bw)$. 
Fix $\bvtheta\in\R^{K}$ and let
$\{\tau_m\}_{m=1}^{\infty}\subset(0,\infty)$ and
$\{\bvtheta_m\}_{m=1}^{\infty}\subset\R^K$ be such that $\tau_{m}\to0_+$ and
$\bvtheta_m\to\bvtheta$. Then
\begin{align*}
\frac{\bG(\btheta_{0}+\tau_{m}\bvtheta_m)-\bG(\btheta_{0})}{\tau_{m}} 
&=\E\left[\frac{M(\Delta\bX^{\top}(\btheta_{0}+\tau_{m}\bvtheta_m),\bW)-M(\Delta\bX^{\top}\btheta_{0},\bW)}{\tau_{m}}\Delta\bX\right]\\
& =\E\left[\frac{M(\Delta\bX^{\top}\btheta_{0}+\tau_{m}\Delta\bX^{\top}\bvtheta_m,\bW)-M(\Delta\bX^{\top}\btheta_{0},\bW)}{\tau_{m}}\Delta\bX\right].
\end{align*}
To apply the Generalized Lebesgue Dominated Convergence Theorem (GLDCT, Theorem
\ref{thm:GLDCT}) coordinatewise, fix $j\in[K]$ and define functions
$\{f_m\}_{m=1}^\infty$ on $\Omega$ by
\[
f_{m}(\omega):=\frac{M\del[1]{\Delta\bX(\omega)^{\top}\btheta_{0}+\tau_{m}\Delta\bX(\omega)^{\top}\bvtheta_m,\bW(\omega)}
-M\del[1]{\Delta\bX(\omega)^{\top}\btheta_{0},\bW(\omega)}}{\tau_{m}}\Delta X_{j}(\omega).
\]
By measurability of $M(\cdot,\cdot)$ established above, each $f_m$ is measurable
and real-valued. By Lemma
\ref{lem:DifferentiabilityOfExpectedTrimmedSqLossDerivative}.\ref{enu:MLipschitz},
the function $M(\cdot,\bw)$ is Lipschitz continuous with Lipschitz constant one.
Hence, by Lipschitz continuity followed by the the Cauchy--Schwarz inequality,
\[
\left|f_{m}(\omega)\right|
\leqslant\left|\Delta X_{j}(\omega)\right|\envert[1]{\Delta\bX(\omega)^{\top}\bvtheta_m}
\leqslant\enVert[0]{\Delta\bX(\omega)}_2^2\enVert{\bvtheta_m}_2=:g_m(\omega).
\]
Since $\E[\|\Delta\bX\|_2^2]<\infty$ by
Assumption~\ref{assu:MomentConditions-TLS}, each $g_m$ is integrable. The
previous display therefore goes to show that $\{f_{m}\}$ is dominated by the
nonnegative integrable sequence $\{g_m\}$. Since $\bvtheta_m\to\bvtheta$, we
have $g_m\to g$ pointwise on $\Omega$, and
\[
\int_\Omega g_m \dif\bbP=\enVert{\bvtheta_m}_2\E[\|\Delta\bX\|_2^2]
\to \enVert{\bvtheta}_2\E[\|\Delta\bX\|_2^2]=
\int_\Omega g\dif\bbP<\infty,
\]
where $g(\omega):=\enVert[0]{\Delta\bX(\omega)}_2^2\enVert[0]{\bvtheta}_2$.
Lemma
\ref{lem:DifferentiabilityOfExpectedTrimmedSqLossDerivative}.\ref{enu:MDiffatxtheta}
shows that $M(\cdot,\bW(\omega))$ is differentiable at
$t=\Delta\bX(\omega)^{\top}\btheta_{0}$, thus yielding the pointwise convergence
\[
f_{m}(\omega) \to \dot{M}_{1}\del[1]{\Delta\bX(\omega)^{\top}\btheta_{0},\bW(\omega)}\Delta X_{j}(\omega)\Delta\bX(\omega)^{\top}\bvtheta=:f(\omega).
\]
Appealing to the GLDCT, stacking over the coordinates $j\in[K]$, and unpacking
the shorthand notations, we get
\begin{align*}
\frac{\bG(\btheta_{0}+\tau_{m}\bvtheta_m)-\bG(\btheta_{0})}{\tau_{m}}
&\to\E\sbr[2]{\del[2]{1-F_{\varepsilon\mid\bW}\left(\max\left\{V_1,V_2\right\}\right)}\Delta\bX\Delta\bX^\top}\bvtheta\\
&=\E\sbr[2]{\del[2]{1-F_{\varepsilon\mid\bW}\del[1]{-\alpha-\min\{\bX_1^\top\btheta_0,\bX_2^\top\btheta_0\}}}\Delta\bX\Delta\bX^\top}\bvtheta.
\end{align*}
Since the limit exists for every $\bvtheta\in\R^K$, is linear in $\bvtheta$, and
is independent of the sequences $\{\tau_m\}$ and $\{\bvtheta_m\}$, it follows
that $\bG$ is Hadamard differentiable at $\btheta_0$. Since $\R^K$ is
finite-dimensional, Hadamard differentiability is equivalent to (Fr\'echet)
differentiability. Hence, $\bG$ is differentiable at $\btheta_0$ with Jacobian
given by
\[
\nabla\bG(\btheta_{0})=\E\sbr[2]{\del[2]{1-F_{\varepsilon\mid\bW}(-\alpha-\min\{\bX_1^\top\btheta_0,\bX_2^\top\btheta_0\})}\Delta\bX\Delta\bX^\top}.
\]
To arrive at the claimed equivalent expression, \eqref{eq:JacobianMidpoint},
\eqref{eq:JacobianAlt1} and \eqref{eq:JacobianAlt2}, recall the latent
outcomes
$Y_\tau^\ast=\alpha+\bX_\tau^\top\btheta_0+\varepsilon_\tau,\tau\in\{1,2\}$, so
that $Y_\tau=\max\{0,Y_\tau^\ast\}$ and, thus,
$\mathbf{1}\{Y_\tau>0\}=\mathbf{1}\{Y_\tau^\ast>0\}$. Condition on $\bW$, and
observe that when $\bX_1^\top\btheta_0\leqslant\bX_2^\top\btheta_0$, we have
\begin{align*}
  F_{\varepsilon\mid\bW}(-\alpha-\min\{\bX_1^\top\btheta_0,\bX_2^\top\btheta_0\})
  &=\E\left[\mathbf{1}\{\varepsilon_1\leqslant-\alpha-\bX_1^\top\btheta_0\}\middle|\bW\right]
  =\E\left[\mathbf{1}\{Y_1^\ast\leqslant 0\}\middle|\bW\right].
\end{align*}
Hence, when $\bX_1^\top\btheta_0\leqslant\bX_2^\top\btheta_0$,
\begin{align*}
1-F_{\varepsilon\mid\bW}(-\alpha-\min\{\bX_1^\top\btheta_0,\bX_2^\top\btheta_0\})
&=\E\left[\mathbf{1}\{Y_1^\ast> 0\}\middle|\bW\right]=\E\left[\mathbf{1}\{Y_1>0\}\middle|\bW\right].
\end{align*}
Similarly, when $\bX_1^\top\btheta_0>\bX_2^\top\btheta_0$,
\begin{align*}
  1-F_{\varepsilon\mid\bW}(-\alpha-\min\{\bX_1^\top\btheta_0,\bX_2^\top\btheta_0\})
  &=\E\left[\mathbf{1}\{Y_2^\ast> 0\}\middle|\bW\right]=\E\left[\mathbf{1}\{Y_2>0\}\middle|\bW\right].
\end{align*}
The expression \eqref{eq:JacobianAlt1} follows by the law of iterated
expectations. The expression in \eqref{eq:JacobianAlt2} follows by the same
argument, grouping the equality event $\bX_1^\top\btheta_0=\bX_2^\top\btheta_0$
with $\mathbf{1}\{Y_2>0\}$ instead. Finally, the expression \eqref{eq:JacobianMidpoint}
is the simple average of the first two.
\end{proof}

We end this section by providing the proofs of Lemmas
\ref{lem:WellDefinedExpectedTrimmedSqLossDerivative},
\ref{lem:ExpectedTrimmedSqLossDerivative} and
\ref{lem:DifferentiabilityOfExpectedTrimmedSqLossDerivative}, in turn.
\begin{proof}[\sc{Proof of Lemma \ref{lem:ExpectedTrimmedSqLossDerivative}}] 
Fix $t\in\R$. By definition of $\dot{m}_{1}^{\tt{tls}}$ in
\eqref{eq:TrimmedSqLossDerivative} and elementary inequalities
\[
|\dot{m}_{1}^{\tt{tls}}(t,\bY)|\leqslant |t|+|\Delta Y|\leqslant |t|+|\Delta\bX^{\top}\btheta_{0}|+|\Delta\varepsilon|.
\]
Thus,
\(
\E\left[\dot{m}_{1}^{\tt{tls}}(t,\bY)\middle|\bW=\bw\right]
\)
exists (in $\R$) if
\(
\E\left[|\Delta\varepsilon|\middle|\bW=\bw\right]<\infty.
\)
Finiteness of the latter is, in turn, guaranteed if
\(
\E\left[|\Delta\varepsilon|^{2}\middle|\bW=\bw\right]<\infty.
\) 
(Again, the previous three conditional expectations should be understood as integrals
against $f_{\bepsilon\mid\bw}(\cdot)$.) Define $q_0:\cW\to[0,\infty]$ by
\[
q_{0}(\bw)=\int_{\cE}|e_{1}-e_{2}|^{2}f_{\bepsilon\mid\bw}(\be)\dif \be,\quad \bw\in\cW,
\]
so that $q_{0}(\bw)=\E[|\Delta\varepsilon|^{2}\mid\bW=\bw]$ by
definition, and set
\[
\cW_0:=\left\{\bw\in\cW\middle|q_{0}(\bw)=\infty,\,\|\Delta\bx\|_{2}>0\right\} .
\]
We will show that $P(\bW\in\cW_0)=P_{\bW}(\cW_0)=0$, so that we can take
$\cW':=\cW\backslash\cW_0$ and the asserted claim will follow.

To this end, consider
\[
g(\bw,\be)=|e_{1}-e_{2}|^{2}\cdot\|\Delta\bx\|_{2}^{4},\quad(\bw,\be)\in\cW\times\cE.
\]
By Assumption \ref{assu:MomentConditions-TLS}, $g$ is $P$-integrable:
\[
\E\left[g\left(\bW,\bepsilon\right)\right]=\E\left[|\Delta\varepsilon|^{2}\cdot\|\Delta \bX\|_{2}^{4}\right]<\infty,
\]
and so defining $q:\cW\to[0,\infty]$ by
\[
q(\bw):=\int_{\cE}|e_{1}-e_{2}|^{2}\cdot\|\Delta\bx\|_{2}^{4}f_{\bepsilon\mid\bw}(\be)\dif \be,\quad \bw\in\cW,
\]
by Lemma \ref{lem: regular conditional probability}, we see that
\[
\int_{\cW}q(\bw)P_{\bW}\left(\dif \bw\right)=\E\left[g\left(\bW,\bepsilon\right)\right]<\infty.
\]
Now, observe that
\[
\cW_0=\widetilde{\cW}_0:=\left\{\bw\in\cW\middle|q(\bw)=\infty,\,\|\Delta\bx\|_{2}>0\right\} .
\]
In turn, introducing sets
\[
\cW_{m}=\left\{ \bw\in\cW\colon\|\Delta\bx\|_{2}>1/m\right\} ,\quad m\in\N,
\]
we have
\[
\int_{\cW_{m}}q(\bw)P_{\bW}\left(\dif \bw\right)\leqslant\int_{\cW}q(\bw)P_{\bW}\left(\dif \bw\right)<\infty.
\]
Therefore, $P_{\bW}(\widetilde{\cW}_0\cap\cW_{m})=0$ for all $m\in\N$. Since
$(\widetilde{\cW}_0\cap\cW_{m})\uparrow\widetilde{\cW}_0$ as $m\to\infty$, we that
$P_{\bW}(\cW_0)=P_{\bW}(\widetilde{\cW}_0)=0$, as desired.
\end{proof}

\begin{proof}[\sc{Proof of Lemma \ref{lem:ExpectedTrimmedSqLossDerivative}}]
Fix $\bw\in\cW'$. 
As $\bw$ is held fixed throughout the proof, abbreviate $v_{1}:=v_{1}(\bw)$ and
$v_{2}:=v_{2}(\bw)$. Also, for notational convenience, abbreviate
$M^{\tt{tls}}(\cdot,\bw)$ by $M(\cdot,\bw)$. Decompose $M(t,\bw)$ as
\begin{align*}
M(t,\bw)&=
  \E\left[\left(-Y_{1}\right)\mathbf{1}\left\{ t\leqslant-Y_{2}\right\} \middle|\bW=\bw\right]\tag{\ensuremath{=:M_{1}(t,\bw)}}\\
  &\quad+\E\left[\left(t+Y_{2}-Y_{1}\right)\mathbf{1}\{ t\in\left(-Y_{2},Y_{1}\right)\} \middle|\bW=\bw\right]\tag{\ensuremath{=:M_{2}(t,\bw)}}\\
  &\quad+\E\left[Y_{2}\mathbf{1}\{ t\geqslant Y_{1}\} \middle|\bW=\bw\right].\tag{\ensuremath{=:M_{3}(t,\bw)}}
\end{align*}
As we demonstrate below, both $M_{1}(t,\bw)$ and $M_{3}(t,\bw)$ cancel against
terms in $M_{2}(t,\bw)$, so we focus on the latter. Further split $M_{2}$ into
the three parts (Parts $a, b$ and $c$):
\begin{align*}
M_{2}(t,\bw) 
  & =t\E\left[\mathbf{1}\{ t\in(-Y_{2},Y_{1})\} \middle|\bW=\bw\right]\tag{\ensuremath{=:M_{2,a}(t,\bw)}}\\
  & \quad+\E\left[Y_{2}\mathbf{1}\{ t\in(-Y_{2},Y_{1})\} \middle|\bW=\bw\right]\tag{\ensuremath{=:M_{2,b}(t,\bw)}}\\
  & \quad+\E\left[\left(-Y_{1}\right)\mathbf{1}\{ t\in(-Y_{2},Y_{1})\} \middle|\bW=\bw\right].\tag{\ensuremath{=:M_{2,c}(t,\bw)}}
\end{align*}
The indicator appearing in all three parts can be written as
\begin{align*}
\mathbf{1}\{ t\in\left(-Y_{2},Y_{1}\right)\}  & =\mathbf{1}\{ Y_{1}>t\} \cdot\mathbf{1}\{ Y_{2}>-t\} \\
  & =\left[1-\mathbf{1}\{ Y_{1}\leqslant t\} \right]\left[1-\mathbf{1}\{ Y_{2}\leqslant-t\} \right]\\
  & =1-\mathbf{1}\{ Y_{1}\leqslant t\} -\mathbf{1}\{ Y_{2}\leqslant-t\} +\mathbf{1}\{ Y_{1}\leqslant t\}\mathbf{1}\{ Y_{2}\leqslant-t\} .
\end{align*}
Since $Y_1,Y_2\geqslant 0$, the event $\{Y_1\leqslant t,\,Y_2\leqslant -t\}$ is
empty unless $t=0$, in which case it reduces to $\{Y_1=0,\,Y_2=0\}$. Hence,
\[
\mathbf{1}\{ Y_{1}\leqslant t\}\mathbf{1}\{ Y_{2}\leqslant-t\} = \mathbf{1}\{ t=0\}\mathbf{1}\{ Y_{1}=0\}\mathbf{1}\{ Y_{2}=0\} .
\]
It follows that the product $(t+Y_{2}-Y_{1})\mathbf{1}\{ Y_{1}\leqslant
t\}\mathbf{1}\{ Y_{2}\leqslant-t\}$ is identically zero, so that it can be
ignored in the derivation of Parts $a,b$ and $c$, which we turn to next.

For \textbf{Part a}, we get
\begin{align*}
M_{2,a}(t,\bw) 
  & =t\E\left[1-\mathbf{1}\{ Y_{1}\leqslant t\} -\mathbf{1}\{ Y_{2}\leqslant-t\} \middle|\bW=\bw\right]\\
  & =t\left(1-\E\left[\mathbf{1}\left\{ Y_{1}\leqslant t\right\} \middle|\bW=\bw\right]-\E\left[\mathbf{1}\left\{ Y_{2}\leqslant-t\right\} \middle|\bW=\bw\right]\right).
\end{align*}
The right-hand side expectations are, respectively,
\begin{align*}
\E\left[\mathbf{1}\{ Y_{1}\leqslant t\} \middle|\bW=\bw\right] 
  & =\E\left[\mathbf{1}\{ \max\left\{ 0,\varepsilon_{1}-v_{1}\right\} \leqslant t\} \middle|\bW=\bw\right]\\
  & =\mathbf{1}\{ t\geqslant0\}\E\left[\mathbf{1}\{ \varepsilon_{1}\leqslant v_{1}+t\} \middle|\bW=\bw\right]=\mathbf{1}\left\{ t\geqslant0\right\} F_{\varepsilon\mid\bw}\left(v_{1}+t\right)
\end{align*}
and
\begin{align*}
\E\left[\mathbf{1}\{ Y_{2}\leqslant-t\} \middle|\bW=\bw\right] & =\E\left[\mathbf{1}\{ \max\left\{ 0,\varepsilon_{2}-v_{2}\right\} \leqslant-t\}\middle|\bW=\bw\right]\\
  & =\mathbf{1}\left\{ t\leqslant0\right\}\E\left[\mathbf{1}\{ \varepsilon_{2}\leqslant v_{2}-t\} \middle|\bW=\bw\right]=\mathbf{1}\{ t\leqslant0\} F_{\varepsilon\mid\bw}(v_{2}-t),
\end{align*}
where we have used that $\varepsilon_{1}$ and $\varepsilon_{2}$ are
(conditionally) identically distributed (Assumption
\ref{assu:Exchangeability-TLS}).
\[
M_{2,a}(t,\bw)=t\left[1-\mathbf{1}\{ t\geqslant0\} F_{\varepsilon\mid\bw}(v_{1}+t)-\mathbf{1}\{ t\leqslant0\} F_{\varepsilon\mid\bw}(v_{2}-t)\right].
\]

For \textbf{Part b}, we get
\begin{align*}
M_{2,b}(t,\bw) & =\E\left[Y_{2}\left(1-\mathbf{1}\{ Y_{1}\leqslant t\} -\mathbf{1}\{ Y_{2}\leqslant-t\} \right)\middle|\bW=\bw\right]\\
  & =\E\left[Y_{2}-Y_{2}\mathbf{1}\{ Y_{1}\leqslant t\} -Y_{2}\mathbf{1}\{ Y_{2}\leqslant-t\} \middle|\bW=\bw\right]\\
  & =\E\left[Y_{2}\middle|\bW=\bw\right]-\E\left[Y_{2}\mathbf{1}\{ Y_{1}\leqslant t\} \middle|\bW=\bw\right]-\E\left[Y_{2}\mathbf{1}\{ Y_{2}\leqslant-t\} \middle|\bW=\bw\right].
\end{align*}
The \emph{first }term on the right-hand side does not depend on $t$,
and the \emph{second} is $-M_{3}(t,\bw)$. The remaining
right-hand side expectation is
\begin{align*}
\E\left[Y_{2}\mathbf{1}\left\{ Y_{2}\leqslant-t\right\}\middle|\bW=\bw\right]
  & =\E\left[\max\left\{ 0,\varepsilon_{2}-v_{2}\right\} \mathbf{1}\{ \max\left\{ 0,\varepsilon_{2}-v_{2}\right\} \leqslant-t\}\middle|\bW=\bw\right]\\
  & =\mathbf{1}\{ t\leqslant0\}\E\left[\max\left\{ 0,\varepsilon_{2}-v_{2}\right\} \mathbf{1}\left\{ \varepsilon_{2}\leqslant v_{2}-t\right\}\middle|\bW=\bw\right]\\
  & =\mathbf{1}\{ t\leqslant0\}\int_{v_{2}}^{v_{2}-t}(u-v_{2})f_{\varepsilon\mid\bw}(u)\dif u,
\end{align*}
so
\[
M_{2,b}(t,\bw)=\E\left[Y_{2}\middle|\bW=\bw\right]-M_{3}(t,\bw)-\mathbf{1}\{ t\leqslant0\}\int_{v_{2}}^{v_{2}-t}(u-v_{2})f_{\varepsilon\mid\bw}(u)\dif u.
\]

For \textbf{Part c}, we get
\begin{align*}
M_{2,c}(t,\bw) & =\E\left[\left(-Y_{1}\right)\left(1-\mathbf{1}\left\{ Y_{1}\leqslant t\right\} -\mathbf{1}\left\{ Y_{2}\leqslant-t\right\} \right)\middle|\bW=\bw\right]\\
  & =\E\left[Y_{1}\mathbf{1}\left\{ Y_{1}\leqslant t\right\} +Y_{1}\mathbf{1}\left\{ Y_{2}\leqslant-t\right\} -Y_{1}\middle|\bW=\bw\right]\\
  & =\E\left[Y_{1}\mathbf{1}\left\{ Y_{1}\leqslant t\right\} \middle|\bW=\bw\right]+\E\left[Y_{1}\mathbf{1}\left\{ Y_{2}\leqslant-t\right\}\middle|\bW=\bw\right]-\E\left[Y_{1}\middle|\bW=\bw\right].
\end{align*}
The \emph{second} term on the right is $-M_{1}(t,\bw)$
and the \emph{third} does not depend on $t$. The remaining right-hand
side expectation is
\begin{align*}
\E\left[Y_{1}\mathbf{1}\{ Y_{1}\leqslant t\} \middle|\bW=\bw\right] & =\E\left[\max\left\{ 0,\varepsilon_{1}-v_{1}\right\} \mathbf{1}\{ \max\left\{ 0,\varepsilon_{1}-v_{1}\right\} \leqslant t\} \middle|\bW=\bw\right]\\
  & =\mathbf{1}\{ t\geqslant0\} \E\left[\max\left\{ 0,\varepsilon_{1}-v_{1}\right\} \mathbf{1}\{ \varepsilon_{1}\leqslant v_{1}+t\}\middle|\bW=\bw\right]\\
  & =\mathbf{1}\{ t\geqslant0\} \int_{v_{1}}^{v_{1}+t}(u-v_{1})f_{\varepsilon\mid\bw}(u)\dif u,
\end{align*}
so
\[
M_{2,c}(t,\bw)=\mathbf{1}\{ t\geqslant0\}\int_{v_{1}}^{v_{1}+t}(u-v_{1})f_{\varepsilon\mid\bw}(u)\dif u-M_{1}(t,\bw)-\E\left[Y_{1}\middle|\bW=\bw\right].
\]
Since $M\equiv M_{1}+M_{2}+M_{3}$ and $M_{2}\equiv M_{2,a}+M_{2,b}+M_{2,c}$,
collecting terms, we see that
\begin{align*}
M(t,\bw) & =t\left[1-\mathbf{1}\{ t\geqslant0\} F_{\varepsilon\mid\bw}(v_{1}+t)-\mathbf{1}\{ t\leqslant0\} F_{\varepsilon\mid\bw}(v_{2}-t)\right]\\
  & \quad-\mathbf{1}\{ t\leqslant0\}\int_{v_{2}}^{v_{2}-t}(u-v_{2})f_{\varepsilon\mid\bw}(u)\dif u\\
  & \quad+\mathbf{1}\{ t\geqslant0\}\int_{v_{1}}^{v_{1}+t}(u-v_{1})f_{\varepsilon\mid\bw}(u)\dif u+\mu_{Y_{2}-Y_{1}}(\bw),
\end{align*}
where $\mu_{Y_{2}-Y_{1}}(\bw)=\E\left[Y_{2}-Y_{1}\middle|\bW=\bw\right]$.
Expressed piecewise, we get
\[
M(t,\bw)=\mu_{Y_{2}-Y_{1}}(\bw)+\begin{cases}
t\left[1-F_{\varepsilon\mid\bw}(v_{2}-t)\right]-\int_{v_{2}}^{v_{2}-t}(u-v_{2})f_{\varepsilon\mid\bw}(u)\dif u, & t<0,\\
0, & t=0,\\
t\left[1-F_{\varepsilon\mid\bw}(v_{1}+t)\right]+\int_{v_{1}}^{v_{1}+t}(u-v_{1})f_{\varepsilon\mid\bw}(u)\dif u, & t>0.
\end{cases}
\]
For $t<0$, integration by parts gives
\[
\int_{v_2}^{v_2-t}(u-v_2)f_{\varepsilon\mid\bw}(u)\dif u
= \left[(u-v_2)F_{\varepsilon\mid\bw}(u)\right]_{v_2}^{v_2-t}
-\int_{v_2}^{v_2-t}F_{\varepsilon\mid\bw}(u)\dif u,
\]
and substituting this identity yields \eqref{eq:ExpectedLossDerivative}; the
case $t>0$ is analogous.
\end{proof}

\begin{proof}[\sc{Proof of Lemma \ref{lem:DifferentiabilityOfExpectedTrimmedSqLossDerivative}}]
Fix $\bw\in\cW'$. 
As in the proof of Lemma \ref{lem:ExpectedTrimmedSqLossDerivative}, abbreviate
$v_{1}:=v_{1}(\bw)$, $v_{2}:=v_{2}(\bw)$ and
$M(\cdot,\bw):=M^{\tt{tls}}(\cdot,\bw)$.

\ref{enu:MLipschitz}.~The function $\dot{m}_{1}^{\tt{tls}}(\cdot,\by)$ defined
in \eqref{eq:TrimmedSqLossDerivative} is Lipschitz continuous with Lipschitz
constant equal to one regardless of $\by\in[0,\infty)\times[0,\infty)$, so $M(\cdot,\bw)$
inherits these properties via Jensen's inequality (conditional on $\bW=\bw$).

\ref{enu:MDiffAwayFromZero}.~From the expression \eqref{eq:ExpectedLossDerivative} for
$M(\cdot,\bw)$, continuity of $F_{\varepsilon\mid\bw}(\cdot)$ and the
fundamental theorem of calculus (and chain rule) imply that $M(\cdot,\bw)$ is
differentiable at every $t\neq0$ with the derivatives taking the form in
\eqref{eq:DerivativeExpectedLoss}. 

\ref{enu:MDirDiffAtZero}.~To show the semi-differentiability of $M(\cdot,\bw)$
at zero, consider first a sequence $\{t_m\}_{m=1}^\infty$ in $(0,\infty)$
converging to zero from above ($t_m\to0_+$). Fix $\epsilon>0$. Continuity of
$F_{\varepsilon\mid\bw}(\cdot)$ ensures that there is a $\delta>0$ such that
$|u-v_1|\leqslant\delta$ implies
$|F_{\varepsilon\mid\bw}(u)-F_{\varepsilon\mid\bw}(v_1)|\leqslant\epsilon$. As
$t_m\to0_+$, for $m$ large enough we have $0<t_m\leqslant\delta$, so that
\begin{align*}
\left|\frac{1}{t_m}\int_{v_{1}}^{v_{1}+t_m}F_{\varepsilon\mid\bw}(u)\dif u-F_{\varepsilon\mid\bw}(v_{1})\right|
&=\left|\frac{1}{t_m}\int_{v_{1}}^{v_{1}+t_m}\left[F_{\varepsilon\mid\bw}(u)-F_{\varepsilon\mid\bw}(v_{1})\right]\dif u\right|\\
&\leqslant\frac{1}{t_m}\int_{v_{1}}^{v_{1}+t_m}\left|F_{\varepsilon\mid\bw}(u)-F_{\varepsilon\mid\bw}(v_{1})\right|\dif u\\
&\leqslant\frac{1}{t_m}\int_{v_{1}}^{v_{1}+t_m}\epsilon\dif u=\epsilon.
\end{align*}
Since $t_m\to0_+$ and $\epsilon>0$ were arbitrary, we have shown that
$\lim_{t\to0_+}(1/t)\int_{v_{1}}^{v_{1}+t}F_{\varepsilon\mid\bw}(u)\dif u=F_{\varepsilon\mid\bw}(v_{1})$,
from which we get the \emph{right} differentiability at zero with right
derivative given by
\[
\lim_{t\to0_+}\frac{M(t,\bw)-M(0,\bw)}{t}=\lim_{t\to0_+}\left(1-\frac{1}{t}\int_{v_{1}}^{v_{1}+t}F_{\varepsilon\mid\bw}(u)\dif u\right)=1-F_{\varepsilon\mid\bw}(v_{1}).
\]
\emph{Left} differentiability at zero with the claimed left derivative follows
analogously.

\ref{enu:MDiffatxtheta}.~If $\Delta\bx^\top\btheta_0\neq 0$, then $M(\cdot,\bw)$
is differentiable at $t=\Delta\bx^\top\btheta_0$ by Item
\ref{enu:MDiffAwayFromZero}, and \eqref{eq:DerivativeExpectedLossAtxtheta0}
follows by substituting $t=\Delta\bx^\top\btheta_0$ into
\eqref{eq:DerivativeExpectedLoss}. If instead $\Delta\bx^\top\btheta_0=0$, then
$v_1(\bw)=v_2(\bw)$ and the left and right derivatives in
\eqref{eq:LeftAndRightDerivativesExpectedLossAtZero} coincide, so $M(\cdot,\bw)$
is differentiable at $t=\Delta\bx^\top\btheta_0=0$ with derivative given in
\eqref{eq:DerivativeExpectedLossAtxtheta0}.
\end{proof}

\subsection{Proof of Theorem \ref{thm:AsymptoticNormality-TLS}}

\begin{proof}[\sc{Proof of Theorem \ref{thm:AsymptoticNormality-TLS}}]
As in \citet{honore_trimmed_1992}, we set up for an application of
\citet[Theorem 3.3]{pakes1989simulation}. Following the proof of \citet[Theorem
2(iv)]{honore_trimmed_1992}, we verify all conditions of \citet[Theorem
3.3]{pakes1989simulation} except for their condition (ii). Our Theorem
\ref{thm:JacobianExistence-TLS} and Assumption
\ref{assu:JacobianInvertibility-TLS} combine to show that $\bG$ is
differentiable at $\btheta_0$ with invertible Jacobian, which is precisely the
desired condition (ii).
\end{proof}

\subsection{Proof of Theorem \ref{thm:VarianceConsistency-TLS}}

\begin{proof}[\sc{Proof of Theorem \ref{thm:VarianceConsistency-TLS}}]
Throughout the proof, we drop the superscripts from the estimators
$\widehat\btheta^{\tt{tls}}$, $\widehat\bfV^{\tt{tls}}$ and
$\widehat\bfJ^{\tt{tls}}$ for notational convenience. We next discuss the proof for
$\widehat\bfV$ and $\widehat\bfJ$, in turn.

\citet[pp.~1043--1044]{pakes1989simulation} give a strategy for proving
consistency of an estimator of $\bfV_0$ of the plug-in form. This strategy
relies on verification of the assumptions in \citet[Lemma
2.17]{pakes1989simulation}. These assumptions are for TLS all verified on
\citet[p.~564]{honore_trimmed_1992} as part of the proof of condition (iii) for
asymptotic normality. Since $\widehat\bfV$ is of the plug-in form, strong
consistency follows from the calculation on
\citet[pp.~1043--1044]{pakes1989simulation} and the strong consistency of TLS as
established in \citet[Theorem 1(iv)]{honore_trimmed_1992}. 

Define the alternative estimators of $\bfJ_0$,
\begin{align*}
\widehat\bfJ_1&:=\frac{1}{n}\sum_{i=1}^n\del[2]{
  \mathbf{1}\{Y_{i1}>0\}\mathbf{1}\{\Delta\bX_i^\top\widehat\btheta^{\tt{tls}}\leqslant0\}
 +\mathbf{1}\{Y_{i2}>0\}\mathbf{1}\{\Delta\bX_i^\top\widehat\btheta^{\tt{tls}}>0\}}\Delta\bX_i\Delta\bX_i^\top
\end{align*}
and
\begin{align*}
\widehat\bfJ_2&:=\frac{1}{n}\sum_{i=1}^n\del[2]{
  \mathbf{1}\{Y_{i1}>0\}\mathbf{1}\{\Delta\bX_i^\top\widehat\btheta^{\tt{tls}}<0\}
 +\mathbf{1}\{Y_{i2}>0\}\mathbf{1}\{\Delta\bX_i^\top\widehat\btheta^{\tt{tls}}\geqslant0\}}\Delta\bX_i\Delta\bX_i^\top,
\end{align*}
which are based on the equivalent expressions for $\bfJ_0$ given in
\eqref{eq:JacobianAlt1} and \eqref{eq:JacobianAlt2}, respectively. Note that
$\widehat\bfJ$ is the simple average of $\widehat\bfJ_1$ and $\widehat\bfJ_2$,
so that consistency of $\widehat\bfJ$ will follow from that of $\widehat\bfJ_1$
and $\widehat\bfJ_2$. We next argue the consistency of $\widehat\bfJ_1$; the
argument for $\widehat\bfJ_2$ is analogous.

For $\widehat\bfJ_1$, it suffices to show that each entry of $\widehat\bfJ_1$
converges in probability to the corresponding entry of $\bfJ_0$. Fix therefore
$(j,k)\in[K]\times[K]$, and consider the function class defined by
\begin{align*}
  \cF&:=\left\{f:\cZ\to\R\middle| f=f(\cdot,\btheta),\btheta\in\R^K\right\},\\
  f(\bz;\btheta)&:=\del[1]{\mathbf{1}\{y_1>0\}\mathbf{1}\{\Delta\bx^\top\btheta\leqslant0\}+\mathbf{1}\{y_2>0\}\mathbf{1}\{\Delta\bx^\top\btheta>0\}}\Delta x_j \Delta x_k,
\end{align*} 
with $\Delta \bx :=\bx_{1}-\bx_{2}$. We here employ linear functional notation
familiar from the empirical process literature, so that $P_n f(\cdot;\btheta)$
denotes the empirical average of $f$ over the sample $\{\bZ_i\}_{i=1}^n$ and $P
f(\cdot;\btheta)$ denotes the integral of $f(\cdot;\btheta)$ against the
distribution of $\bZ$. 

The consistency argument is divided into three steps: In Step 1 we establish a
uniform law of large numbers (ULLN) for $\cF$. In Step 2 we show that $P
f(\cdot;\btheta)\to Pf(\cdot;\btheta_0)$ as $\btheta\to\btheta_0$. In Step 3 we
use the previous two steps and the fact that $\widehat J_{1,j,k}$ is of the
(plug-in) form $P_n f(\cdot;\widehat\btheta)$, to argue the consistency
of $\widehat J_{1,j,k}$ from that of $\widehat\btheta$. 

\textbf{Step 1:} Aiming towards a ULLN, we assume that the reader is familiar
with the notions of a Vapnik-\u{C}ervonenkis (VC) class of sets and a
VC-subgraph class of functions as defined in \citet[Section
2]{pakes1989simulation} or \citet[Section 2.6]{vdVW1996weak}.\footnote{While the
two references attach different meanings to the notion of a ``subgraph,'' they
lead to equivalent definitions of a VC-subgraph class of functions, cf.
\citet[p.~141]{vdVW1996weak}.} Define the function class
\begin{align*}
  \cG&:=\left\{g:\cZ\times\R\to\R\middle|g=g(\cdot,\cdot;\gamma,\gamma_1,\gamma_2,\bdelta),(\gamma,\gamma_1,\gamma_2,\bdelta)\in\R^{3+K}\right\},\\
  g(\bz,s;\gamma,\gamma_1,\gamma_2,\bdelta)&:=\gamma s+\gamma_1y_1+\gamma_2y_2+\Delta\bx^\top \bdelta.
\end{align*}
Then $\cG$ forms a vector space of real-valued measurable functions of dimension
$3+K$. \citet[Lemma 2.6.15]{vdVW1996weak} shows that the class $\cG$ is
VC-subgraph (of VC index at most $5+K$), where the \emph{subgraph} of any function
$f:\cZ\to\R$ is defined as the area below its graph:
\[
\mathrm{subgraph}(f):=\left\{(\bz,s)\in\cZ\times\R\middle|s<f(\bz)\right\}.
\]
It then follows from \citet[Lemma 2.4]{pakes1989simulation} that the sets of the
form $\{g\geqslant r\}$ or $\{g>r\}$ with $g\in\cG$ and $r\in\R$ form a VC
class. Call it $\cC$. Consider next the function class $\cH$ defined by
\begin{align*}
\cH&:=\left\{h:\cZ\to\R\middle| h=h(\cdot;\btheta),\btheta\in\R^K\right\}\\
h(\bz;\btheta)&:=\mathbf{1}\{y_1>0\}\mathbf{1}\{\Delta\bx^\top\btheta\leqslant0\}+\mathbf{1}\{y_2>0\}\mathbf{1}\{\Delta\bx^\top\btheta>0\}.
\end{align*}
For any $\btheta\in\R^K$, the subgraph of $h(\cdot,\btheta)$ can be expressed as
\begin{align*}
\mathrm{subgraph}(h(\cdot;\btheta))
&=\left(\{y_1>0\}\cap\{\Delta\bx^\top\btheta>0\}^c\cap\{s\geqslant1\}^c\right)\\
&\quad\cup\left(\{y_2>0\}\cap\{\Delta\bx^\top\btheta>0\}\cap\{s\geqslant1\}^c\right)\cup\{s\geqslant 0\}^c\\
&=\left(\{g_1>0\}\cap\{g_2>0\}^c\cap\{g_3\geqslant1\}^c\right)\\
&\quad\cup\left(\{g_4>0\}\cap\{g_2>0\}\cap\{g_3\geqslant1\}^c\right)\cup\{g_3\geqslant 0\}^c,
\end{align*}
for appropriate choices of $g_1,\ldots,g_4\in\cG$. The \emph{first} part of the
above union is the intersection of three sets, one of which lies in $\cC$ and
two of which are complements of sets in $\cC$. The \emph{second} part of the
union is the intersection of three sets, two of which lie in $\cC$ and one of
which is the complement of a set in $\cC$. The \emph{third} part of the
union is the complement of a set in $\cC$. It therefore follows from the
permanence properties in \citet[Lemma 2.5]{pakes1989simulation} that the
subgraphs
$
\{\mathrm{subgraph}(h(\cdot;\btheta))\mid\btheta\in\R^K\}
$
form a VC class, meaning that $\cH$ is VC-subgraph. Consider the (fixed)
function $g_{j,k}:\cZ\to\R$ given by
$
  g_{j,k}(\bz):=\Delta x_j \Delta x_k.
$
Then \citet[Lemma 2.6.18]{vdVW1996weak} implies that $\cF=\cH\cdot
  g_{j,k}=\{\bz\mapsto h(\bz)g_{j,k}(\bz)\mid h\in\cH\}$ is VC-subgraph. The
  function $|g_{j,k}|$ is an envelope for $\cF$, which is integrable by
  Assumption \ref{assu:MomentConditions-TLS}.
  \citet[Lemmas 2.8 and 2.12]{pakes1989simulation}
combine to show a strong ULLN for $\cF$:
$
\sup_{f\in\cF}|(P_n-P)f|\to_{\text{a.s.}}0.
$

\textbf{Step 2:} Let $\{\btheta_m\}_{m=1}^\infty$ be a sequence in $\R^K$
converging to $\btheta_0$ as $m\to\infty$. We set up for an application of the
Lebesgue Dominated Convergence Theorem (LDCT). To this end, iterate expectations
to write $Pf(\cdot;\btheta)$ as
\[
Pf(\cdot;\btheta)=\E\sbr[2]{\del[2]{\bbP(Y_1>0\mid\bW)\mathbf{1}\{\Delta \bX^\top\btheta\leqslant0\}+\bbP(Y_2>0\mid\bW)\mathbf{1}\{\Delta\bX^\top\btheta>0\}}\Delta X_j \Delta X_k}.
\]
Consider the sequence of (integrand) functions $f_m,m\in\N$, defined on $\Omega$
by
\begin{align*}
  f_m(\omega)
  &:=\left(\bbP(Y_1>0\mid\bW)(\omega)\mathbf{1}\{\Delta \bX(\omega)^\top\btheta_m\leqslant0\}+\bbP(Y_2>0\mid\bW)(\omega)\mathbf{1}\{\Delta\bX(\omega)^\top\btheta_m>0\}\right)\\
  &\qquad \times\Delta X_j(\omega) \Delta X_k(\omega).
\end{align*}
The $f_m$ are measurable and bounded by the $\bbP$-integrable $|g_{j,k}(\bZ)|$.
Fix $\omega\in\Omega$. We consider the two cases: (i) $\Delta
\bX(\omega)^\top\btheta_0\neq0$ and (ii) $\Delta \bX(\omega)^\top\btheta_0=0$,
in turn. \emph{Case (i):} If $\Delta\bX(\omega)^\top\btheta_0\neq0$, then the
signs of $\Delta \bX(\omega)^\top\btheta_m$ and $\Delta
\bX(\omega)^\top\btheta_0$ eventually agree, so that, in particular,
\begin{align*}
  f_m(\omega)&\to \left(\bbP(Y_1>0\mid\bW)(\omega)\mathbf{1}\{\Delta \bX(\omega)^\top\btheta_0\leqslant0\}+\bbP(Y_2>0\mid\bW)(\omega)\mathbf{1}\{\Delta\bX(\omega)^\top\btheta_0>0\}\right)\\
  &\qquad \times\Delta X_j(\omega) \Delta X_k(\omega).
\end{align*}
\emph{Case (ii):} If $\Delta \bX(\omega)^\top\btheta_0=0$, then the conditional
exchangeability of $\varepsilon_1$ and $\varepsilon_2$ (Assumption
\ref{assu:Exchangeability-TLS}) implies that the conditional probabilities
$\bbP(Y_1>0\mid\bW)(\omega)$ and $\bbP(Y_2>0\mid\bW)(\omega)$ are equal. In this
case, 
\begin{align*}
  f_m(\omega)
  &:=\left(\bbP(Y_1>0\mid\bW)(\omega)\mathbf{1}\{\Delta \bX(\omega)^\top\btheta_m\leqslant0\}+\bbP(Y_2>0\mid\bW)(\omega)\mathbf{1}\{\Delta\bX(\omega)^\top\btheta_m>0\}\right)\\
  &\qquad \times\Delta X_j(\omega) \Delta X_k(\omega)\\
  &=\bbP(Y_1>0\mid\bW)(\omega)\left(\mathbf{1}\{\Delta \bX(\omega)^\top\btheta_m\leqslant0\}+\mathbf{1}\{\Delta\bX(\omega)^\top\btheta_m>0\}\right)\Delta X_j(\omega) \Delta X_k(\omega)\\
  &=\bbP(Y_1>0\mid\bW)(\omega)\Delta X_j(\omega) \Delta X_k(\omega)\\
  &=\bbP(Y_1>0\mid\bW)(\omega)\mathbf{1}\{\Delta \bX(\omega)^\top\btheta_0\leqslant0\}\Delta X_j(\omega) \Delta X_k(\omega)\\
  &=\left(\bbP(Y_1>0\mid\bW)(\omega)\mathbf{1}\{\Delta \bX(\omega)^\top\btheta_0\leqslant0\}+\bbP(Y_2>0\mid\bW)(\omega)\mathbf{1}\{\Delta\bX(\omega)^\top\btheta_0>0\}\right)\\
  &\qquad \times\Delta X_j(\omega) \Delta X_k(\omega).
\end{align*}
It follows that $f_m$ converges pointwise to $f$ defined on $\Omega$ by
\begin{align*}
  f(\omega)&:=\left(\bbP(Y_1>0\mid\bW)(\omega)\mathbf{1}\{\Delta \bX(\omega)^\top\btheta_0\leqslant0\}+\bbP(Y_2>0\mid\bW)(\omega)\mathbf{1}\{\Delta\bX(\omega)^\top\btheta_0>0\}\right)\\
  &\qquad \times\Delta X_j(\omega) \Delta X_k(\omega).
\end{align*}
The LDCT therefore goes to show that $P f(\cdot,\btheta)\to Pf(\cdot,\btheta_0)$
as $\btheta\to\btheta_0$, as desired. 

\textbf{Step 3:} The triangle inequality shows that
\begin{align*}
  \envert[1]{\widehat J_{1,j,k}-J_{0,j,k}}
  &=\envert[1]{\del[0]{P_n-P} f(\cdot;\widehat\btheta)+P\sbr[0]{f(\cdot;\widehat\btheta)-f(\cdot;\btheta_0)}}\\
  &\leqslant\envert[1]{\del[0]{P_n-P} f(\cdot;\widehat\btheta)}+\envert[1]{P\sbr[0]{f(\cdot;\widehat\btheta)-f(\cdot;\btheta_0)}}\\
  &\leqslant \sup_{f\in\cF}\envert[1]{\del[0]{P_n-P} f}+\envert[1]{P\sbr[0]{f(\cdot;\widehat\btheta)-f(\cdot;\btheta_0)}}. 
\end{align*}
The first term on the right-hand side converges to zero almost surely by Step 1.
The second term converges to zero almost surely by strong consistency of TLS
\citep[Theorem 1(iv)]{honore_trimmed_1992}, Step 2, and the continuous mapping
theorem. It follows from the previous display that $\widehat
J_{1,j,k}\to_{\text{a.s.}} J_{0,j,k}$, which finishes the proof of the (strong)
consistency of $\widehat\bfJ_1$. The (strong) consistency of $\widehat\bfJ_2$
follows by parallel reasoning, and the (strong) consistency of $\widehat\bfJ$
in turn follows from the continuous mapping theorem.
\end{proof}

\subsection{Proof of Theorem \ref{thm:ConsistencyHonoreHessianEstimator-TLS}}
\begin{proof}[\textsc{Proof of Theorem \ref{thm:ConsistencyHonoreHessianEstimator-TLS}}]
After some preliminary observations (Step 0), in Step 1 we show that the leading
term satisfies
$\widehat\bfL^{\tt{H92}}=\bfJ_0^{\tt{tls}}+o_{L^1}(1)+o_{\mathrm{a.s.}}(1)$, and
in Step 2 that the remainder term satisfies
$\widehat\bfR^{\tt{H92}}=\bzero_{K\times K}+o_{L^1}(1)+o_{\mathrm{a.s.}}(1)$. It
then follows from our decomposition of $\widehat\bGamma^{\tt{H92}}$ that
$\widehat\bGamma^{\tt{H92}}=\bfJ_0^{\tt{tls}}+o_{L^1}(1)+o_{\mathrm{a.s.}}(1)$.
To ease notation, we drop the superscripts ``H92'' and ``tls'' in the rest of
the proof.

\textbf{Step 0 (Preliminaries):}
By conditional stationarity (implied by Assumption
\ref{assu:Exchangeability-TLS}),
\begin{align*}
F_{Y_{\tau}\mid\bw}\left(y\right) & :=\bbP\left(Y_{\tau}\leqslant y\middle|\bW=\bw\right)=\bone\left\{ y\geqslant0\right\} F_{\varepsilon\mid\bw}(y-a-\bx_{\tau}^{\top}\btheta_{0}).
\end{align*}
Note that continuity of $F_{\varepsilon\mid\bw}\left(\cdot\right)$
(implied by Assumption \ref{assu:continuity tls}) implies continuity of each
$F_{Y_{\tau}\mid\bw}(\cdot)$ on $(0,\infty)$. As observed by
\citet{honore_trimmed_1992}, conditional stationarity creates a certain
symmetry. Specifically, for $y=0$, we see that
\begin{equation}
\begin{cases}
F_{Y_{1}\mid\bw}\left(0\right)=F_{Y_{2}\mid\bw}\left(-\Delta\bx^{\top}\btheta_{0}\right) & \text{if }\Delta\bx^{\top}\btheta_{0}\leqslant0,\\
F_{Y_{2}\mid\bw}\left(0\right)=F_{Y_{1}\mid\bw}\left(\Delta\bx^{\top}\btheta_{0}\right) & \text{if }\Delta\bx^{\top}\btheta_{0}\geqslant0,
\end{cases}\label{eq:OutcomeCDFsSymmetry}
\end{equation}
which we will use below. For brevity, we employ the shorthand:
\begin{itemize}
\item Conditional complementary CDFs are denoted
$\overline{F}_{Y_{\tau}\mid\bw}\left(\cdot\right):=1-F_{Y_{\tau}\mid\bw}\left(\cdot\right)$.
\item Indicator functions are occasionally abbreviated using ``Iverson braces,''
$\left\{ \cdot\right\} :=\bone\{\cdot\}$.
\end{itemize}

\textbf{Step 1 (Leading term):}
We show $\widehat{\bfL}=\bfJ_0+o_{L^1}(1)+o_{\mathrm{a.s.}}(1)$, equivalently,
$\widehat{L}_{j,k}\to J_{0,j,k}$ for each $(j,k)\in[K]\times[K]$. To this end,
fix $\left(j,k\right)\in\left[K\right]\times\left[K\right]$, and consider the
function class
\begin{align*}
\mathcal{L}_{j,k} & :=\left\{ \ell_{j,k}:\mathcal{Z}\to\R\middle|\ell_{j,k}:=\ell_{j,k}\left(\cdot;\btheta\right),\btheta\in\R^{K}\right\} 
\end{align*}
where
\[
\ell_{j,k}(\bz;\btheta):=\del[1]{\bone\{-y_{2}<\Delta\bx^{\top}\btheta<0\}+\bone\{0<\Delta\bx^{\top}\btheta<y_{1}\}}\Delta x_{j}\Delta x_{k}.
\]
With minor modifications, arguments along the lines of the proof of Theorem
\ref{thm:VarianceConsistency-TLS} show that $\mathcal{L}_{j,k}$ is VC-subgraph
with $P$-integrable envelope $\bz\mapsto|\Delta x_{j}\Delta x_{k}|$.
\citet[Lemmas 2.8 and 2.12]{pakes1989simulation} therefore combine to show the
strong ULLN
\[
\sup_{f\in\mathcal{L}_{j,k}}\left|\left(P_{n}-P\right)f\right|\overset{\mathrm{a.s.}}{\to}0,
\]
where, as in the proof of
Theorem \ref{thm:VarianceConsistency-TLS}, we use empirical process notation.
Hence,
\[
\widehat{L}_{j,k}=P_{n}\ell_{j,k}(\cdot;\widehat{\btheta})=P\ell_{j,k}(\cdot;\widehat{\btheta})+o_{\mathrm{a.s.}}\left(1\right)=L_{j,k}(\widehat{\btheta})+o_{\mathrm{a.s.}}\left(1\right).
\]
where we have defined $L_{j,k}:\R^{K}\to\R$ by
\[
L_{j,k}\left(\btheta\right):=P\ell_{j,k}\left(\cdot,\btheta\right)=\E\sbr[1]{\del[1]{\bone\{-Y_{2}<\Delta\bX^{\top}\btheta<0\}+\bone\{0<\Delta\bX^{\top}\btheta<Y_{1}\}}\Delta X_{j}\Delta X_{k}},
\]
with $\bZ:=(Y_{1},\bX_{1},Y_{2},\bX_{2})$ having distribution $P$ and being
independent of the sample and, thus, the sequence of TLS estimators. Iterating
expectations, we get
\begin{align*}
L_{j,k}\left(\btheta\right) & =\E\sbr[1]{\del[1]{\{\Delta\bX^{\top}\btheta<0\}\{-Y_{2}<\Delta\bX^{\top}\btheta\}+\{\Delta\bX^{\top}\btheta>0\}\{\Delta\bX^{\top}\btheta<Y_{1}\}}\Delta X_{j}\Delta X_{k}}\\
 & =\E\sbr[1]{\del[1]{\{\Delta\bX^{\top}\btheta<0\}\overline{F}_{Y_{2}\mid\bW}(-\Delta\bX^{\top}\btheta)+\{\Delta\bX^{\top}\btheta>0\}\overline{F}_{Y_{1}\mid\bW}(\Delta\bX^{\top}\btheta)}\Delta X_{j}\Delta X_{k}}\\
 & =\E\sbr[1]{\del[1]{\{\Delta\bX^{\top}\btheta<0\}\overline{F}_{Y_{2}\mid\bW}(|\Delta\bX^{\top}\btheta|)+\{\Delta\bX^{\top}\btheta>0\}\overline{F}_{Y_{1}\mid\bW}(|\Delta\bX^{\top}\btheta|)}\Delta X_{j}\Delta X_{k}}.
\end{align*}
Using Remark \ref{rem:JacobianAlt1} and iterating expectations, the TLS Hessian
can be written as
\begin{align*}
\mathbf{J}_{0} & =\E\sbr[1]{\del[1]{\{ Y_{1}>0\} \{ \Delta\bX^{\top}\btheta_{0}\leqslant0\} +\{ Y_{2}>0\} \{ \Delta\bX^{\top}\btheta_{0}>0\}}\Delta\bX\Delta\bX^{\top}}\\
 & =\E\sbr[1]{\del[1]{\{\Delta\bX^{\top}\btheta_{0}\leqslant0\} \left[1-F_{Y_{1}\mid\bW}\left(0\right)\right]+\{ \Delta\bX^{\top}\btheta_{0}>0\} \left[1-F_{Y_{2}\mid\bW}\left(0\right)\right]}\Delta\bX\Delta\bX^{\top}}\\
 & =\E\sbr[1]{\del[1]{\{\Delta\bX^{\top}\btheta_{0}\leqslant0\} \overline{F}_{Y_{1}\mid\bW}\left(0\right)+\{\Delta\bX^{\top}\btheta_{0}>0\} \overline{F}_{Y_{2}\mid\bW}\left(0\right)}\Delta\bX\Delta\bX^{\top}}.
\end{align*}
We will show that $L_{j,k}(\widehat{\btheta})\to J_{0,j,k}$
in $L^{1}\left(\bbP\right)$. Inserting the two expressions
and bounding from above, we get
\begin{align*}
 & \E_{\widehat{\btheta}}[|L_{j,k}(\widehat{\btheta})-J_{0,j,k}|]\\
 & =\E_{\widehat{\btheta}}\Big[\Big|\E_{\bW}\sbr[1]{\del[1]{\{\Delta\bX^{\top}\widehat{\btheta}<0\}\overline{F}_{Y_{2}\mid\bW}(|\Delta\bX^{\top}\widehat{\btheta}|)+\{\Delta\bX^{\top}\widehat{\btheta}>0\}\overline{F}_{Y_{1}\mid\bW}(|\Delta\bX^{\top}\widehat{\btheta}|)}\Delta X_{j}\Delta X_{k}}\\
 & \qquad\qquad-\E_{\bW}\sbr[1]{\del[1]{\{\Delta\bX^{\top}\btheta_{0}\leqslant0\} \overline{F}_{Y_{1}\mid\bW}\left(0\right)+\{ \Delta\bX^{\top}\btheta_{0}>0\} \overline{F}_{Y_{2}\mid\bW}\left(0\right)}\Delta X_{j}\Delta X_{k}}\Big|\Big].
\end{align*}
By Jensen's inequality and Fubini--Tonelli theorem, we can bound
from above by
\[
\E_{\widehat{\btheta}}[|L_{j,k}(\widehat{\btheta})-J_{0,j,k}|]\leqslant\int_{\mathcal{W}}g_{n}\mathrm{d}P_{\bW},
\]
where $g_{n}:\mathcal{W}\to[0,\infty)$ is defined as
\[
g_{n}\left(\bw\right):=\left[\int_{\Omega}f_{n}\left(\omega,\bw\right)\bbP\left(\mathrm{d}\omega\right)\right]\left|\Delta x_{j}\Delta x_{k}\right|,
\]
and $f_{n}:\Omega\times\mathcal{W}\to[0,\infty)$ is defined as
\begin{align*}
f_{n}\left(\omega,\bw\right) & :=\big|\{\Delta\bx^{\top}\widehat{\btheta}\left(\omega\right)<0\}\overline{F}_{Y_{2}\mid\bw}(|\Delta\bx^{\top}\widehat{\btheta}\left(\omega\right)|)+\{\Delta\bx^{\top}\widehat{\btheta}\left(\omega\right)>0\}\overline{F}_{Y_{1}\mid\bw}(|\Delta\bx^{\top}\widehat{\btheta}\left(\omega\right)|)\\
 & \qquad-\{ \Delta\bx^{\top}\btheta_{0}\leqslant0\} \overline{F}_{Y_{1}\mid\bw}\left(0\right)-\{ \Delta\bx^{\top}\btheta_{0}>0\} \overline{F}_{Y_{2}\mid\bw}\left(0\right)\big|.
\end{align*}
Note that $0\leqslant f_{n}\leqslant1$ and $0\leqslant g_{n}\leqslant g$ for the
$P_{\bW}$-integrable function $g:\bw\mapsto\left|\Delta x_{j}\Delta
x_{k}\right|$. To conclude that $\int_{\mathcal{W}}g_{n}\mathrm{d}P_{\bW}\to0$
by the LDCT, it suffices to show that $g_{n}\to0$ pointwise on $\mathcal{W}$. To
this end, we split into the three exhaustive cases for $\bw$: \textbf{Case 1:}
$\Delta\bx=\bzero$; \textbf{Case 2:} $\Delta\bx^{\top}\btheta_{0}\neq0$; and,
\textbf{Case 3:} $\Delta\bx\neq\bzero$ and $\Delta\bx^{\top}\btheta_{0}=0$. We
consider each case in turn.

\textbf{\emph{Case 1:}} Consider $\bw$ such that $\Delta\bx=\bzero$.
Then $g_{n}\left(\bw\right)=0$, which $\to0$ trivially.

For Cases 2 and 3, recall that \citet[Theorem 1(iv)]{honore_trimmed_1992} shows
strong consistency of TLS, $\widehat{\btheta}\to_{\mathrm{a.s.}}\btheta_{0}$,
which means that there is a set $A\subseteq\Omega$ such that
$\bbP\left(A\right)=1$ and $\widehat{\btheta}\left(\omega\right)\to\btheta_{0}$
for each $\omega\in A$. Fix such an $A\subseteq\Omega$.

\textbf{\emph{Case 2:}} Consider $\bw$ such that
$\Delta\bx^{\top}\btheta_{0}\neq0$. We split into the \textbf{subcases (a)}
$\Delta\bx^{\top}\btheta_{0}>0$ and \textbf{(b)}
$\Delta\bx^{\top}\btheta_{0}<0$. In either case, we set up for an application of
the bounded convergence theorem (BCT) with the limit integrand being zero
everywhere on $\Omega$.

\textbf{\emph{Case 2(a):}} Fix $\omega\in A$. Then
$\Delta\bx^{\top}\widehat{\btheta}\left(\omega\right)\to\Delta\bx^{\top}\btheta_{0}>0$,
so that $\Delta\bx^{\top}\widehat{\btheta}\left(\omega\right)>0$ eventually, in
which case
\begin{align*}
f_{n}\left(\omega,\bw\right) & = \envert[1]{\overline{F}_{Y_{1}\mid\bw}(|\Delta\bx^{\top}\widehat{\btheta}\left(\omega\right)|)-\overline{F}_{Y_{2}\mid\bw}\left(0\right)}=\envert[1]{F_{Y_{1}\mid\bw}(|\Delta\bx^{\top}\widehat{\btheta}\left(\omega\right)|)-F_{Y_{2}\mid\bw}\left(0\right)}.
\end{align*}
By continuity of $F_{Y_{1}\mid\bw}\left(\cdot\right)$
on $\left(0,\infty\right)$ and $\Delta\bx^{\top}\btheta_{0}>0$,
$F_{Y_{1}\mid\bw}(|\Delta\bx^{\top}\widehat{\btheta}\left(\omega\right)|)\to F_{Y_{1}\mid\bw}(\Delta\bx^{\top}\btheta_{0})$,
so that by (\ref{eq:OutcomeCDFsSymmetry}) we get
\[
f_{n}\left(\omega,\bw\right)\to\envert[1]{F_{Y_{1}\mid\bw}(\Delta\bx^{\top}\btheta_{0})-F_{Y_{2}\mid\bw}\left(0\right)}=0.
\]
Since $\bbP\left(A\right)=1$, the BCT implies that
\[
g_{n}\left(\bw\right)=\int_{\Omega}f_{n}\left(\omega,\bw\right)\bbP\left(\mathrm{d}\omega\right)=\int_{\Omega}\bone\left\{ \omega\in A\right\} f_{n}\left(\omega,\bw\right)\bbP\left(\mathrm{d}\omega\right)\to0.
\]

\textbf{\emph{Case 2(b):}} The argument differs only from Case 2(a) in terms of
labels, so we omit it.

\textbf{\emph{Case 3:}} From $\Delta\bx^{\top}\btheta_{0}=0$
(i.e., $\bx_{1}^{\top}\btheta_{0}=\bx_{2}^{\top}\btheta_{0}$)
one gets identical (conditional) marginal outcome CDFs, $F_{Y_{1}|\bw}(\cdot)=F_{Y_{2}|\bw}(\cdot)$,
so the triangle inequality implies
\begin{align*}
f_{n}\left(\omega,\bw\right) & =\big|\{\Delta\bx^{\top}\widehat{\btheta}\left(\omega\right)<0\}\overline{F}_{Y_{2}\mid\bw}(|\Delta\bx^{\top}\widehat{\btheta}\left(\omega\right)|)\\
 & \quad+\{\Delta\bx^{\top}\widehat{\btheta}\left(\omega\right)>0\}\overline{F}_{Y_{1}\mid\bw}(|\Delta\bx^{\top}\widehat{\btheta}\left(\omega\right)|)-\overline{F}_{Y_{1}\mid\bw}\left(0\right)\big|\\
 & =\big|\{\Delta\bx^{\top}\widehat{\btheta}\left(\omega\right)\neq0\}\overline{F}_{Y_{1}\mid\bw}(|\Delta\bx^{\top}\widehat{\btheta}\left(\omega\right)|)-\overline{F}_{Y_{1}\mid\bw}\left(0\right)\big|\\
 & \leqslant\big|\overline{F}_{Y_{1}\mid\bw}(|\Delta\bx^{\top}\widehat{\btheta}\left(\omega\right)|)-\overline{F}_{Y_{1}\mid\bw}\left(0\right)\big|+\{\Delta\bx^{\top}\widehat{\btheta}\left(\omega\right)=0\}\overline{F}_{Y_{1}\mid\bw}(|\Delta\bx^{\top}\widehat{\btheta}\left(\omega\right)|)\\
 & \leqslant\big|F_{Y_{1}\mid\bw}(|\Delta\bx^{\top}\widehat{\btheta}\left(\omega\right)|)-F_{Y_{1}\mid\bw}\left(0\right)\big|+\{\Delta\bx^{\top}\widehat{\btheta}\left(\omega\right)=0\},
\end{align*}
and, thus,
\begin{align*}
g_{n}\left(\bw\right) & =\int_{\Omega}f_{n}\left(\omega,\bw\right)\bbP\left(\mathrm{d}\omega\right)\leqslant\int_{\Omega}\big|F_{Y_{1}\mid\bw}(|\Delta\bx^{\top}\widehat{\btheta}\left(\omega\right)|)-F_{Y_{1}\mid\bw}\left(0\right)\big|\bbP\left(\mathrm{d}\omega\right)+\bbP(\Delta\bx^{\top}\widehat{\btheta}=0).
\end{align*}
For each $\omega\in A$, right continuity of CDFs and
$|\Delta\bx^{\top}\widehat{\btheta}(\omega)|\to|\Delta\bx^{\top}\btheta_0|=0$
imply that $F_{Y_{1}\mid\bw}(|\Delta\bx^{\top}\widehat{\btheta}(\omega)|)\to
F_{Y_{1}\mid\bw}(0)$. The BCT therefore shows
\begin{align*}
 & \int_{\Omega}\envert[1]{F_{Y_{1}\mid\bw}(|\Delta\bx^{\top}\widehat{\btheta}\left(\omega\right)|)-F_{Y_{1}\mid\bw}\left(0\right)}\bbP\left(\mathrm{d}\omega\right)\\
 & =\int_{\Omega}\bone\left\{ \omega\in A\right\} \envert[1]{F_{Y_{1}\mid\bw}(|\Delta\bx^{\top}\widehat{\btheta}\left(\omega\right)|)-F_{Y_{1}\mid\bw}\left(0\right)}\bbP\left(\mathrm{d}\omega\right)\to0,
\end{align*}
showing that the first term in the upper bound on $g_{n}(\bw)$ goes to zero. To
handle the second term, note that we have already established
$\sqrt{n}$-asymptotic normality (Theorem \ref{thm:AsymptoticNormality-TLS}),
meaning that
$\sqrt{n}(\widehat{\btheta}-\btheta_{0})\rightsquigarrow\boldsymbol{G}$ in
$\R^{K}$, with $\boldsymbol{G}$ denoting a non-degenerate Gaussian in
$\R^{K}$. Since $\Delta\bx^{\top}\btheta_{0}=0$, by the continuous
mapping theorem (CMT) we therefore get
$\sqrt{n}\Delta\bx^{\top}\widehat{\btheta}=\Delta\bx^{\top}\sqrt{n}(\widehat{\btheta}-\btheta_{0})\rightsquigarrow\Delta\bx^{\top}\boldsymbol{G}$
in $\R$. Since $\Delta\bx\neq\bzero$, $\Delta\bx^{\top}\boldsymbol{G}$
is a non-degenerate Gaussian in $\R$ for which
$\R\backslash\left\{ 0\right\} $ is a continuity set. It therefore
follows from the portmanteau theorem, that
\[
\bbP(\Delta\bx^{\top}\widehat{\btheta}\neq0)=\bbP(\sqrt{n}\Delta\bx^{\top}\widehat{\btheta}\neq0)\to\bbP(\Delta\bx^{\top}\boldsymbol{G}\neq0)=1,
\]
and, thus, $\bbP(\Delta\bx^{\top}\widehat{\btheta}=0)\to0.$
Hence, $g_{n}\left(\bw\right)\to0$. It follows that $g_{n}\to0$
pointwise on $\mathcal{W}$, so that domination by a $P_{\bW}$-integrable
function implies $\int_{\mathcal{W}}g_{n}\mathrm{d}P_{\bW}\to0$
via the LDCT. Hence, $\E_{\widehat{\btheta}}[|L_{j,k}(\widehat{\btheta})-J_{0,j,k}|]\to0$,
meaning that $L_{j,k}(\widehat{\btheta})\to J_{0,j,k}$
in $L^{1}\left(\bbP\right)$. Conclude that
\[
\widehat{L}_{j,k}=L_{j,k}(\widehat{\btheta})+o_{\mathrm{a.s.}}\left(1\right)=J_{0,j,k}+o_{L^{1}}\left(1\right)+o_{\mathrm{a.s.}}\left(1\right).
\]

\textbf{Step 2 (Remainder):}
We argue that $\widehat{\mathbf{R}}=\bzero_{K\times
K}+o_{L^{1}}\left(1\right)+o_{\mathrm{a.s.}}\left(1\right)$, which is equivalent
to the statement convergence element by element. To this end, fix
$\left(j,k\right)\in\left[K\right]\times\left[K\right]$, and consider the
function class
\begin{align*}
\mathcal{R}_{j,k} & :=\left\{ r_{j,k}:\mathcal{Z}\to\R\middle|r_{j,k}:=r_{j,k}\left(\cdot;\btheta\right),\btheta\in\R^{K}\right\} ,
\end{align*}
where
\[
r_{j,k}\left(\bz;\btheta\right):=\bone\{ y_{1}>0,y_{2}>0,\Delta\bx^{\top}\btheta=0\} \Delta x_{j}\Delta x_{k}.
\]
With minor modifications, arguments along the lines of the proof of Theorem
\ref{thm:VarianceConsistency-TLS} show that $\mathcal{R}_{j,k}$ is VC-subgraph
with $P$-integrable envelope $\bz\mapsto|\Delta x_{j}\Delta x_{k}|$.
\citet[Lemmas 2.8 and 2.12]{pakes1989simulation} therefore combine to show the
strong ULLN
\[
\sup_{f\in\mathcal{R}_{j,k}}\left|\left(P_{n}-P\right)f\right|\overset{\mathrm{a.s.}}{\to}0.
\]
Hence,
\[
\widehat{R}_{j,k}=P_{n}r_{j,k}(\cdot,\widehat{\btheta})=Pr_{j,k}(\cdot,\widehat{\btheta})+o_{\mathrm{a.s.}}\left(1\right)=R_{j,k}(\widehat{\btheta})+o_{\mathrm{a.s.}}\left(1\right).
\]
where we have defined $R_{j,k}:\R^{K}\to\R$ by
\[
R_{j,k}\left(\btheta\right):=Pr_{j,k}\left(\cdot,\btheta\right)=\E\sbr[1]{\bone\{ Y_{1}>0,Y_{2}>0,\Delta\bX^{\top}\btheta=0\} \Delta X_{j}\Delta X_{k}},
\]
with $\boldsymbol{Z}:=(Y_{1},\bX_{1},Y_{2},\bX_{2})$
having distribution $P$ and being independent of the sample and,
thus, the sequence of TLS estimators. By Jensen's inequality and Fubini--Tonelli
theorem, we get the upper bound
\begin{align*}
\E_{\widehat{\btheta}}[|R_{j,k}(\widehat{\btheta})|] & \leqslant\int_{\Omega}\left[\int_{\mathcal{W}}\mathbf{1}\{\Delta\bx^{\top}\widehat{\btheta}\left(\omega\right)=0\}\left|\Delta x_{j}\Delta x_{k}\right|P_{\bW}\left(\mathrm{d}\bw\right)\right]\bbP\left(\mathrm{d}\omega\right)=\int_{\mathcal{W}}g_{n}\mathrm{d}P_{\bW},
\end{align*}
with $g_{n}:\mathcal{W}\to[0,\infty)$ defined by
\[
g_{n}\left(\bw\right):=\bbP(\Delta\bx^{\top}\widehat{\btheta}=0)\left|\Delta x_{j}\Delta x_{k}\right|.
\]
Note that $0\leqslant g_{n}\leqslant g$ for the $P_{\bW}$-integrable
$g:\bw\mapsto\left|\Delta x_{j}\Delta x_{k}\right|$. To show
$\int_{\mathcal{W}}g_{n}\mathrm{d}P_{\bW}\to0$, by the LDCT, it thus suffices to
show that $g_{n}\to0$ pointwise on $\mathcal{W}$. To this end, we split into the
three exhaustive cases for $\bw$: \textbf{Case 1:} $\Delta\bx=\bzero$;
\textbf{Case 2:} $\Delta\bx^{\top}\btheta_{0}\neq0$; and, \textbf{Case 3:}
$\Delta\bx\neq\bzero$ but $\Delta\bx^{\top}\btheta_{0}=0$. We consider each case
in turn.

\textbf{\emph{Case 1:}} If $\Delta\bx=\bzero$,
then $g_{n}\left(\bw\right)=0$, which $\to0$ trivially.

\textbf{\emph{Case 2:}} Consider $\bw$ such that
$\Delta\bx^{\top}\btheta_{0}\neq0$. \citet[Theorem 1(iv)]{honore_trimmed_1992}
shows strong consistency of TLS, so
$\Delta\bx^\top\widehat{\btheta}\to\Delta\bx^\top\btheta_0\neq0$ almost surely,
which implies $\bbP(\Delta\bx^\top\widehat{\btheta}=0)\to0$.

\textbf{\emph{Case 3:}} Consider $\bw$ such that $\Delta\bx\neq\bzero$
but $\Delta\bx^{\top}\btheta_{0}=0$. Then
$\bbP(\Delta\bx^{\top}\widehat{\btheta}=0)\to0$
by the argument used in Case 3 for the leading term. It follows that
$g_{n}\to0$ pointwise on $\mathcal{W}$, so that
\[
\E_{\widehat{\btheta}}[|R_{j,k}(\widehat{\btheta})|]\leqslant\int_{\mathcal{W}}g_{n}\mathrm{d}P_{\bW}\to0
\]
by the LDCT. Hence, $R_{j,k}(\widehat{\btheta})\to0$ in
$L^{1}\left(\bbP\right)$. It follows that
$\widehat{R}_{j,k}=o_{L^{1}}\left(1\right)+o_{\mathrm{a.s.}}\left(1\right)$.
\end{proof}

\newpage
\bookmarksetup{startatroot}
\part*{Supplemental Appendix}
\addcontentsline{toc}{part}{Supplemental Appendix}
\renewcommand{\thesection}{S\arabic{section}}
\renewcommand{\theHsection}{supplement.\arabic{section}}
\setcounter{section}{0}


\section{Proofs for Trimmed Least Absolute Deviations (Section
\ref{sec:TLAD})} By appropriately modifying the conditional PDF
$(\bw,\be)\mapsto f_{\bepsilon\mid\bw}(\be)$ on $\bW$-null set of values $\bw$,
we can and will strengthen Assumptions \ref{assu:Regularity-TLAD} and
\ref{assu:extra continuity} to say that Assumption \ref{assu:Regularity-TLAD}
reads as ``There is a constant $C\in(0,\infty)$ such that $\sup_{e\in\mathbb
R}f_{\varepsilon_1 - \varepsilon_2\mid\bw}(e)\leqslant C$ and $\sup_{e\in\mathbb
R}f_{\varepsilon\mid\bw}(e)\leqslant C$ for {\em all} $\bw\in\mathcal W$,'' and
Assumption \ref{assu:extra continuity} reads as ``The functions $\be\mapsto
f_{\bepsilon\mid \bw}(\be)$, $e\mapsto f_{\varepsilon\mid\bw}(e)$, and $e\mapsto
f_{\varepsilon_1 - \varepsilon_2\mid\bw}(e)$ are continuous for {\em all}
$\bw\in\mathcal W$.'' Also, we let $(\bw,e)\mapsto F_{\Delta Y^{\ast}\mid
\bw}(e)$ and $(\bw,e)\mapsto F_{\Delta \varepsilon\mid \bw}(e)$ denote the CDF
of $\Delta Y^{\ast} = Y_1^{\ast} - Y_2^{\ast}$ and the CDF of $\Delta\varepsilon
= \varepsilon_1 - \varepsilon_2$, both conditional on $\bW = \bw$ for all
$\bw\in\mathcal W$. In addition, we let $(\bw,y)\mapsto F_{Y_1^{\ast}\mid
\bw}(y)$ and $(\bw,y)\mapsto F_{Y_2^{\ast}\mid \bw}(y)$ denote the marginal CDF
of $Y_1^{\ast}$ and the marginal CDF of $Y_2^{\ast}$, both conditional on $\bW =
\bw$ for all $\bw\in\mathcal W$. Finally, we let $(\bw,e)\mapsto
F_{\varepsilon\mid \bw}(e)$ denote the common marginal CDF of $\varepsilon_1$
and $\varepsilon_2$ conditional on $\bW = \bw$ for all $\bw\in\mathcal W$.

\subsection{Proof of Theorem \ref{thm:LossHessianExistence-TLAD}}

Observe that the function $m^{\texttt{tlad}}(t,\by)$ in
\eqref{eq:TrimmedAbsLoss} can be rewritten as
\begin{equation}\label{eq: mtlad alternative expressions original}
\left.
\begin{aligned}
m^{\texttt{tlad}}(t,\by) & =\bone\{y_{1}>0\}\bone\{y_{2}>0\}|y_{1}-y_{2}-t|\\
 & \quad+\bone\{y_{1}>0\}\bone\{y_{2}=0\}\max\{0,y_{1}-t\}\\
 & \quad+\bone\{y_{1}=0\}\bone\{y_{2}>0\}\max\{0,y_{2}+t\}
\end{aligned}
\right\}
\end{equation}
for all $t\in\mathbb R$ and $\by\in[0,\infty)^{2}$. Thus, defining $\widetilde{m}^{\texttt{tlad}}(t,\by) :=m^{\texttt{tlad}}(t,\by)-m^{\texttt{tlad}}(0,\by)$, we have
\begin{equation}\label{eq: mtlad alternative expressions}
\left.
\begin{aligned}
\widetilde{m}^{\texttt{tlad}}(t,\by) 
  & =\bone\{y_{1}>0\}\bone\{y_{2}>0\}\left(|y_{1}-y_{2}-t|-|y_{1}-y_{2}|\right)\\
  & \quad+\bone\{y_{1}>0\}\bone\{y_{2}=0\}\left(\max\{0,y_{1}-t\}-\max\{0,y_{1}\}\right)\\
  & \quad+\bone\{y_{1}=0\}\bone\{y_{2}>0\}\left(\max\{0,y_{2}+t\}-\max\{0,y_{2}\}\right).
\end{aligned}
\right\}
\end{equation}
for all $t\in\mathbb R$ and $\by\in[0,\infty)^{2}$. Then for any $t\in\R$, we have
$
\E[|\widetilde{m}^{\tt{tlad}}(t,\bY)|]\leqslant|t|<\infty.
$
Hence, using the construction in Appendix \ref{sec: measurability appendix}
along with the modification of the conditional PDF on $\bW$-null sets, we can
define a measurable (conditional shifted trimmed absolute loss) function
$\ell:\R\times\cW\to\R$ by
\begin{equation}\label{eq:ExpectedLossDefn-TLAD}
\ell(t,\bw):=\E\left[\widetilde{m}^{\tt{tlad}}(t,\bY)\middle|\bW=\bw\right].
\end{equation}
Iterating expectations, we can relate the functions $L$ in
\eqref{eq:PopulationLossFunction-TLAD} and $\ell$ in
\eqref{eq:ExpectedLossDefn-TLAD} through
\[
L(\btheta)=\E\sbr[2]{\E\sbr[1]{\widetilde{m}^{\tt{tlad}}(\Delta\bX^{\top}\btheta,\bY)\mid\bW}}=\E\sbr[1]{\ell\del[1]{\Delta\bX^{\top}\btheta,\bW}}.
\]
The differentiability properties of $L$ (at $\btheta_0$) will by and large be
deduced from those of $\ell(\cdot,\bw)$, which we turn to next. As we will
condition on $\bW=\bw$ throughout, we abbreviate these conditional expectations
$\E_{\bw}\left[\cdot\right]:=\E\left[\cdot\middle|\bW=\bw\right]$. From
$\bone\{Y_{\tau}>0\}=\bone\{Y_{\tau}^{\ast}>0\}$ and
$\bone\{Y_{\tau}=0\}=\bone\{Y_{\tau}^{\ast}\leqslant0\}$ it
follows that
\begin{align*}
\ell(t,\bw) & =\E_{\bw}\left[\bone\{Y_{1}^{\ast}>0\}\bone\{Y_{2}^{\ast}>0\}\left(|Y_{1}^{\ast}-Y_{2}^{\ast}-t|-|Y_{1}^{\ast}-Y_{2}^{\ast}|\right)\right]\\
 & \quad+\E_{\bw}\left[\bone\{Y_{1}^{\ast}>0\}\bone\{Y_{2}^{\ast}\leqslant0\}\left(\max\{0,Y_{1}^{\ast}-t\}-\max\{0,Y_{1}^{\ast}\}\right)\right]\\
 & \quad+\E_{\bw}\left[\bone\{Y_{1}^{\ast}\leqslant0\}\bone\{Y_{2}^{\ast}>0\}\left(\max\{0,Y_{2}^{\ast}+t\}-\max\{0,Y_{2}^{\ast}\}\right)\right].
\end{align*}
To show Theorem \ref{thm:LossHessianExistence-TLAD}, we rely on the following
two lemmas, the proofs of which can be found at the end of this section.
\begin{lem}[\textbf{Derivative of Expected Trimmed Absolute
Loss}]\label{lem:DerivativeExpectedTrimmedAbsLoss} Let Assumptions
\ref{assu:Non-Degeneracy-TLAD}--\ref{assu:extra continuity} hold and fix
$\bw\in\cW$. Then the function $\ell(\cdot,\bw)$ defined in
\eqref{eq:ExpectedLossDefn-TLAD} is Lipschitz continuous on $\R$ with Lipschitz
constant one and differentiable with derivative given by
\begin{equation}
  \left.\begin{aligned}
    \dot\ell_{1}(t,\bw)
     & =\E_{\bw}\sbr[1]{\bone\{Y_{1}^{\ast}>0\}\bone\{Y_{2}^{\ast}>0\}\left(2\bone\{Y_{1}^{\ast}-Y_{2}^{\ast}\leqslant t\}-1\right)}\\
     & \quad-\E_{\bw}\sbr[1]{\bone\{Y_{1}^{\ast}>0\}\bone\{Y_{2}^{\ast}\leqslant0\}\bone\{Y_{1}^{\ast}>t\}}\\
     & \quad+\E_{\bw}\sbr[1]{\bone\{Y_{1}^{\ast}\leqslant0\}\bone\{Y_{2}^{\ast}>0\}\bone\{Y_{2}^{\ast}>-t\}}.
  \end{aligned}\right\}\label{eq:DerivativeExpectedTrimmedAbsLoss}
\end{equation}
\end{lem}

\begin{rem}[\textbf{Comparison with
\citet{honore_trimmed_1992}}]\label{rem:AltExpForDerivOfExpectedTrimmedAbsLoss}
Using the sign information, we can also cast the derivative as
\begin{align*}
\dot\ell_{1}(t,\bw)
  &=\E_{\bw}\left[\bone\{Y_{2}^{\ast}>0\}\bone\{Y_{2}^{\ast}>\max\{0,Y_{1}^{\ast}\}-t\}
    -\bone\{Y_{1}^{\ast}>0\}\bone\{Y_{1}^{\ast}>\max\{0,Y_{2}^{\ast}\}+t\}\right]\\
  &=\E_{\bw}\left[\bone\{Y_{2}>0\}\bone\{Y_{2}>Y_{1}-t\}
    -\bone\{Y_{1}>0\}\bone\{Y_{1}>Y_{2}+t\}\right].
\end{align*}
The first right-hand side (using the latent outcomes) matches the (censored
TLAD) expression in \citet[Lemma A.1]{honore_trimmed_1992}. Multiplying by
$\Delta\bx$ and iterating expectations, the second right-hand side (using the
observable outcomes) gives rise to $\bfV_0^{\tt{tlad}}$ in
\eqref{eq:VarianceSandwichMeat-TLAD}.\hfill$\diamondsuit$
\end{rem}

\begin{lem}[\textbf{Second-Order Differentiability of Expected Trimmed Absolute
Loss}]\label{lem:SecondOrderDerivativeExpectedTrimmedAbsLoss} Let Assumptions
\ref{assu:Non-Degeneracy-TLAD}--\ref{assu:extra continuity} hold and fix
$\bw\in\cW$. Then:
\begin{enumerate}[(1)]
\item\label{enu:DotEllDiffAwayFromZero} $\dot{\ell}_1(\cdot,\bw)$ is
differentiable at $t\neq0$ with derivative given by
\begin{equation}\label{eq:SecondOrderDerivativeExpectedTrimmedAbsLoss}
\ddot{\ell}_{11}(t,\bw)=\begin{cases}
  2\int_{0}^{+\infty}f_{\bY^{\ast}\mid\bw}\left(z,z-t\right)\dif z+\int_{-\infty}^{0}f_{\bY^{\ast}\mid\bw}\left(z,-t\right)\dif z, & t<0,\\
  2\int_{0}^{+\infty}f_{\bY^{\ast}\mid\bw}\left(z+t,z\right)\dif z+\int_{-\infty}^{0}f_{\bY^{\ast}\mid\bw}\left(t,z\right)\dif z, & t>0.
\end{cases}
\end{equation}
\item\label{enu:DotEllDirDiffAtZero} $\dot{\ell}_1(\cdot,\bw)$ is semi-differentiable at $t=0$ with left and
right derivatives given by
\begin{align}
  \ddot{\ell}_{11-}(0,\bw)&=2\int_{0}^{+\infty}f_{\bY^{\ast}\mid\bw}\left(z,z\right)\dif z+\int_{-\infty}^{0}f_{\bY^{\ast}\mid\bw}\left(z,0\right)\dif z
  \quad\text{and}\label{eq:SecondOrderLeftDerivativeExpectedTrimmedAbsLoss}\\
  \ddot{\ell}_{11+}(0,\bw)&=2\int_{0}^{+\infty}f_{\bY^{\ast}\mid\bw}\left(z,z\right)\dif z+\int_{-\infty}^{0}f_{\bY^{\ast}\mid\bw}\left(0,z\right)\dif z.\label{eq:SecondOrderRightDerivativeExpectedTrimmedAbsLoss}
\end{align}
\item\label{enu:DotEllDirDiffAtxtheta} $\dot{\ell}_1(\cdot,\bw)$ is
differentiable at $t=\Delta\bx^\top\btheta_0$ with derivative given by
\begin{equation}\label{eq:SecondOrderDerivativeExpectedTrimmedAbsLossAtxtheta}
\left.\begin{aligned}
  \ddot{\ell}_{11}(\Delta\bx^{\top}\btheta_{0},\bw)
  &=2\int_{0}^{+\infty}f_{\bY^{\ast}\mid\bw}\del[1]{z+\max\cbr[0]{0,\Delta\bx^\top\btheta_0},z-\min\cbr[0]{0,\Delta\bx^\top\btheta_0}}\dif z\\
  & \qquad+\bone\{\Delta\bx^\top\btheta_0\geqslant0\}\int_{-\infty}^{0}f_{\bY^{\ast}\mid\bw}\del[1]{\Delta\bx^\top\btheta_0,z}\dif z\\
  & \qquad+\bone\{\Delta\bx^\top\btheta_0<0\}\int_{-\infty}^{0}f_{\bY^{\ast}\mid\bw}\del[1]{z,-\Delta\bx^\top\btheta_0}\dif z.
\end{aligned}\right\}
\end{equation}
\end{enumerate}
\end{lem}
\begin{proof}[\sc{Proof of Theorem \ref{thm:LossHessianExistence-TLAD}}]
First, fix $\btheta\in\R^K$ and $\bvtheta\in\R^K$ and let
$\{\tau_m\}_{m=1}^\infty$ and $\{\bvtheta_m\}_{m=1}^\infty$ be such that
$\tau_m\to0_+$ and $\bvtheta_m\to\bvtheta$. Then
\begin{align*}
  \frac{L(\btheta+\tau_m\bvtheta_m)-L(\btheta)}{\tau_m}
  & =\E\left[\frac{\ell\del[1]{\Delta\bX^\top(\btheta+\tau_m\bvtheta_m),\bW}-\ell\del[1]{\Delta\bX^\top\btheta,\bW}}{\tau_m}\right]
\end{align*}
To apply the Generalized Lebesgue Dominated Convergence Theorem (GLDCT, Theorem
\ref{thm:GLDCT}), define functions $\{f_m\}_{m=1}^\infty$ on $\Omega$ by
\begin{align*}
  f_m(\omega)
  &:=\frac{\ell\del[1]{\Delta\bX(\omega)^\top\btheta+\tau_m\Delta\bX(\omega)^\top\bvtheta_m,\bW(\omega)}-\ell\del[1]{\Delta\bX(\omega)^\top\btheta,\bW(\omega)}}{\tau_m}.
\end{align*}
By measurability of $\ell$ established above, each $f_m$ is measurable. Lemma
\ref{lem:DerivativeExpectedTrimmedAbsLoss} shows that $\ell(\cdot,\bw)$ is
Lipschitz continuous with Lipschitz constant one. Hence, by Lipschitz continuity
followed by the Cauchy-Schwarz inequality,
\[
\envert[1]{f_m(\omega)}\leqslant\envert[1]{\Delta\bX(\omega)^\top\bvtheta_m}\leqslant\enVert[0]{\Delta\bX}_2\enVert[0]{\bvtheta_m}_2=:g_m(\omega).
\]
Since $\E[\|\Delta\bX\|_2]<\infty$ by Assumption
\ref{assu:MomentConditions-TLAD}, each $g_m$ is integrable. The previous display
therefore goes to show that $\{f_m\}$ is dominated by the nonnegative sequence
$\{g_m\}$. Since $\bvtheta_m\to\bvtheta$, we have $g_m\to g$ pointwise on
$\Omega$, and
\[
\int_\Omega g_m\dif\bbP = \enVert[0]{\bvtheta_m}_2\E[\|\Delta\bX\|_2]
\to
\enVert[0]{\bvtheta}_2\E[\|\Delta\bX\|_2]
=
\int_\Omega g\dif\bbP
<\infty,
\]
where $g(\omega):=\enVert[0]{\bvtheta}_2\|\Delta\bX(\omega)\|_2$. Lemma
\ref{lem:DerivativeExpectedTrimmedAbsLoss} also shows that
$\ell(\cdot,\bW(\omega))$ is differentiable at
$t=\Delta\bX(\omega)^\top\btheta$, thus yielding the pointwise convergence
\begin{align*}
  f_m(\omega)&\to \dot{\ell}_1\del[1]{\Delta\bX(\omega)^\top\btheta,\bW(\omega)}\Delta\bX(\omega)^\top\bvtheta=:f(\omega).
\end{align*}
Appealing to the GLDCT, we conclude that
\begin{align*}
  \frac{L(\btheta+\tau_m\bvtheta_m)-L(\btheta)}{\tau_m}
  &\to\E\sbr[1]{\dot{\ell}_1\del[1]{\Delta\bX^\top\btheta,\bW}\Delta\bX^\top}\bvtheta.
\end{align*}
Since the limit exists for every $\bvtheta\in\R^K$, is linear in $\bvtheta$, and
is independent of the sequences $\{\tau_m\}$ and $\{\bvtheta_m\}$, we conclude
that $L$ is Hadamard differentiable at $\btheta$. Since $\R^K$ is
finite-dimensional, Hadamard differentiability is equivalent to (Fr\'echet)
differentiability. Since $\btheta\in\R^K$ was arbitrary, $L$ is everywhere
differentiable with gradient $\nabla L:\R^K\to\R^K$ given by
$
\nabla L(\btheta)=\E\sbr[1]{\dot{\ell}_1\del[1]{\Delta\bX^\top\btheta,\bW}\Delta\bX}.
$

We next establish second-order differentiability of $L$ at $\btheta_0$. To this
end, fix $\bvtheta\in\R^K$ and let $\{\tau_m\}_{m=1}^\infty$ and
$\{\bvtheta_m\}_{m=1}^\infty$ be such that $\tau_m\to0_+$ and
$\bvtheta_m\to\bvtheta$. Then
\begin{align*}
\frac{\nabla L(\btheta_0+\tau_m\bvtheta_m)-\nabla L(\btheta_0)}{\tau_m}
&=\E\left[\frac{\dot{\ell}_1\del[1]{\Delta\bX^\top(\btheta_0+\tau_m\bvtheta_m),\bW}-\dot{\ell}_1\del[1]{\Delta\bX^\top\btheta_0,\bW}}{\tau_m}\Delta\bX\right].
\end{align*}
To apply the GLDCT coordinatewise, consider the measure space
$(\R^{2K+1},\cB_{2K+1},\mu_{\bW})$, where
$\mu_{\bW}\left(\cdot\right):=\bbP(\bW^{-1}\left(\cdot\right))$ denotes the law
of $\bW$, and fix $j\in[K]$. Define functions $\{f_m\}_{m=1}^{\infty}$ on
$\R^{2K+1}$ by
\begin{align*}
  f_m(\bw)
  &:=\bone\{\bw\in\cW\}\frac{\dot{\ell}_1\del[1]{\Delta\bx^\top(\btheta_0+\tau_m\bvtheta_m),\bw}-\dot{\ell}_1\del[1]{\Delta\bx^\top\btheta_0,\bw}}{\tau_m}\Delta x_j.
\end{align*}
Each $f_m$ is measurable and real-valued. In addition, by
\eqref{eq:DerivativeExpectedTrimmedAbsLoss} and Assumption
\ref{assu:Regularity-TLAD},
\begin{align*}
\envert{f_m(\bw)} & \leqslant \frac{\envert[1]{\dot{\ell}_1\del[1]{\Delta\bx^\top(\btheta_0+\tau_m\bvtheta_m),\bw}-\dot{\ell}_1\del[1]{\Delta\bx^\top\btheta_0,\bw}}}{\tau_m}\cdot|\Delta x_j|\\
&\leqslant 4C|\Delta\bx^\top\bvtheta_m|\cdot|\Delta x_j|
\leqslant
4C\enVert[0]{\bvtheta_m}_2\enVert[0]{\Delta\bx}_2^2
=:g_m(\bw).
\end{align*}
Since $\E[\|\Delta\bX\|_2^2]<\infty$ by Assumption
\ref{assu:MomentConditions-TLAD}, each $g_m$ is $\mu_{\bW}$-integrable. The
previous display therefore goes to show that $\{f_m\}$ is dominated by the
nonnegative sequence $\{g_m\}$. Since $\bvtheta_m\to\bvtheta$, we have $g_m\to
g$ pointwise on $\R^{2K+1}$, and
\[
\int_{\R^{2K+1}}g_m\dif\mu_{\bW} = 4C\enVert[0]{\bvtheta_m}_2\E[\|\Delta\bX\|_2^2]
\to
4C\enVert[0]{\bvtheta}_2\E[\|\Delta\bX\|_2^2]
=
\int_{\R^{2K+1}}g\dif\mu_{\bW}
<\infty,
\]
where $g(\bw):=4C\enVert[0]{\bvtheta}_2\|\Delta\bx\|_2^2$. Lemma
\ref{lem:SecondOrderDerivativeExpectedTrimmedAbsLoss}.\ref{enu:DotEllDirDiffAtxtheta}
shows that $\dot{\ell}_1(\cdot,\bw)$ is differentiable at
$t=\Delta\bx^\top\btheta_0$, thus yielding the pointwise convergence
\[
f_m(\bw)\to\bone\{\bw\in\cW\}\ddot{\ell}_{11}\del[1]{\Delta\bx^\top\btheta_0,\bw}\Delta x_j\Delta\bx^\top\bvtheta=:f(\bw).
\]
Appealing to the GLDCT, stacking over the coordinates $j\in[K]$, we conclude
that
\[
\frac{\nabla L(\btheta_0+\tau_m\bvtheta_m)-\nabla L(\btheta_0)}{\tau_m}\to\E\sbr[2]{\ddot{\ell}_{11}\del[1]{\Delta\bX^\top\btheta_0,\bW}\Delta\bX\Delta\bX^\top}\bvtheta.
\]
Since the limit exists for every $\bvtheta\in\R^K$, is linear in $\bvtheta$, and
is independent of the sequences $\{\tau_m\}$ and $\{\bvtheta_m\}$, we conclude
that $\nabla L$ is Hadamard differentiable at $\btheta_0$. As Hadamard
differentiability is here equivalent to (Fr\'echet) differentiability, we
conclude that $\nabla L$ is differentiable at $\btheta_0$. Hence, $L$ is twice
differentiable at $\btheta_0$ with Hessian given by
\[
\nabla^2L(\btheta_0)=\E\sbr[2]{\ddot{\ell}_{11}\del[1]{\Delta\bX^\top\btheta_0,\bW}\Delta\bX\Delta\bX^\top}.
\]
The expression \eqref{eq:LossHessian-TLAD} now follows from the previous display
and \eqref{eq:SecondOrderDerivativeExpectedTrimmedAbsLossAtxtheta}.
\end{proof}

We end this section by providing the proofs for Lemmas
\ref{lem:DerivativeExpectedTrimmedAbsLoss} and
\ref{lem:SecondOrderDerivativeExpectedTrimmedAbsLoss}.

\begin{proof}[\sc{Proof of Lemma \ref{lem:DerivativeExpectedTrimmedAbsLoss}}]
$\ell(\cdot,\bw)$ inherits 1-Lipchitzness from $m^{\tt{tlad}}(\cdot,\by)$ via (a
conditional version of) Jensen's inequality. For the differentiability claim,
observe that the expectand underlying $\ell(t,\bw)$ is
\begin{align*}
h(t,\by^{\ast}) & :=\bone\{y_{1}^{\ast}>0\}\bone\{y_{2}^{\ast}>0\}\left(|y_{1}^{\ast}-y_{2}^{\ast}-t|-|y_{1}^{\ast}-y_{2}^{\ast}|\right)\\
  & \qquad+\bone\{y_{1}^{\ast}>0\}\bone\{y_{2}^{\ast}\leqslant0\}\left(\max\{0,y_{1}^{\ast}-t\}-\max\{0,y_{1}^{\ast}\}\right)\\
  & \qquad+\bone\{y_{1}^{\ast}\leqslant0\}\bone\{y_{2}^{\ast}>0\}\left(\max\{0,y_{2}^{\ast}+t\}-\max\{0,y_{2}^{\ast}\}\right).
\end{align*}
For each $\by^{\ast}\in\R^{2}$, $h(\cdot,\by^{\ast})$
is finite convex on $\R$. This function \emph{fails}
to be differentiable at $t$ only when $\by^{\ast}$ lies
in the subset $N(t)$ of $\R^{2}$ defined by
\begin{align*}
N(t) & := \begin{cases}
\left\{ \by^{\ast}\in\R^{2}\middle| y_{1}^{\ast}\leqslant0,y_{2}^{\ast}=-t\right\} \cup\left\{ \by^{\ast}\in\R^{2}\middle| y_{1}^{\ast}>0,y_{2}^{\ast}>0,y_{1}^{\ast}-y_{2}^{\ast}=t\right\} , & t<0,\\
\left\{ \by^{\ast}\in\R^{2}\middle| y_{1}^{\ast}>0,y_{2}^{\ast}>0,y_{1}^{\ast}=y_{2}^{\ast}\right\} , & t=0,\\
\left\{ \by^{\ast}\in\R^{2}\middle| y_{1}^{\ast}=t,y_{2}^{\ast}\leqslant0\right\} \cup\left\{ \by^{\ast}\in\R^{2}\middle| y_{1}^{\ast}>0,y_{2}^{\ast}>0,y_{1}^{\ast}-y_{2}^{\ast}=t\right\} , & t>0.
\end{cases}
\end{align*}
For $\by^{\ast}\in\R^{2}\backslash N(t)$, we have 
\begin{align*}
\dot{h}_{1}(t,\by^{\ast}) & =\bone\{y_{1}^{\ast}>0\}\bone\{y_{2}^{\ast}>0\}\left(2\bone\{y_{1}^{\ast}-y_{2}^{\ast}\leqslant t\}-1\right)\\
  & \qquad-\bone\{y_{1}^{\ast}>0\}\bone\{y_{2}^{\ast}\leqslant0\}\bone\{y_{1}^{\ast}> t\}\\
  & \qquad+\bone\{y_{1}^{\ast}\leqslant0\}\bone\{y_{2}^{\ast}>0\}\bone\{y_{2}^{\ast}>-t\}.
\end{align*}
As $N(t)$ is the union of a finite number of rays in the plane, we have
$\lambda_{2}(N(t))=0$ for all $t\in\R$. 
Given that the distribution
$\mu_{\bY^{\ast}\mid\bw}$ of $\bY^{\ast}$ conditional on
$\bW=\bw$ is absolutely continuous with respect to $\lambda_{2}$, the points
$N(t)$ leading to expectand non-differentiability at $t$ are therefore
negligible, $\mu_{\bY^{\ast}\mid\bw}(N(t))=0$. \citet[Proposition
2.3]{bertsekas1973stochastic} therefore tells us that $\ell(\cdot,\bw)$
\emph{is} differentiable with derivative given by
\begin{align*}
  \dot\ell_{1}(t,\bw) & =\int_{\R^{2}\backslash N(t)}\dot{h}_{1}(t,\by^{\ast})\mu_{\bY^{\ast}\mid\bw}(\dif \by^{\ast})=\int_{\R^{2}\backslash N(t)}\dot{h}_{1}(t,\by^{\ast})f_{\bY^{\ast}\mid\bw}(\by^{\ast})\dif \by^{\ast}\\
    & =\E_{\bw}\sbr[1]{\bone\{Y_{1}^{\ast}>0\}\bone\{Y_{2}^{\ast}>0\}\left(2\bone\{Y_{1}^{\ast}-Y_{2}^{\ast}\leqslant t\}-1\right)}\\
    & \quad-\E_{\bw}\sbr[1]{\bone\{Y_{1}^{\ast}>0\}\bone\{Y_{2}^{\ast}\leqslant0\}\bone\{Y_{1}^{\ast}>t\}}\\
    & \quad+\E_{\bw}\sbr[1]{\bone\{Y_{1}^{\ast}\leqslant0\}\bone\{Y_{2}^{\ast}>0\}\bone\{Y_{2}^{\ast}>-t\}},
\end{align*}
which gives the asserted claim.  
\end{proof}

\begin{proof}[\sc{Proof of Lemma \ref{lem:SecondOrderDerivativeExpectedTrimmedAbsLoss}}]
For notational convenience, abbreviate $M(t,\bw):=\dot\ell_{1}(t,\bw)$ and
decompose as follows:
\begin{align*}
  M(t,\bw)
    & =\E_{\bw}\sbr[1]{\bone\{Y_{1}^{\ast}>0\}\bone\{Y_{2}^{\ast}>0\}\left(2\bone\{Y_{1}^{\ast}-Y_{2}^{\ast}\leqslant t\}-1\right)}\tag{\ensuremath{=:M_{a}(t,\bw)}}\\
    & \quad-\E_{\bw}\sbr[1]{\bone\{Y_{1}^{\ast}>0\}\bone\{Y_{2}^{\ast}\leqslant0\}\bone\{Y_{1}^{\ast}>t\}}\tag{\ensuremath{=:M_{b}(t,\bw)}}\\
    & \quad+\E_{\bw}\sbr[1]{\bone\{Y_{1}^{\ast}\leqslant0\}\bone\{Y_{2}^{\ast}>0\}\bone\{Y_{2}^{\ast}>-t\}}.\tag{\ensuremath{=:M_{c}(t,\bw)}}
\end{align*}
We next establish (at least) semi-differentiability of the contributions $a,b,$
and $c$ in turn, and provide the right $(+)$ and left $(-)$ derivatives. To this
end, fix $t\in\R$ and let $\{\tau_{m}\}_{m=1}^{\infty}$ be a sequence in
$(0,\infty)$ satisfying $\tau_{m}\to0$ $(\tau_{m}\to0_{+})$.

\noindent\textbf{Part a, \emph{right} differentiability:} Express the difference
quotient as
\begin{align}
  & \frac{M_{a}(t+\tau_{m},\bw)-M_{a}(t,\bw)}{\tau_{m}}\nonumber\\
  & =2\tau_{m}^{-1}\E_{\bw}\sbr[1]{\bone\{Y_{1}^{\ast}>0\}\bone\{Y_{2}^{\ast}>0\}\left(\bone\{\Delta Y^{\ast}\leqslant t+\tau_{m}\}-\bone\{\Delta Y^{\ast}\leqslant t\}\right)}\nonumber\\
  & =2\tau_{m}^{-1}\int_{\R^{2}}\bone\{y_{1}^{\ast}>0\}\bone\{y_{2}^{\ast}>0\}\bone\{y_{1}^{\ast}-y_{2}^{\ast}\in(t,t+\tau_{m}]\}f_{\bY^{\ast}\mid\bw}\left(\by^{\ast}\right)\dif \by^{\ast}\nonumber\\
  & =2\tau_{m}^{-1}\int_{\R}\bone\{y_{2}^{\ast}>0\}\left[\int_{\R}\bone\{y_{1}^{\ast}>0\}\bone\{y_{1}^{\ast}-y_{2}^{\ast}\in(t,t+\tau_{m}]\}f_{\bY^{\ast}\mid\bw}\left(y_{1}^{\ast},y_{2}^{\ast}\right)\dif y_{1}^{\ast}\right]\dif y_{2}^{\ast}\nonumber\\
  & =2\tau_{m}^{-1}\int_{\R}\bone\{y_{2}^{\ast}>0\}\left[\int_{\R}\bone\{y_{1}^{\ast}\geqslant0\}\bone\{y_{1}^{\ast}-y_{2}^{\ast}\in[t,t+\tau_{m}]\}f_{\bY^{\ast}\mid\bw}\left(y_{1}^{\ast},y_{2}^{\ast}\right)\dif y_{1}^{\ast}\right]\dif y_{2}^{\ast}\nonumber\\
  & =2\tau_{m}^{-1}\int_{\R}\bone\{y_{2}^{\ast}>0\}\left[\int_{\R}\bone\{u+y_{2}^{\ast}\geqslant0\}\bone\{u\in[t,t+\tau_{m}]\}f_{\bY^{\ast}\mid\bw}\left(u+y_{2}^{\ast},y_{2}^{\ast}\right)\dif u\right]\dif y_{2}^{\ast}\nonumber\\
  & =\int_{\R}2\bone\{y_{2}^{\ast}>0\}\tau_{m}^{-1}\left[\int_{[t,t+\tau_{m}]}\bone\{u\geqslant-y_{2}^{\ast}\}f_{\bY^{\ast}\mid\bw}\left(u+y_{2}^{\ast},y_{2}^{\ast}\right)\dif u\right]\dif y_{2}^{\ast}, \label{eq: cutting equations 1}
\end{align}
where we have used non-negativity to invoke Tonelli's theorem, absolute
continuity to modify the inner integral on a Lebesgue null set in $\R$ (which
changes with $y_{2}^{\ast}$), and the change of variables
$u:=y_{1}^{\ast}-y_{2}^{\ast}$. Consider the measure space
$(\R,\mathcal{B},\lambda)$ and define the (outer integrand) function $f_{m}$ by
\begin{align*}
f_{m}(y_{2}^{\ast}) & :=2\bone\{y_{2}^{\ast}>0\}\tau_{m}^{-1}\int_{[t,t+\tau_{m}]}\bone\{u\geqslant-y_{2}^{\ast}\}f_{\bY^{\ast}\mid\bw}\left(u+y_{2}^{\ast},y_{2}^{\ast}\right)\dif u.
\end{align*}
Then $f_{m}$ is non-negative and bounded from above by $g_{m}$ defined by
\[
g_{m}(y_{2}^{\ast}):=2\tau_{m}^{-1}\int_{[t,t+\tau_{m}]}f_{\bY^{\ast}\mid\bw}\left(u+y_{2}^{\ast},y_{2}^{\ast}\right)\dif u.
\]
Assumption \ref{assu:Regularity-TLAD} implies that $f_{\Delta
Y^{\ast}\mid\bw}(\cdot)=f_{\Delta
\varepsilon\mid\bw}(\cdot-\Delta\bx^\top\btheta_0)$ is bounded by a constant
$C$, so Tonelli's theorem yields
\begin{align*}
\int_{\R}g_{m}(y_{2}^{\ast})\dif y_{2}^{\ast} & =2\tau_{m}^{-1}\int_{[t,t+\tau_{m}]}\left[\int_{\R}f_{\bY^{\ast}\mid\bw}\left(u+y_{2}^{\ast},y_{2}^{\ast}\right)\dif y_{2}^{\ast}\right]\dif u\\
  & =2\tau_{m}^{-1}\int_{[t,t+\tau_{m}]}f_{\Delta Y^{\ast}\mid\bw}(u)\dif u\leqslant C,
\end{align*}
showing that $g_{m}$ (and thus $f_{m}$) is integrable. Since
$f_{\bY^{\ast}\mid\bw}(\cdot,\cdot) = f_{\bepsilon\mid\bw}(\cdot - a - \bx_1^\top\btheta_0,\cdot - a - \bx_2^\top\btheta_0)$ is 
continuous (Assumption \ref{assu:extra continuity}) and $\bone\{\cdot\geqslant-y_{2}^{\ast}\}$ is \emph{right}
continuous, both inner integrands $u\mapsto
f_{\bY^{\ast}\mid\bw}\left(u+y_{2}^{\ast},y_{2}^{\ast}\right)$
and
$u\mapsto\bone\{u\geqslant-y_{2}^{\ast}\}f_{\bY^{\ast}\mid\bw}\left(u+y_{2}^{\ast},y_{2}^{\ast}\right)$
are right continuous for each $y_{2}^{\ast}\in\R$. As
$\tau_{m}\to0_{+}$, it follows from right continuity that both
\[
g_{m}(y_{2}^{\ast})\to2f_{\bY^{\ast}\mid\bw}\left(t+y_{2}^{\ast},y_{2}^{\ast}\right)=:g(y_{2}^{\ast}),
\]
pointwise in $y_{2}^{\ast}\in\R$ and
\begin{align*}
f_{m}(y_{2}^{\ast})   & \to2\bone\{y_{2}^{\ast}>0\}\bone\{t\geqslant-y_{2}^{\ast}\}f_{\bY^{\ast}\mid\bw}\left(t+y_{2}^{\ast},y_{2}^{\ast}\right)=:f(y_{2}^{\ast})
\end{align*}
pointwise in $y_{2}^{\ast}\in\R$. Also, since $f_{\Delta
Y^{\ast}\mid\bw}(\cdot)=f_{\Delta
\varepsilon\mid\bw}(\cdot-\Delta\bx^\top\btheta_0)$ is continuous (Assumption \ref{assu:extra continuity}),
\begin{align*}
\int_{\R}g_{m}(y_{2}^{\ast})\dif y_{2}^{\ast} & =2\tau_{m}^{-1}\int_{[t,t+\tau_{m}]}f_{\Delta Y^{\ast}\mid\bw}(u)\dif u\\
  & \to2f_{\Delta Y^{\ast}\mid\bw}(t)=2\int_{\R}f_{\bY^{\ast}\mid\bw}\left(t+y_{2}^{\ast},y_{2}^{\ast}\right)\dif y_{2}^{\ast}=\int_{\R}g(y_{2}^{\ast})\dif y_{2}^{\ast}<\infty.
\end{align*}
It thus follows from the Generalized Lebesgue Dominated Convergence Theorem
(GLDCT) in Theorem \ref{thm:GLDCT} that $f$ is integrable and $\int
f_{m}\dif \lambda\to\int f\dif \lambda$. The latter convergence
translates to
\begin{align*}
\frac{M_{a}(t+\tau_{m},\bw)-M_{a}(t,\bw)}{\tau_{m}} & \to2\int_{\R}\bone\{y_{2}^{\ast}>0\}\bone\{t\geqslant-y_{2}^{\ast}\}f_{\bY^{\ast}\mid\bw}\left(t+y_{2}^{\ast},y_{2}^{\ast}\right)\dif y_{2}^{\ast},
\end{align*}
showing that $M_{a}(\cdot,\bw)$ is \emph{right} \emph{differentiable}
at $t$ with \emph{right derivative}
\[
\dot{M}_{a,1+}(t,\bw)=2\int_{\max\{0,-t\}}^{+\infty}f_{\bY^{\ast}\mid\bw}\left(t+y_{2}^{\ast},y_{2}^{\ast}\right)\dif y_{2}^{\ast}.
\]
\textbf{Part a, \emph{left} differentiability:} Express the difference quotient
as
\begin{align}
& \frac{M_{a}(t-\tau_{m},\bw)-M_{a}(t,\bw)}{(-\tau_{m})} \nonumber\\
& =-2\tau_{m}^{-1}\E_{\bw}\sbr[1]{\bone\{Y_{1}^{\ast}>0\}\bone\{Y_{2}^{\ast}>0\}\left(\bone\{\Delta Y^{\ast}\leqslant t-\tau_{m}\}-\bone\{\Delta Y^{\ast}\leqslant t\}\right)}\nonumber\\
& =\int_{\R}2\bone\{y_{2}^{\ast}>0\}\tau_{m}^{-1}\left[\int_{[t-\tau_{m},t]}\bone\{u\geqslant-y_{2}^{\ast}\}f_{\bY^{\ast}\mid\bw}\left(u+y_{2}^{\ast},y_{2}^{\ast}\right)\dif u\right]\dif y_{2}^{\ast},\label{eq: cutting equations 2}\\
& =\int_{\R}2\bone\{y_{2}^{\ast}>0\}\tau_{m}^{-1}\left[\int_{[t-\tau_{m},t]}\bone\{u>-y_{2}^{\ast}\}f_{\bY^{\ast}\mid\bw}\left(u+y_{2}^{\ast},y_{2}^{\ast}\right)\dif u\right]\dif y_{2}^{\ast},\nonumber
\end{align}
where \eqref{eq: cutting equations 2} follows by the same argument as that leading to \eqref{eq: cutting equations 1}, with $(t,t+\tau_m)$ replaced by $(t-\tau_m,t)$. We then proceed as with the proof of right differentiability, where we now use left continuity of $\mathbf 1\{\cdot > -y_2^{\ast}\}$ instead of right continuity of $\mathbf 1\{\cdot \geqslant -y_2^{\ast}\}$ to conclude that
\begin{align*}
\frac{M_{a}(t-\tau_{m},\bw)-M_{a}(t,\bw)}{(-\tau_{m})} & \to2\int_{\R}\bone\{y_{2}^{\ast}>0\}\bone\{t>-y_{2}^{\ast}\}f_{\bY^{\ast}\mid\bw}\left(t+y_{2}^{\ast},y_{2}^{\ast}\right)\dif y_{2}^{\ast}.
\end{align*}
Hence, $M_{a}(\cdot,\bw)$ is \emph{left} \emph{differentiable} at $t$
with \emph{left derivative}
\[
\dot{M}_{a,1-}(t,\bw)=2\int_{\max\{0,-t\}}^{+\infty}f_{\bY^{\ast}\mid\bw}\left(t+y_{2}^{\ast},y_{2}^{\ast}\right)\dif y_{2}^{\ast}.
\]
\textbf{Part a, \emph{two-sided} differentiability:} The left and right
derivatives exist and agree for all $t\in\R$, so $M_{a}(\cdot,\bw)$ is
differentiable with derivative given by
\begin{align*}
\dot{M}_{a,1}(t,\bw) & =2\int_{\max\{0,-t\}}^{+\infty}f_{\bY^{\ast}\mid\bw}\left(t+y_{2}^{\ast},y_{2}^{\ast}\right)\dif y_{2}^{\ast}.
\end{align*}
In case that $t<0$, using the change of variables $z:=t+y_{2}^{\ast}$, we have
$y_{2}^{\ast}=z-t$, and the range of integration becomes $[0,+\infty)$. We can
therefore express this derivative in the (more symmetric looking) form
\begin{align*}
\dot{M}_{a,1}(t,\bw) & =2\int_{0}^{+\infty}f_{\bY^{\ast}\mid\bw}\left(z+\max\{0,t\},z-\min\{0,t\}\right)\dif z.
\end{align*}
\textbf{Part b, \emph{right} differentiability:} Express the function as
\begin{align*}
M_{b}(t,\bw) & =\E_{\bw}\sbr[1]{\bone\{Y_{1}^{\ast}>0\}\bone\{Y_{2}^{\ast}\leqslant0\}\left(\bone\{Y_{1}^{\ast}\leqslant t\}-1\right)}.
\end{align*}
Then using absolute continuity, non-negativity and Tonelli's theorem, we get
\begin{align}
  & \frac{M_{b}(t+\tau_{m},\bw)-M_{b}(t,\bw)}{\tau_{m}}\nonumber\\
  & =\tau_{m}^{-1}\E_{\bw}\sbr[1]{\bone\{Y_{1}^{\ast}>0\}\bone\{Y_{2}^{\ast}\leqslant0\}\left(\bone\{Y_{1}^{\ast}\leqslant t+\tau_{m}\}-\bone\{Y_{1}^{\ast}\leqslant t\}\right)}\nonumber\\
  & =\tau_{m}^{-1}\int_{\R^{2}}\bone\{y_{1}^{\ast}>0\}\bone\{y_{2}^{\ast}\leqslant0\}\left(\bone\{y_{1}^{\ast}\leqslant t+\tau_{m}\}-\bone\{y_{1}^{\ast}\leqslant t\}\right)f_{\bY^{\ast}\mid\bw}\left(\by^{\ast}\right)\dif \by^{\ast}\nonumber\\
  & =\int_{\R}\bone\{y_{2}^{\ast}\leqslant0\}\left[\tau_{m}^{-1}\int_{\R}\bone\{y_{1}^{\ast}>0\}\bone\{y_{1}^{\ast}\in(t,t+\tau_{m}]\}f_{\bY^{\ast}\mid\bw}\left(y_{1}^{\ast},y_{2}^{\ast}\right)\dif y_{1}^{\ast}\right]\dif y_{2}^{\ast}\nonumber\\
  & =\int_{\R}\bone\{y_{2}^{\ast}\leqslant0\}\left[\tau_{m}^{-1}\int_{\R}\bone\{y_{1}^{\ast}\geqslant0\}\bone\{y_{1}^{\ast}\in[t,t+\tau_{m}]\}f_{\bY^{\ast}\mid\bw}\left(y_{1}^{\ast},y_{2}^{\ast}\right)\dif y_{1}^{\ast}\right]\dif y_{2}^{\ast}\nonumber\\
  & =\int_{\R}\bone\{y_{2}^{\ast}\leqslant0\}\left[\tau_{m}^{-1}\int_{[t,t+\tau_{m}]}\bone\{y_{1}^{\ast}\geqslant0\}f_{\bY^{\ast}\mid\bw}\left(y_{1}^{\ast},y_{2}^{\ast}\right)\dif y_{1}^{\ast}\right]\dif y_{2}^{\ast}.\label{eq: cutting equations 3}
\end{align}
Consider the measure space $(\R,\mathcal{B},\lambda)$ and define the
(outer integrand) function $f_{m}$ by
\begin{align*}
f_{m}(y_{2}^{\ast}) & :=\bone\{y_{2}^{\ast}\leqslant0\}\tau_{m}^{-1}\int_{[t,t+\tau_{m}]}\bone\{y_{1}^{\ast}\geqslant0\}f_{\bY^{\ast}\mid\bw}\left(y_{1}^{\ast},y_{2}^{\ast}\right)\dif y_{1}^{\ast}.
\end{align*}
Then $f_{m}$ is non-negative and bounded from above by $g_{m}$ defined by
\[
g_{m}(y_{2}^{\ast}):=\tau_{m}^{-1}\int_{[t,t+\tau_{m}]}f_{\bY^{\ast}\mid\bw}\left(y_{1}^{\ast},y_{2}^{\ast}\right)\dif y_{1}^{\ast}.
\]
As $f_{Y_{1}^{\ast}\mid\bw}(\cdot)=f_{\varepsilon_{1}\mid\bw}(\cdot-a - \bx_1^{\top}\btheta_0)$ is
assumed bounded by a constant $C$ (Assumption \ref{assu:Regularity-TLAD}),
Tonelli's theorem yields
\begin{align*}
\int_{\R}g_{m}(y_{2}^{\ast})\dif y_{2}^{\ast} & =\tau_{m}^{-1}\int_{[t,t+\tau_{m}]}\left[\int_{\R}f_{\bY^{\ast}\mid\bw}\left(y_{1}^{\ast},y_{2}^{\ast}\right)\dif y_{2}^{\ast}\right]\dif y_{1}^{\ast}\\
  & =\tau_{m}^{-1}\int_{[t,t+\tau_{m}]}f_{Y_{1}^{\ast}\mid\bw}(y_{1}^{\ast})\dif y_{1}^{\ast}\leqslant C,
\end{align*}
showing that $g_{m}$ (and thus $f_{m}$) is integrable. Since
$f_{\bY^{\ast}\mid\bw}(\cdot,\cdot) = f_{\bepsilon\mid\bw}(\cdot - a - \bx_1^\top\btheta_0,\cdot - a - \bx_2^\top\btheta_0)$ is 
continuous (Assumption \ref{assu:extra continuity}) and $\bone\{\cdot\geqslant0\}$ is \emph{right} continuous,
both inner integrands $y_{1}^{\ast}\mapsto
f_{\bY^{\ast}\mid\bw}\left(y_{1}^{\ast},y_{2}^{\ast}\right)$
and
$y_{1}^{\ast}\mapsto\bone\{y_{1}^{\ast}\geqslant0\}f_{\bY^{\ast}\mid\bw}\left(y_{1}^{\ast},y_{2}^{\ast}\right)$
are right continuous for each $y_{2}^{\ast}\in\R$. As
$\tau_{m}\to0_{+}$, it follows from right continuity that both
\[
g_{m}(y_{2}^{\ast})\to f_{\bY^{\ast}\mid\bw}\left(t,y_{2}^{\ast}\right)=:g(y_{2}^{\ast}),
\]
pointwise in $y_{2}^{\ast}\in\R$ and
\begin{align*}
f_{m}(y_{2}^{\ast})  & \to\bone\{y_{2}^{\ast}\leqslant0\}\bone\{t\geqslant0\}f_{\bY^{\ast}\mid\bw}\left(t,y_{2}^{\ast}\right)=:f(y_{2}^{\ast})
\end{align*}
pointwise in $y_{2}^{\ast}\in\R$. Also, since $f_{Y_{1}^{\ast}\mid\bw}(\cdot)=f_{\varepsilon_{1}\mid\bw}(\cdot- a - \bx_1^{\top}\btheta_0)$ is continuous (Assumption \ref{assu:extra continuity}),
\begin{align*}
\int_{\R}g_{m}(y_{2}^{\ast})\dif y_{2}^{\ast} & =\tau_{m}^{-1}\int_{[t,t+\tau_{m}]}f_{Y_{1}^{\ast}\mid\bw}(y_{1}^{\ast})\dif y_{1}^{\ast} \to f_{Y_{1}^{\ast}\mid\bw}(t)=\int_{\R}f_{\bY^{\ast}\mid\bw}\left(t,y_{2}^{\ast}\right)\dif y_{2}^{\ast}=\int_{\R}g(y_{2}^{\ast})\dif y_{2}^{\ast}<\infty.
\end{align*}
The GLDCT (Theorem \ref{thm:GLDCT}) therefore shows that $f$ is integrable and
$\int f_{m}\dif \lambda\to\int f\dif \lambda$, the latter convergence
meaning that
\begin{align*}
\frac{M_{b}(t+\tau_{m},\bw)-M_{b}(t,\bw)}{\tau_{m}} & \to\int_{\R}\bone\{y_{2}^{\ast}\leqslant0\}\bone\{t\geqslant0\}f_{\bY^{\ast}\mid\bw}\left(t,y_{2}^{\ast}\right)\dif y_{2}^{\ast}.
\end{align*}
Hence, $M_{b}(\cdot,\bw)$ is \emph{right} \emph{differentiable} at
$t$ with \emph{right derivative}
\[
\dot{M}_{b,1+}(t,\bw)=\bone\{t\geqslant0\}\int_{-\infty}^{0}f_{\bY^{\ast}\mid\bw}\left(t,z\right)\dif z.
\]
\textbf{Part b, \emph{left} differentiability:} Express the difference quotient as
\begin{align}
  & \frac{M_{b}(t-\tau_{m},\bw)-M_{b}(t,\bw)}{\left(-\tau_{m}\right)}\nonumber\\
  & =\left(-\tau_{m}\right)^{-1}\E_{\bw}\sbr[1]{\bone\{Y_{1}^{\ast}>0\}\bone\{Y_{2}^{\ast}\leqslant0\}\left(\bone\{Y_{1}^{\ast}\leqslant t-\tau_{m}\}-\bone\{Y_{1}^{\ast}\leqslant t\}\right)}\nonumber\\
  & =\int_{\R}\bone\{y_{2}^{\ast}\leqslant0\}\left[\tau_{m}^{-1}\int_{[t-\tau_{m},t]}\bone\{y_{1}^{\ast}\geqslant 0\}f_{\bY^{\ast}\mid\bw}\left(y_{1}^{\ast},y_{2}^{\ast}\right)\dif y_{1}^{\ast}\right]\dif y_{2}^{\ast}\label{eq: cutting equations 4}\\
  & =\int_{\R}\bone\{y_{2}^{\ast}\leqslant0\}\left[\tau_{m}^{-1}\int_{[t-\tau_{m},t]}\bone\{y_{1}^{\ast}>0\}f_{\bY^{\ast}\mid\bw}\left(y_{1}^{\ast},y_{2}^{\ast}\right)\dif y_{1}^{\ast}\right]\dif y_{2}^{\ast},\nonumber
\end{align}
where \eqref{eq: cutting equations 4} follows by the same argument as that leading to \eqref{eq: cutting equations 3}, with $(t,t+\tau_m)$ replaced by $(t-\tau_m,t)$. We then proceed as with the proof of right differentiability, where we now use left continuity of $\mathbf 1\{\cdot > 0\}$ instead of right continuity of $\mathbf 1\{\cdot \geqslant 0\}$ to conclude that
\begin{align*}
\frac{M_{b}(t-\tau_{m},\bw)-M_{b}(t,\bw)}{(-\tau_{m})} & \to\int_{\R}\bone\{y_{2}^{\ast}\leqslant0\}\bone\{t>0\}f_{\bY^{\ast}\mid\bw}\left(t,y_{2}^{\ast}\right)\dif y_{2}^{\ast}.
\end{align*}
Hence, $M_{b}(\cdot,\bw)$ is \emph{left} \emph{differentiable}
at $t$ with \emph{left derivative}
\[
\dot{M}_{b,1-}(t,\bw)=\bone\{t>0\}\int_{-\infty}^{0}f_{\bY^{\ast}\mid\bw}\left(t,z\right)\dif z.
\]
\textbf{Part c, \emph{right} differentiability:} Express the function as
\[
M_{c}(t,\bw)=\E_{\bw}\sbr[1]{\bone\{Y_{1}^{\ast}\leqslant0\}\bone\{Y_{2}^{\ast}>0\}\left(1-\bone\{Y_{2}^{\ast}\leqslant-t\}\right)}.
\]
Then using absolute continuity, non-negativity and Tonelli's theorem, we get
\begin{align}
  & \frac{M_{c}(t+\tau_{m},\bw)-M_{c}(t,\bw)}{\tau_{m}}\nonumber\\
  & =\tau_{m}^{-1}\E_{\bw}\sbr[1]{\bone\{Y_{1}^{\ast}\leqslant0\}\bone\{Y_{2}^{\ast}>0\}\left(\bone\{Y_{2}^{\ast}\leqslant-t\}-\bone\{Y_{2}^{\ast}\leqslant-t-\tau_{m}\}\right)}\nonumber\\
  & =\tau_{m}^{-1}\int_{\R^{2}}\bone\{y_{1}^{\ast}\leqslant0\}\bone\{y_{2}^{\ast}>0\}\bone\{y_{2}^{\ast}\in(-t-\tau_{m},-t]\}f_{\bY^{\ast}\mid\bw}\left(\by^{\ast}\right)\dif \by^{\ast}\nonumber\\
  & =\int_{\R}\bone\{y_{1}^{\ast}\leqslant0\}\left[\tau_{m}^{-1}\int_{\R}\bone\{y_{2}^{\ast}>0\}\bone\{y_{2}^{\ast}\in(-t-\tau_{m},-t]\}f_{\bY^{\ast}\mid\bw}\left(y_{1}^{\ast},y_{2}^{\ast}\right)\dif y_{2}^{\ast}\right]\dif y_{1}^{\ast}\nonumber\\
  & =\int_{\R}\bone\{y_{1}^{\ast}\leqslant0\}\left[\tau_{m}^{-1}\int_{[-t-\tau_{m},-t]}\bone\{y_{2}^{\ast}>0\}f_{\bY^{\ast}\mid\bw}\left(y_{1}^{\ast},y_{2}^{\ast}\right)\dif y_{2}^{\ast}\right]\dif y_{1}^{\ast}.\label{eq: cutting equations 5}
\end{align}
Consider the measure space $(\R,\mathcal{B},\lambda)$ and define the
(outer integrand) function $f_{m}$ by
\begin{align*}
f_{m}(y_{1}^{\ast}) & :=\bone\{y_{1}^{\ast}\leqslant0\}\tau_{m}^{-1}\int_{[-t-\tau_{m},-t]}\bone\{y_{2}^{\ast}>0\}f_{\bY^{\ast}\mid\bw}\left(y_{1}^{\ast},y_{2}^{\ast}\right)\dif y_{2}^{\ast}.
\end{align*}
Then $f_{m}$ is non-negative and bounded from above by $g_{m}$ defined by
\[
g_{m}(y_{1}^{\ast}):=\tau_{m}^{-1}\int_{[-t-\tau_{m},-t]}f_{\bY^{\ast}\mid\bw}\left(y_{1}^{\ast},y_{2}^{\ast}\right)\dif y_{2}^{\ast}.
\]
As $f_{Y_{2}^{\ast}\mid\bw}(\cdot)=f_{\varepsilon_{2}\mid\bw}(\cdot - a - \bx_2^{\top}\btheta_0)$ is bounded by $C$ (Assumption \ref{assu:Regularity-TLAD}),
Tonelli's theorem yields
\begin{align*}
\int_{\R}g_{m}(y_{1}^{\ast})\dif y_{1}^{\ast} & =\tau_{m}^{-1}\int_{[-t-\tau_{m},-t]}\left[\int_{\R}f_{\bY^{\ast}\mid\bw}\left(y_{1}^{\ast},y_{2}^{\ast}\right)\dif y_{1}^{\ast}\right]\dif y_{2}^{\ast}\\
  & =\tau_{m}^{-1}\int_{[-t-\tau_{m},-t]}f_{Y_{2}^{\ast}\mid\bw}(y_{2}^{\ast})\dif y_{2}^{\ast}\leqslant C,
\end{align*}
showing that $g_{m}$ (and thus $f_{m}$) is integrable. Since
$f_{\bY^{\ast}\mid\bw}(\cdot,\cdot) = f_{\bepsilon\mid\bw}(\cdot - a - \bx_1^\top\btheta_0,\cdot - a - \bx_2^\top\btheta_0)$ is 
continuous (Assumption \ref{assu:extra continuity}) and $\bone\{\cdot>0\}$ is \emph{left} continuous, both inner
integrands $y_{2}^{\ast}\mapsto
f_{\bY^{\ast}\mid\bw}\left(y_{1}^{\ast},y_{2}^{\ast}\right)$
and
$y_{2}^{\ast}\mapsto\bone\{y_{2}^{\ast}>0\}f_{\bY^{\ast}\mid\bw}\left(y_{1}^{\ast},y_{2}^{\ast}\right)$
are left continuous for each $y_{1}^{\ast}\in\R$. As $\tau_{m}\to0_{+}$,
it follows from left continuity that both
\[
g_{m}(y_{1}^{\ast})\to f_{\bY^{\ast}\mid\bw}\left(y_{1}^{\ast},-t\right)=:g(y_{1}^{\ast})
\]
pointwise in $y_{1}^{\ast}\in\R$ and
\begin{align*}
f_{m}(y_{1}^{\ast})  & \to\bone\{y_{1}^{\ast}\leqslant0\}\bone\{-t>0\}f_{\bY^{\ast}\mid\bw}\left(y_{1}^{\ast},-t\right)=:f(y_{1}^{\ast})
\end{align*}
pointwise in $y_{1}^{\ast}\in\R$. Also, since
$f_{Y_{2}^{\ast}\mid\bw}(\cdot)=f_{\varepsilon_{2}\mid\bw}(\cdot - a -
\bx_2^{\top}\btheta_0)$ is continuous (Assumption \ref{assu:extra continuity}),
\begin{align*}
\int_{\R}g_{m}(y_{1}^{\ast})\dif y_{1}^{\ast} & =\tau_{m}^{-1}\int_{[-t-\tau_{m},-t]}f_{Y_{2}^{\ast}\mid\bw}(y_{2}^{\ast})\dif y_{2}^{\ast}\\
  & \to f_{Y_{2}^{\ast}\mid\bw}(-t)=\int_{\R}f_{\bY^{\ast}\mid\bw}\left(y_{1}^{\ast},-t\right)\dif y_{1}^{\ast}=\int_{\R}g(y_{1}^{\ast})\dif y_{1}^{\ast}<\infty.
\end{align*}
The GLDCT (Theorem \ref{thm:GLDCT}) now shows that $f$ integrable and $\int
f_{m}\dif \lambda\to\int f\dif \lambda$, the latter convergence
meaning that
\begin{align*}
\frac{M_{c}(t+\tau_{m},\bw)-M_{c}(t,\bw)}{\tau_{m}} & \to\int_{\R}\bone\{y_{1}^{\ast}\leqslant0\}\bone\{-t>0\}f_{\bY^{\ast}\mid\bw}\left(y_{1}^{\ast},-t\right)\dif y_{1}^{\ast},
\end{align*}
showing that $M_{c}(\cdot,\bw)$ is \emph{right} \emph{differentiable}
at $t$ with \emph{right derivative}
\[
\dot{M}_{c,1+}(t,\bw)=\bone\{t<0\}\int_{-\infty}^{0}f_{\bY^{\ast}\mid\bw}\left(z,-t\right)\dif z.
\]
\textbf{Part c, \emph{left} differentiability:} Express the difference quotient as
\begin{align}
  & \frac{M_{c}(t-\tau_{m},\bw)-M_{c}(t,\bw)}{\left(-\tau_{m}\right)}\nonumber\\
   & =\left(-\tau_{m}\right)^{-1}\E_{\bw}\sbr[1]{\bone\{Y_{1}^{\ast}\leqslant0\}\bone\{Y_{2}^{\ast}>0\}\left(\bone\{Y_{2}^{\ast}\leqslant-t\}-\bone\{Y_{2}^{\ast}\leqslant-t+\tau_{m}\}\right)}\nonumber\\
  & =\int_{\R}\bone\{y_{1}^{\ast}\leqslant0\}\left[\tau_{m}^{-1}\int_{[-t,-t+\tau_{m}]}\bone\{y_{2}^{\ast}>0\}f_{\bY^{\ast}\mid\bw}\left(y_{1}^{\ast},y_{2}^{\ast}\right)\dif y_{2}^{\ast}\right]\dif y_{1}^{\ast} \label{eq: cutting equations 6}\\ 
  & =\int_{\R}\bone\{y_{1}^{\ast}\leqslant0\}\left[\tau_{m}^{-1}\int_{[-t,-t+\tau_{m}]}\bone\{y_{2}^{\ast}\geqslant0\}f_{\bY^{\ast}\mid\bw}\left(y_{1}^{\ast},y_{2}^{\ast}\right)\dif y_{2}^{\ast}\right]\dif y_{1}^{\ast}.\nonumber
\end{align}
where \eqref{eq: cutting equations 6} follows by the same argument as that leading to \eqref{eq: cutting equations 5}, with $(-t-\tau_m,-t)$ replaced by $(-t,-t+\tau_m)$. We then proceed as with the proof of right differentiability, where we now use right continuity of $\mathbf 1\{\cdot \geqslant 0\}$ instead of left continuity of $\mathbf 1\{\cdot > 0\}$ to conclude that
\begin{align*}
\frac{M_{c}(t-\tau_{m},\bw)-M_{c}(t,\bw)}{\left(-\tau_{m}\right)} & \to\int_{\R}\bone\{y_{1}^{\ast}\leqslant0\}\bone\{-t\geqslant0\}f_{\bY^{\ast}\mid\bw}\left(y_{1}^{\ast},-t\right)\dif y_{1}^{\ast}.
\end{align*}
Hence, $M_{c}(\cdot,\bw)$ is \emph{left} \emph{differentiable} at $t$
with \emph{left derivative}
\[
\dot{M}_{c,1-}(t,\bw)=\bone\{t\leqslant0\}\int_{-\infty}^{0}f_{\bY^{\ast}\mid\bw}\left(z,-t\right)\dif z.
\]
\textbf{Harvesting our results:} We see that $M(\cdot,\bw)$ is
semi-differentiable with \emph{right} derivative given by
\begin{align*}
\dot{M}_{1+}(t,\bw) & =\dot{M}_{a,1+}(t,\bw)+\dot{M}_{b,1+}(t,\bw)+\dot{M}_{c,1+}(t,\bw)\\
  & =2\int_{0}^{+\infty}f_{\bY^{\ast}\mid\bw}\left(z+\max\{0,t\},z-\min\{0,t\}\right)\dif z\\
  & \qquad+\bone\{t\geqslant0\}\int_{-\infty}^{0}f_{\bY^{\ast}\mid\bw}\left(t,z\right)\dif z+\bone\{t<0\}\int_{-\infty}^{0}f_{\bY^{\ast}\mid\bw}\left(z,-t\right)\dif z\\
  & =\begin{cases}
2\int_{0}^{+\infty}f_{\bY^{\ast}\mid\bw}\left(z,z-t\right)\dif z+\int_{-\infty}^{0}f_{\bY^{\ast}\mid\bw}\left(z,-t\right)\dif z, & t<0,\\
2\int_{0}^{+\infty}f_{\bY^{\ast}\mid\bw}\left(z,z\right)\dif z+\int_{-\infty}^{0}f_{\bY^{\ast}\mid\bw}\left(0,z\right)\dif z, & t=0,\\
2\int_{0}^{+\infty}f_{\bY^{\ast}\mid\bw}\left(z+t,z\right)\dif z+\int_{-\infty}^{0}f_{\bY^{\ast}\mid\bw}\left(t,z\right)\dif z, & t>0,
\end{cases}
\end{align*}
and \emph{left} derivative given by
\begin{align*}
\dot{M}_{1-}(t,\bw) & =\dot{M}_{a,1-}(t,\bw)+\dot{M}_{b,1-}(t,\bw)+\dot{M}_{c,1-}(t,\bw)\\
  & =2\int_{0}^{+\infty}f_{\bY^{\ast}\mid\bw}\left(z+\max\{0,t\},z-\min\{0,t\}\right)\dif z\\
  & \qquad+\bone\{t>0\}\int_{-\infty}^{0}f_{\bY^{\ast}\mid\bw}\left(t,z\right)\dif z+\bone\{t\leqslant0\}\int_{-\infty}^{0}f_{\bY^{\ast}\mid\bw}\left(z,-t\right)\dif z\\
  & =\begin{cases}
2\int_{0}^{+\infty}f_{\bY^{\ast}\mid\bw}\left(z,z-t\right)\dif z+\int_{-\infty}^{0}f_{\bY^{\ast}\mid\bw}\left(z,-t\right)\dif z, & t<0,\\
2\int_{0}^{+\infty}f_{\bY^{\ast}\mid\bw}\left(z,z\right)\dif z+\int_{-\infty}^{0}f_{\bY^{\ast}\mid\bw}\left(z,0\right)\dif z, & t=0,\\
2\int_{0}^{+\infty}f_{\bY^{\ast}\mid\bw}\left(z+t,z\right)\dif z+\int_{-\infty}^{0}f_{\bY^{\ast}\mid\bw}\left(t,z\right)\dif z, & t>0.
\end{cases}
\end{align*}

Thus, for $t\neq0$ the left and right derivatives
agree, meaning that $M(\cdot,\bw)[=\dot{\ell}_1(\cdot,\bw)]$ is (fully)
differentiable. For $t=0$, $M(\cdot,\bw)$ is generally only
semi-differentiable with left and right derivatives as in
\eqref{eq:SecondOrderLeftDerivativeExpectedTrimmedAbsLoss} and
\eqref{eq:SecondOrderRightDerivativeExpectedTrimmedAbsLoss}, respectively. However, for $\bw$ such that
$\Delta\bx^\top\btheta_0=0$, the (conditional) exchangeability of
$\varepsilon_1$ and $\varepsilon_2$ implies symmetry of
$f_{\bY^{\ast}\mid\bw}(\cdot,\cdot)$ so that the left and right derivatives
agree also at $\Delta\bx^\top\btheta_0=0$. Hence, $M(\cdot,\bw)$ is
differentiable at $\Delta\bx^\top\btheta_0$.
\end{proof}

\subsection{Proof of Theorem \ref{thm:AsymptoticNormality-TLAD}}
\begin{proof}[\sc{Proof of Theorem \ref{thm:AsymptoticNormality-TLAD}}] 
As in \citet{honore_trimmed_1992}, we set up for an application of
\citet[Theorem 3.3]{pakes1989simulation}. Following the proof of \citet[Theorem
2(iii)]{honore_trimmed_1992}, we verify all conditions of \citet[Theorem
3.3]{pakes1989simulation} except for their condition (ii). Our Theorem
\ref{thm:LossHessianExistence-TLAD} shows that the Hessian
$\bfH_0^{\tt{tlad}}=\nabla^2L(\btheta_0)$ of the population (trimmed absolute)
loss $L$ at $\btheta_0$ exists. Since this matrix is assumed to be of full rank,
the desired condition (ii) follows with the gradient mapping $\nabla
L:\R^K\to\R^K$ playing the role of $G$ in \citet[Theorem
3.3]{pakes1989simulation}.
\end{proof}

\subsection{Proof of Theorem \ref{thm: bootstrap consistency lad}}
\begin{proof}[\sc{Proof of Theorem \ref{thm: bootstrap consistency lad}}] 
We apply \citet[Remark~3.3]{AG92} to establish bootstrap consistency for the
TLAD estimator. In our notation, this remark delivers conditional weak
convergence (in the bounded-Lipschitz metric, in probability) of the bootstrap
law of the shifted and scaled
$\sqrt{n}(\tbtheta^{\tt{tlad}}-\widehat{\btheta}^{\tt{tlad}})$ to the
same Gaussian limit as in Theorem \ref{thm:AsymptoticNormality-TLAD}, provided
that conditions (A.1)$'$, (A.2)$'$, (A.4)$'$, (A.5), and (A.6) of \citet{AG92}
hold, and provided that their stochastic equicontinuity condition
\begin{equation}\label{eq: AG92 stochastic equicontinuity condition}
  \lim_{\overline{\delta}\to0_{+}}\limsup_{n\to\infty}
  \E\sbr[2]{\sup_{\|\bdelta\|_2\leqslant\overline{\delta}}|\bbG_n(r(\cdot,\bdelta))|}=0,
\end{equation}
holds for the empirical process $\bbG_n=\sqrt n(P_n-P)$ and certain remainder
functions $r(\cdot,\bdelta)$ which we define in \eqref{eq: remainder function
for TLAD bootstrap consistency} below.

\citet[Example~4.10]{AG92} gives a convenient route to verify (A.5) and (A.6)
via auxiliary conditions (B.1)--(B.4). However, their condition (B.2) is too
restrictive in our setting. In particular, it fails when the criterion depends
on regressors through a nontrivial linear index in dimension $K>1$.\footnote{The
problem with the strength of condition (B.2) in \cite{AG92} was overlooked in
\cite{H95}, who used \citet{AG92} results to establish bootstrap validity for
the quantile regression estimator. However, the conclusions in \cite{H95} are
valid as condition (B.2) can be bypassed using an argument parallel to the one
we use below.} We therefore verify (A.5) and (A.6) directly, using the TLAD
structure and a Glivenko--Cantelli argument, while checking (A.1)$'$, (A.2)$'$,
(A.4)$'$, and \eqref{eq: AG92 stochastic equicontinuity condition} by empirical
process methods. Once these conditions are established, \citet[Remark~3.3]{AG92}
yields the desired bootstrap convergence. The consistency of
$\widehat{\bSigma}^{\tt{tlad}}$ then follows by applying the quantile mapping
implied by the Gaussian limit to the one- and two-dimensional linear
combinations used in the construction of $\widehat{\bSigma}^{\tt{tlad}}$.

Setting up for \citet[Remark 3.3]{AG92}, let $\bZ:=(Y_1,\bX_1,Y_2,\bX_2)$,
$\mathcal Z:=[0,\infty)\times \R^K\times[0,\infty)\times\R^K$, and for any
$\bz=(y_1,\bx_1, y_2,\bx_2)\in\mathcal Z$, let $\Delta\bx := \bx_1 - \bx_2$ and
$\by := (y_1,y_2)$. Also, let $\bZ_i:=(Y_{i1},\bX_{i1},Y_{i2},\bX_{i2})$ and
$\tbZ_i:=(\widetilde{Y}_{i1},\tbX_{i1},\widetilde{Y}_{i2},\tbX_{i2})$ for all
$i\in[n]$. We first verify conditions (A.1)$'$, (A.2)$'$, (A.4)$'$ and
\eqref{eq: AG92 stochastic equicontinuity condition} from \cite{AG92}, in turn,
with their $S$, $\Theta$, $x$, $\theta\mapsto g(x,\theta)$, and $G$ replaced by
our $\mathcal Z$, $\R^K-\{\btheta_0\}$, $\bz$,
$\bdelta\mapsto-\widetilde{m}^{\texttt{tlad}}(\Delta\bx^\top(\btheta_0+\bdelta),\by)$,
and
$\bdelta\mapsto\E[-\widetilde{m}^{\texttt{tlad}}(\Delta\bX^\top(\btheta_0+\bdelta),\bY)]$,
respectively, with $\widetilde m^{\texttt{tlad}}$ defined in \eqref{eq: mtlad
alternative expressions} and $\Delta\bX:=\bX_1 - \bX_2$. Note that their
$\theta$ plays the role of a deviation $\btheta - \btheta_0$ from our
$\btheta_0$. We denote such deviations by $\bdelta$.

\medskip\noindent\textbf{Condition (A.1)}$\boldsymbol{'}$: \citet[Lemma
B.2]{honore_trimmed_1992} (for his $S_n^\circ$) shows that $\btheta_0$ uniquely
minimizes $\btheta\mapsto\E[\widetilde
m^{\texttt{tlad}}(\Delta\bX^\top\btheta,\bY)]$, which means that
$\bdelta\mapsto\E$$[-\widetilde
m^{\texttt{tlad}}(\Delta\bX^\top(\btheta_0+\bdelta),\bY)]$ is uniquely maximized
at $\bdelta=\mathbf{0}$. In particular, (A.1)$'$ holds. 

\medskip\noindent\textbf{Condition (A.2)}$\boldsymbol{'}$ holds by optimality of
$\btheta_0$, second-order differentiability of the convex TLAD population loss
$L$ at $\btheta_0$ (Theorem \ref{thm:LossHessianExistence-TLAD}), and the
assumption that the Hessian $\bfH_0^{\tt{tlad}}=\nabla^2 L(\btheta_0)$ is
nonsingular. Since $L$ is convex, $\nabla^2 L(\btheta_0)$ is positive
semidefinite, and nonsingularity therefore implies that $\bfH_0^{\tt{tlad}}$ is
positive definite.

\medskip\noindent\textbf{Condition (A.4)}$\boldsymbol{'}$: To verify (A.4)$'$,
define the function $\Delta:\cZ\to\R^K$ by
\begin{align*}
\Delta(\bz)& := \bone\{y_{1}>0\}\bone\{y_{2}>0\}(\bone\{y_1-y_2<\Delta\bx^\top\btheta_0\} - \bone\{y_1-y_2>\Delta\bx^\top\btheta_0\})\Delta\bx\\
 & \quad-\bone\{y_{1}>0\}\bone\{y_{2}=0\}\bone\{y_1>\Delta\bx^\top\btheta_0\}\Delta\bx\\
 & \quad+\bone\{y_{1}=0\}\bone\{y_{2}>0\}\bone\{y_{2}>-\Delta\bx^\top\btheta_0\}\Delta\bx
\end{align*}
and define the function $r:\cZ\times\R^K\to\R$ by
\begin{equation}\label{eq: remainder function for TLAD bootstrap consistency}
  r(\bz,\bdelta):= 
  \begin{cases}
    \left[\widetilde m^{\texttt{tlad}}(\Delta\bx^\top(\btheta_0+\bdelta),\by) - \widetilde m^{\texttt{tlad}}(\Delta\bx^\top\btheta_0,\by) - \bdelta^\top\Delta(\bz)\right]/\|\bdelta\|_2, &\bdelta\neq\mathbf{0}, \\
    0, &\bdelta=\mathbf{0}.
  \end{cases}
\end{equation}
We will show that the function classes
$\cR_{\odelta}:=\{r(\cdot,\bdelta)|\bdelta\in\R^K,\|\bdelta\|_2\leqslant\odelta\},\odelta\in(0,\infty)$,
admit square integrable envelopes
$R_{\odelta}:\cZ\to[0,\infty),\odelta\in(0,\infty)$, satisfying
$\lim_{\odelta\to0_+}\E[R_{\odelta}(\bZ)^2]=0$.

To this end, denote $\Delta y := y_1 - y_2$, and split $r(\cdot,\bdelta)$ into
three parts as
\begin{align*}
r(\bz,\bdelta) & = r_{11}(\bz,\bdelta) + r_{10}(\bz,\bdelta) + r_{01}(\bz,\bdelta),
\end{align*}
according to the three cases $\{y_1>0,y_2>0\}$, $\{y_1>0,y_2=0\}$, and
$\{y_1=0,y_2>0\}$ in the definition of both $\widetilde m^{\texttt{tlad}}$ and
$\Delta$. Abbreviate $e_{12}:=\Delta y - \Delta\bx^\top\btheta_0$, $e_1:= y_1 -
\Delta\bx^\top\btheta_0$, and $e_2:= y_2 + \Delta\bx^\top\btheta_0$, which are
all linear combinations of $y_1,y_2$, and $\Delta\bx$. Then inserting the
definitions of $\widetilde m^{\texttt{tlad}}$ and $\Delta(\cdot)$, we can
express the three parts as, respectively,
\begin{align*}
  r_{11}(\bz,\bdelta) & = 
  \bone\{y_{1}>0,y_{2}>0\}\left[|e_{12} - \Delta\bx^\top\bdelta|-|e_{12}|+\left(\bone\{e_{12}>0\} - \bone\{e_{12}<0\}\right)\Delta\bx^\top\bdelta\right]/\|\bdelta\|_2,\\
  r_{10}(\bz,\bdelta) & = 
  \bone\{y_{1}>0,y_{2}=0\} \left[\max\{0,e_1 - \Delta\bx^\top\bdelta\}-\max\{0,e_1\}+\bone\{e_1>0\}\Delta\bx^\top\bdelta\right]/\|\bdelta\|_2,\\
  r_{01}(\bz,\bdelta) & = 
  \bone\{y_{1}=0,y_{2}>0\}\left[\max\{0,e_2 + \Delta\bx^\top\bdelta\}-\max\{0,e_2\}-\bone\{e_2>0\}\Delta\bx^\top\bdelta\right]/\|\bdelta\|_2.
\end{align*}
if $\bdelta\neq\mathbf{0}$, and each part is understood as identically zero if
$\bdelta=\mathbf{0}$. We upper bound $r_{11}(\cdot,\bdelta)$,
$r_{10}(\cdot,\bdelta)$, and $r_{01}(\cdot,\bdelta)$ in turn. Splitting into the
two cases $|u|>|t|$ and $|u|\leqslant|t|$, one can derive the elementary
inequalities
\begin{align*}
  &\left||u-t| - |u| + (\bone\{u>0\} - \bone\{u<0\})t\right| 
  \leqslant 2|t|\bone\{|u|\leqslant|t|\},\quad u,t\in\R,\\
 & \left|\max\{0,u-t\}-\max\{0,u\}+\bone\{u>0\}t\right|
  \leqslant 2|t|\bone\{|u|\leqslant|t|\},\quad u,t\in\R.
\end{align*}
The \emph{first} elementary inequality followed by the Cauchy-Schwarz inequality
yields
\[
|r_{11}(\bz,\bdelta)|\leqslant 2 \|\Delta\bx\|_2\bone\{y_1>0,y_2>0\}\bone\{|e_{12}|\leqslant \|\Delta\bx\|_2\|\bdelta\|_2\},
\]
while the \emph{second} elementary inequality followed by the Cauchy-Schwarz
inequality yields
\begin{align*}
|r_{10}(\bz,\bdelta)| & \leqslant 2\|\Delta\bx\|_2\bone\{y_1>0,y_2=0\}\bone\{|e_1|\leqslant \|\Delta\bx\|_2\|\bdelta\|_2\}\quad\text{and}\\
|r_{01}(\bz,\bdelta)| & \leqslant 2\|\Delta\bx\|_2\bone\{y_1=0,y_2>0\}\bone\{|e_2|\leqslant \|\Delta\bx\|_2\|\bdelta\|_2\}.
\end{align*}
Combining these three bounds, we see that $\cR_{\odelta}$ admits the envelope
\begin{equation}\label{eq: R odelta envelope}
\left.
\begin{aligned}
R_{\odelta}(\bz) & := 2\|\Delta\bx\|_2\bone\{y_1>0,y_2>0\}\bone\{|\Delta y - \Delta\bx^\top\btheta_0|\leqslant \|\Delta\bx\|_2\odelta\}\\
 & \quad + 2\|\Delta\bx\|_2\bone\{y_1>0,y_2=0\}\bone\{|y_1 - \Delta\bx^\top\btheta_0|\leqslant \|\Delta\bx\|_2\odelta\}\\
 & \quad + 2\|\Delta\bx\|_2\bone\{y_1=0,y_2>0\}\bone\{|y_2 + \Delta\bx^\top\btheta_0|\leqslant \|\Delta\bx\|_2\odelta\}.
\end{aligned}
\right\}
\end{equation}
The functions $\{R_{\odelta}(\cdot)^2\}_{\odelta\in(0,\infty)}$ are dominated by
$\bz\mapsto 4\|\Delta\bx\|_2^2$, which is integrable by Assumption
\ref{assu:MomentConditions-TLAD}. Consider the expectation of the square of the
first contribution to $R_{\odelta}(\bz)$ in \eqref{eq: R odelta envelope}.
Iterating expectations, we can express this expectation as
\begin{align*}
& \E\left[4\|\Delta\bX\|_2^2\bone\{Y_1>0,Y_2>0\}\bone\{|\Delta Y - \Delta\bX^\top\btheta_0|\leqslant \|\Delta\bX\|_2\odelta\}\right]\\
&= \E\left[4\|\Delta\bX\|_2^2\bbP\del[1]{Y_1>0,Y_2>0,|\Delta Y - \Delta\bX^\top\btheta_0|\leqslant \|\Delta\bX\|_2\odelta\mid\bW}\right].
\end{align*}
Using Assumption \ref{assu:Regularity-TLAD}, the inner probability is almost
surely bounded by
\begin{align*}
\bbP\del[1]{|\Delta\varepsilon|\leqslant \|\Delta\bX\|_2\odelta\mid\bW}=\int_{-\|\Delta\bX\|_2\odelta}^{\|\Delta\bX\|_2\odelta}f_{\Delta\varepsilon\mid\bW}(t)\dif t
\leqslant 2C\|\Delta\bX\|_2\odelta,
\end{align*}
which goes to zero as $\odelta\to0_+$. Similar reasoning applies to the squares
of the second and third contributions to $R_{\odelta}(\bz)$ in \eqref{eq: R
odelta envelope}. It therefore follows from the Lebesgue Dominated Convergence
Theorem that $\E[R_{\odelta}(\bZ)^2]\to0$ as $\odelta\to0_+$. Hence, (A.4)$'$
holds.

\medskip\noindent\textbf{Condition \eqref{eq: AG92 stochastic equicontinuity
condition}}: To verify \eqref{eq: AG92 stochastic equicontinuity condition}, we
will show that the function class
$\cR:=\cR_{\infty}=\{r(\cdot,\bdelta)|\bdelta\in\R^K\}$ is VC-subgraph. The
desired \eqref{eq: AG92 stochastic equicontinuity condition} will then follow
from a (local) maximal inequality for VC-subgraph classes and the
$L^2(P)$-continuity provided by Condition (A.4)'. To argue that $\cR$ is
VC-subgraph, define the function class
\begin{align*}
  \cG&:=\left\{g:\cZ\times\R\to\R\middle|g=g(\cdot,\cdot;\gamma,\gamma_1,\gamma_2,\bdelta),(\gamma,\gamma_1,\gamma_2,\bdelta)\in\R^{3+K}\right\},\\
  g(\bz,s;\gamma,\gamma_1,\gamma_2,\bdelta)&:=\gamma s+\gamma_1y_1+\gamma_2y_2+\Delta\bx^\top \bdelta.
\end{align*}
Then $\cG$ forms a vector space of real-valued measurable functions of dimension
$3+K$. \citet[Lemma 2.6.15]{vdVW1996weak} shows that $\cG$ is VC-subgraph (of VC
index at most $5+K$), where the \emph{subgraph} of any function $f:\cZ\to\R$ is
defined as the area below its graph:
\[
\mathrm{subgraph}(f):=\left\{(\bz,s)\in\cZ\times\R\middle|s<f(\bz)\right\}.
\]
It then follows from \citet[Lemma 2.4]{pakes1989simulation} that the sets of the
form $\{g\geqslant r\}$ or $\{g>r\}$ with $g\in\cG$ and $r\in\R$ form a VC
class. Call it $\cC$. 

For any $\bdelta\in\R^K$, the subgraph of $r(\cdot,\bdelta)$ can be
expressed as a
\begin{align*}
 \text{subgraph}(r(\cdot,\bdelta))
 &=\left(\{y_1>0,y_2>0\}\cap\text{subgraph}(r_{11}(\cdot,\bdelta))\right)\\
 &\quad\cup\left(\{y_1>0,y_2=0\}\cap\text{subgraph}(r_{10}(\cdot,\bdelta))\right)\\
 &\quad\cup\left(\{y_1=0,y_2>0\}\cap\text{subgraph}(r_{01}(\cdot,\bdelta))\right)\\
 &\quad\cup\left(\{y_1=0,y_2=0\}\cap\{s\geqslant0\}^c\right).
\end{align*}
For each case, the corresponding subgraph can be expressed as a union of (as
most) four intersections of (at most three) sets in $\cC$. For example, for the
case $\{y_1>0,y_2>0\}$, when $\bdelta\neq\bzero$, the subgraph of
$r_{11}(\cdot,\bdelta)$ is given by the solutions $(\bz,s)$ to the inequality
\[
  |e_{12} - \Delta\bx^\top\bdelta|-|e_{12}|+\left(\bone\{e_{12}>0\} - \bone\{e_{12}<0\}\right)\Delta\bx^\top\bdelta - \|\bdelta\|_2 s > 0.
\]
For given signs of $e_{12}=\Delta y - \Delta\bx^\top\btheta_0$ and $e_{12} -
\Delta\bx^\top\bdelta$, the above inequality is (piecewise) linear in
$s,y_1,y_2$, and $\Delta\bx$, and thus of the form $\{g_1>0\}$ with $g_1\in\cG$.
The signs (nonnegative or negative) of $e_{12}$ and $e_{12} -
\Delta\bx^\top\bdelta$ can themselves be expressed in the form
$\{g_2\geqslant0\}$ (or $\{g_2>0\}$) and $\{g_3\geqslant0\}$ (or $\{g_3>0\}$)
with $g_2,g_3\in\cG$, which means that the subgraph of $r_{11}(\cdot,\bdelta)$
can be written in the form $C_1\cap C_2 \cap C_3$ with $C_1,C_2,C_3\in\cC$. The
other cases are similar. It follows from the VC permanence properties in
\citet[Lemma 2.5]{pakes1989simulation} that the subgraphs of $\cR$ form a VC
class, meaning that $\cR$ is VC-subgraph. 

The $\{\cR_{\odelta}\}_{\odelta\in(0,\infty)}$ inherit the VC-subgraph property
from $\cR$.  It therefore follows from \citet[Theorem 2.14.1]{vdVW1996weak}
(with their $p=1$) that for any $n\in\N$ and any $\odelta\in(0,\infty)$, we have
$
  \E[\|\bbG_n\|_{\cR_{\odelta}}]\leqslant A J(1,\cR_{\odelta})\sqrt{\E[R_{\odelta}(\bZ)^2]},
$
where $A\in(0,\infty)$ is a universal constant, and the uniform covering
integral $J(1,\cR_{\odelta})$ \citep[see][p.~239]{vdVW1996weak} of
$\cR_{\odelta}$ is relative to the envelope $R_{\odelta}(\cdot)$ in \eqref{eq: R
odelta envelope}. As the $\{\cR_{\odelta}\}_{\odelta\in(0,\infty)}$ also inherit
a common VC index $V(\cR)$ from $\cR$, the VC-subgraph property implies
$\sup_{\odelta\in(0,\infty)}J(1,\cR_{\odelta})<\infty$. As we have already
argued $\E[R_{\odelta}(\bZ)^2]\to0$ as $\odelta\to0_+$, \eqref{eq: AG92
stochastic equicontinuity condition} follows from the previous display.

Next, we verify conditions (A.5) and (A.6).

\medskip\noindent\textbf{Condition (A.5)} amounts to
$\widehat{\btheta}^{\texttt{tlad}}\to_{\mathrm{a.s.}}\btheta_0$. \citet[Theorem
1(iii)]{honore_trimmed_1992} establishes strong consistency of TLAD under
conditions implied by the assumptions of our extended asymptotic normality
result (Theorem \ref{thm:AsymptoticNormality-TLAD}), so (A.5) holds.

\medskip\noindent\textbf{Condition (A.6)}: Let $b(\bdelta):=\|\bdelta\|_2$ and
define
\[
q(\bz,\bdelta)
:=
\widetilde m^{\texttt{tlad}}
(\Delta\bx^\top(\btheta_0+\bdelta),\by)
-\widetilde m^{\texttt{tlad}}
(\Delta\bx^\top\btheta_0,\by),
\quad
(\bz,\bdelta)\in\cZ\times\R^K,
\]
so that $\bzero$ uniquely minimizes $\bdelta\mapsto\E[q(\bZ,\bdelta)]$, the
minimum being zero. Condition (A.6) will follow once this population objective
is uniformly separated from its minimum away from $\bzero$ (Step 1) and the
bootstrap criterion uniformly approximates its population counterpart (Step 2). 

\textbf{Step 1:} Fix $\epsilon\in(0,\infty)$ and denote
$
C_\epsilon:=\{\bdelta\in\R^K:\|\bdelta\|_2\geqslant\epsilon\}.
$
We first show that
\begin{equation}\label{eq:key_identifiability_A6}
a_\epsilon
:=
\inf_{\bdelta\in C_\epsilon}
\E[q(\bZ,\bdelta)]/b(\bdelta)
>0,
\end{equation}
which amounts to well-separatedness of $\btheta_0$. In what follows, we first
decompose $\E[q(\bZ,\bdelta)]$, then extract its dominant component as
deviations from $\btheta_0$ grow in magnitude, and finally argue
well-separatedness using our rank condition.

Using the representation \eqref{eq: mtlad alternative expressions} of the
shifted TLAD loss, $\E[q(\bZ,\bdelta)]$ decomposes into three contributions:
\begin{align*}
\E[q(\bZ,\bdelta)]
&=
\E[
\bone\{Y_1>0,Y_2>0\}
(
|\Delta\varepsilon-\Delta\bX^\top\bdelta|
-|\Delta\varepsilon|
)
]
\\
&\quad+
\E[
\bone\{Y_1>0,Y_2=0\}
(
\max\{0,E_1-\Delta\bX^\top\bdelta\}-\max\{0,E_1\}
)
]
\\
&\quad+
\E[
\bone\{Y_1=0,Y_2>0\}
(
\max\{0,E_2+\Delta\bX^\top\bdelta\}-\max\{0,E_2\}
)
],
\end{align*}
where $E_1:=Y_1-\Delta\bX^\top\btheta_0$ and $E_2:=Y_2+\Delta\bX^\top\btheta_0$.
Let $\{\bdelta_\ell\}_{\ell=1}^\infty\subset\R^K\backslash\{\bzero\}$ be any
sequence satisfying $\|\bdelta_\ell\|_2\to\infty$. Then, as $\ell\to\infty$,
$
|\Delta\varepsilon-\Delta\bX^\top\bdelta_\ell|
=
|\Delta\bX^\top\bdelta_\ell|+o(\|\bdelta_\ell\|_2)\;\text{a.s.}
$
and therefore
$
|\Delta\varepsilon-\Delta\bX^\top\bdelta_\ell|
-|\Delta\varepsilon|
-|\Delta\bX^\top\bdelta_\ell|
\to o(\|\bdelta_\ell\|_2)
\;\text{a.s.}
$
Moreover, the elementary inequality
\(
\envert[1]{
|u-t|-|u|-|t|
}
\leqslant 2|u|
\)
and the Cauchy-Schwarz inequality together yield
\[
\left|
\frac{
|\Delta\varepsilon-\Delta\bX^\top\bdelta_\ell|
-|\Delta\varepsilon|
-|\Delta\bX^\top\bdelta_\ell|
}{\|\bdelta_\ell\|_2}
\right|
\leqslant
2\|\Delta\bX\|_2.
\]
Since $\E[\|\Delta\bX\|_2]<\infty$
(Assumption~\ref{assu:MomentConditions-TLAD}), by Lebesgue Dominated
Convergence Theorem,
\begin{align}\label{eq:A6_interior_expansion}
&\E[
\bone\{Y_1>0,Y_2>0\}
(
|\Delta\varepsilon-\Delta\bX^\top\bdelta_\ell|
-|\Delta\varepsilon|
)]\nonumber\\
&\qquad=
\E[\bone\{Y_1>0,Y_2>0\}
|\Delta\bX^\top\bdelta_\ell|]
+o(\|\bdelta_\ell\|_2)
\end{align}
as $\ell\to\infty$. Similar reasoning applied to the latter two parts of
$\E[q(\bZ,\bdelta)]$ gives both
\begin{align}
&\E[\bone\{Y_1>0,Y_2=0\}
(\max\{0,E_1-\Delta\bX^\top\bdelta_\ell\}-\max\{0,E_1\})]
\nonumber\\
&\qquad=
\E[\bone\{Y_1>0,Y_2=0\}
\max\{0,-\Delta\bX^\top\bdelta_\ell\}]
+o(\|\bdelta_\ell\|_2)\quad\text{and}
\label{eq:A6_boundary1}
\\
&\E[\bone\{Y_1=0,Y_2>0\}
(\max\{0,E_2+\Delta\bX^\top\bdelta_\ell\}-\max\{0,E_2\})]
\nonumber\\
&\qquad=
\E[\bone\{Y_1=0,Y_2>0\}
\max\{0,\Delta\bX^\top\bdelta_\ell\}]
+o(\|\bdelta_\ell\|_2)
\label{eq:A6_boundary2}
\end{align}
as $\ell\to\infty$. Importantly, these dominated convergence arguments hold
along any diverging sequence $\|\bdelta_\ell\|\to\infty$. Consequently, the
expansion is uniform in the sense that
\begin{align*}
&\lim_{M\to\infty}\sup_{\|\bdelta\|_2>M}
\bigg|\frac{\E[q(\bZ,\bdelta)]}{b(\bdelta)}
-\bigg(\frac{\E[\bone\{Y_1>0,Y_2>0\}
|\Delta\bX^\top\bdelta|]}{\|\bdelta\|_2}\\
&+\frac{\E[\bone\{Y_1>0,Y_2=0\}
\max\{0,-\Delta\bX^\top\bdelta\}]}{\|\bdelta\|_2}+\frac{\E[\bone\{Y_1=0,Y_2>0\}
\max\{0,\Delta\bX^\top\bdelta\}]}{\|\bdelta\|_2}
\bigg)
\bigg|=0
\end{align*}
(If not, there would be a constant $\eta>0$ and a sequence $\{\bdelta_\ell\}$ with
$\|\bdelta_\ell\|_2\to\infty$ such that the remainder is at least $\eta$ for all
$\ell$, contradicting the above expansions.)

The rank condition (Assumption~\ref{assu:RankRegressors-TLAD}) implies
$
\inf_{\|\bdelta\|_2=1}
\E[\bone\{Y_1>0,Y_2>0\}
|\Delta\bX^\top\bdelta|]
>0.
$
Since the dominant components in \eqref{eq:A6_boundary1} and
\eqref{eq:A6_boundary2} are both nonnegative, it follows from the above
decomposition that there is a constant $M\in(0,\infty)$ such that
$
\inf_{\|\bdelta\|_2>M}
\E[q(\bZ,\bdelta)]/b(\bdelta)>0.
$
Noting that $t\mapsto m^{\texttt{tlad}}(t,\by)$ is 1-Lipschitz uniformly in
$\by\in[0,\infty)\times[0,\infty)$, from the Jensen and Cauchy-Schwarz
inequalities and Assumption \ref{assu:MomentConditions-TLAD}, we deduce that
$\bdelta\mapsto\E[q(\bZ,\bdelta)]$ is Lipschitz continuous with Lipschitz
constant $\E[\|\Delta\bX\|_2]<\infty$. In particular,
$\bdelta\mapsto\E[q(\bZ,\bdelta)]$ is continuous. Since $\bdelta=\bzero$ is the
unique minimizer and $\E[q(\bZ,\bzero)]=0$, $\bdelta\mapsto\E[q(\bZ,\bdelta)]$
must be bounded away from zero also on the compact set
$\{\epsilon\leqslant\|\bdelta\|_2\leqslant M\}$. Combining the two regions
yields \eqref{eq:key_identifiability_A6}.

\textbf{Step 2:} Next, define the function class
$\cF:=\{f:\cZ\to\R|f=f(\cdot,\bdelta),\bdelta\in\R^K\}$ by
\[
 f(\bz,\bdelta):=
 \begin{cases}
  q(\bz,\bdelta)/b(\bdelta), & \bdelta \neq \bzero, \\
  0, & \bdelta = \bzero.
 \end{cases}
\] 
Arguments analogous to those used in verifying stochastic equicontinuity show
that the class $\{f(\cdot,\bdelta):\bdelta\in\R^K\}$ is VC-subgraph with
square-integrable envelope $\bz\mapsto\|\Delta\bx\|_2$. Hence it is $P$-Donsker.
It follows from \citet[Theorem~2.1(i)]{AG92} and the triangle inequality that
$$
\bbP\left( \sup_{\bdelta\in C_{\epsilon}}\envert[3]{\frac{1}{n}\sum_{i=1}^n \frac{q(\tbZ_i,\bdelta)}{b(\bdelta)} - \E\left[\frac{q(\bZ,\bdelta)}{b(\bdelta)}\right]} > \frac{a_{\epsilon}}{4} \middle| \{\bZ_i\}_{i=1}^n \right) \overset{\text{a.s.}}{\to}0.
$$
Combining this convergence result with \eqref{eq:key_identifiability_A6}, we
obtain
$
\bbP(\inf_{\bdelta\in C_{\epsilon}} n^{-1}\sum_{i=1}^n q(\tbZ_i,\bdelta)/b(\bdelta) > a_{\epsilon}/2 | \{\bZ_i\}_{i=1}^n )\overset{\mathrm{a.s.}}{\to}1,
$
and so
$
\bbP( \inf_{\bdelta\in C_{\epsilon}} n^{-1}\sum_{i=1}^n q(\tbZ_i,\bdelta) > \epsilon a_{\epsilon}/2 | \{\bZ_i\}_{i=1}^n)\overset{\mathrm{a.s.}}{\to}1.
$
On the other hand, we have
$
\bbP( n^{-1}\sum_{i=1}^n q(\tbZ_i,\bzero) \leqslant \E[q(\bZ,\bzero)] + \epsilon a_{\epsilon}/2 | \{\bZ_i\}_{i=1}^n )\overset{\mathrm{a.s.}}{\to}1.
$
Combining these convergence results and using the fact that $\E[q(\bZ,\bzero)] = 0$ shows that
$
\bbP(\|\tbtheta^{\texttt{tlad}} - \btheta_0\|_2 < \epsilon | \{\bZ_i\}_{i=1}^n)
\overset{\mathrm{a.s.}}{\to}1.
$
Since $\epsilon\in(0,\infty)$ was arbitrary, the previous display produces (A.6)
via the triangle inequality and strong consistency of
$\widehat\btheta^{\texttt{tlad}}$ from (A.5).

\medskip\noindent\textbf{Arcones and Gin{\'e} Remark 3.3 application:} As explained above,
\citet[Remark~3.3]{AG92} implies that the conditional law of $\sqrt
n(\tbtheta^{\texttt{tlad}}-\widehat\btheta^{\tt{tlad}})$ given the
original sample converges in probability in the BL metric to
$\cN(\bzero,\bSigma_0^{\tt{tlad}})$. Fix $\bpsi\in\R^K\backslash\{\bzero\}$ and
define $g:\R^K\to\R$ by $g(\bh)=|\bpsi^\top \bh|$. Since $g$ is Lipschitz, the
continuous mapping property for BL convergence
\citep[e.g.,][Proposition~10.7(i)]{kosorok2008introduction} yields BL
convergence of the conditional law of $g(\sqrt
n(\tbtheta^{\texttt{tlad}}-\widehat\btheta^{\tt{tlad}}))$ to that of
$|\cN(0,\bpsi^\top\bSigma_0^{\tt{tlad}}\bpsi)|$. Because $|\cN(0,\sigma^2)|$ has
a continuous and strictly increasing distribution function on $(0,\infty)$ for
$\sigma^2>0$, its quantiles are unique, and thus the corresponding conditional
quantiles converge in probability. Applying this with $\bpsi$ equal to
coordinate vectors and pairwise sums of coordinate vectors yields the
consistency of $\widehat\bSigma^{\tt{tlad}}$.
\end{proof}



\section{Tools}\label{sec:Tools}

\subsection{A Generalized Dominated Convergence Theorem}
We state and prove a generalized version of the Lebesgue Dominated
Convergence Theorem (LDCT), closely related to
\citet[Theorem 4.11]{royden2023real}, which is stated without proof.
We include a proof for completeness.

Let $(X,\mathcal A,\mu)$ be a measure space. All statements hold
$\mu$-almost everywhere (a.e.) unless stated otherwise.

\begin{thm}
[\textbf{Generalized Lebesgue Dominated Convergence Theorem}]
\label{thm:GLDCT}
Suppose $\{f_n\}_{n\in\N}$ are measurable and $\{g_n\}_{n\in\N}$ are integrable,
$f_n\to f$ a.e., $g_n\to g$ a.e.,
$|f_n|\leqslant g_n$ a.e.\ for all $n$, and
$
\int g_n\dif\mu \to \int g\dif\mu <\infty .
$
Then:
\begin{inparaenum}[(i)]
\item\label{enu:fnInt} $\{f_n\}_{n\in\N}$ are integrable,
\item\label{enu:fInt} $f$ is integrable,
\item\label{enu:fnIntConv} $\int f_n\dif\mu \to \int f\dif\mu$.
\end{inparaenum}
\end{thm}

\begin{proof}[\sc{Proof of Theorem \ref{thm:GLDCT}}]
Since $|f_n|\leqslant g_n$ a.e.\ and $|f_n|\geqslant 0$, we have $g_n\geqslant
0$ a.e.\ for all $n$. Taking limits on the a.e.\ set where $f_n\to f$ and
$g_n\to g$ yields $g\geqslant 0$ a.e.

\ref{enu:fnInt}.
Because $|f_n|\leqslant g_n$ a.e.\ and $g_n$ is integrable, we have
$
\int |f_n|\dif\mu \leqslant \int g_n\dif\mu <\infty,
$
so each $f_n$ is integrable.

\ref{enu:fInt}.
Similarly, $|f|\leqslant g$ a.e.\ and $\int g\dif\mu<\infty$, hence
$
\int |f|\dif\mu \leqslant \int g\dif\mu <\infty,
$
so $f$ is integrable.

\ref{enu:fnIntConv}. By the triangle inequality, $|f_n-f|\leqslant
|f_n|+|f|\leqslant g_n+g$ a.e. Define $h_n:=g_n+g-|f_n-f|$, so $h_n\geqslant 0$
a.e.\ and $h_n\to 2g$ a.e.\ Then Fatou's lemma gives
\[
2\int g\dif\mu
=\int \liminf_{n\to\infty} h_n\dif\mu
\leqslant \liminf_{n\to\infty}\int h_n\dif\mu
= \liminf_{n\to\infty}\Big(\int g_n\dif\mu+\int g\dif\mu-\int |f_n-f|\dif\mu\Big).
\]
Using $\int g_n\dif\mu\to\int g\dif\mu$ yields
$
2\int g\dif\mu \leqslant 2\int g\dif\mu - \limsup_{n\to\infty}\int |f_n-f|\dif\mu,
$
and so it follows that $\int |f_n-f|\dif\mu\to 0$. Hence $\int f_n\dif\mu\to\int f\dif\mu$.
\end{proof}

\subsection{Conditional PDF Existence and Measurability Considerations}\label{sec: measurability appendix} Recall
that we denote by $\mathcal W\subseteq\mathbb R^{K}\times\mathbb R^K\times
\mathbb R$ the support of $\bW = (\bX_1,\bX_2,\alpha)$ and set $\mathcal E =
\mathbb R\times\mathbb R$. In addition, denote by $\mathcal B_{\mathcal
W}:=\{B\cap\mathcal W\mid B\in\mathcal B_{2K+1}\}$ the restriction of
$\mathcal B_{2K+1}$ to $\mathcal W$ and set $\mathcal B_{\mathcal E}= \mathcal
B_2$. To handle measure-theoretic issues in this paper, we assume that $\bW \in
\mathcal W$ and $\bepsilon = (\varepsilon_1,\varepsilon_2)\in \mathcal E$ are
defined on a product probability space $(\mathcal W\times\mathcal E,\mathcal
B_{\mathcal W}\otimes\mathcal B_{\mathcal E},P)$ as coordinate projection maps:
$\bW(\bw,\be)=\bw$ and $\bepsilon(\bw,\be)=\be$ for all $(\bw,\be)\in\mathcal
W\times\mathcal E$, where $P:\mathcal B_{\mathcal W}\otimes\mathcal B_{\mathcal
E}\to[0,1]$ is the joint distribution of the pair $(\bW,\bepsilon)$. We also let
$P_{\bW}:\mathcal B_{\mathcal W}\to[0,1]$ denote the marginal distribution of
$\bW$. 

\begin{lem}\label{lem: regular conditional probability} Let Assumptions
\ref{assu:continuity tls} and \ref{assu:Exchangeability-TLS} hold. Then there
exists a measurable function $(\bw,\be)\mapsto f_{\bepsilon\mid \bw}(\be)$,
mapping $\mathcal W\times\mathcal E$ to $[0,\infty)$, such that (i)
$f_{\bepsilon\mid \bw}(e_1,e_2) = f_{\bepsilon\mid \bw}(e_2,e_1)$ for all
$\bw\in\mathcal W$ and $\be = (e_1,e_2)\in\mathcal E$; (ii) for any
$P$-integrable function $(\bw,\be)\mapsto g(\bw,\be)$,
\begin{equation}\label{eq: i am exhausted}
\E[g(\bW,\bepsilon)] = \int_{\mathcal W}\int_{\mathcal E} g(\bw,\be)f_{\bepsilon\mid \bw}(\be)\lambda_2(\dif \be)P_{\bW}(\dif \bw);
\end{equation}
and, (iii) for all $\bw\in\cW$, $\be\mapsto f_{\bepsilon\mid \bw}(\be)$ is a
PDF.
\end{lem} 

This lemma implies that the function $(\bw,\be)\mapsto
f_{\bepsilon\mid\bw}(\be)$ is a version of the joint PDF of $\bepsilon =
(\varepsilon_1,\varepsilon_2)$ conditional on $\bW = \bw$ for all
$\bw\in\mathcal W$ in the sense that for any $P$-integrable function
$(\bw,\be)\mapsto g(\bw,\be)$,
\begin{equation}\label{eq: version of conditional expectation}
\bw\mapsto \int_{\mathcal E} g(\bw,\be)f_{\bepsilon | \bw}(\be)\lambda_2(\dif \be)\text{ is a version of }\E[g(\bW,\bepsilon)|\bW].
\end{equation}
Throughout the paper, we deal with conditional expectations of functions indexed
by continuous parameters. This could lead to measurability problems as
conditional expectations are defined only up to probability zero sets and those
sets could be parameter-dependent. To ensure that all objects based on
conditional expectations in the paper are suitably measurable, we therefore fix
a version of each conditional expectation given $\bW=\bw$ as the integral over
$\be\mapsto f_{\bepsilon\mid\bw}(\be)$ against the Lebesgue measure $\lambda_2$,
which is justified by \eqref{eq: version of conditional expectation}.

\begin{proof}[\sc{Proof of Lemma \ref{lem: regular conditional probability}}]
Let $\kappa:\mathcal W\times \mathcal B_{\mathcal E}\to[0,1]$ be a (regular)
conditional distribution of $\bepsilon = (\varepsilon_1,\varepsilon_2)$ given
$\bW$. In other words, $\kappa$ is a function such that
(i) for all $\bw\in\mathcal W$, $\kappa(\bw,\cdot)$ is a probability measure
on $(\mathcal E,\mathcal B_{\mathcal E})$;
(ii) for all $B\in\mathcal B_{\mathcal E}$, $\kappa(\cdot,B)$ is measurable;
(iii) for all $A\in\mathcal B_{\mathcal W}$ and $B\in\mathcal B_{\mathcal E}$,
$P(A\times B)=\int_A \kappa(\bw,B)P_{\bW}(\dif \bw)$.

Existence of the function $\kappa$ follows, for example, from Theorems 10.2.1
and 10.2.2 in \citet{D04}. Assumption \ref{assu:continuity tls} then means that
for $P_{\bW}$-almost all $\bw$, the probability measure $\kappa(\bw,\cdot)$ is
absolutely continuous with respect to the Lebesgue measure $\lambda_2$.
Assumption \ref{assu:Exchangeability-TLS} in turn means that for
$P_{\bW}$-almost all $\bw$, the probability measure $\kappa(\bw,\cdot)$
satisfies $\kappa(\bw,\tau(B)) = \kappa(\bw,B)$ for all $B\in\mathcal
B_{\mathcal E}$, where $\tau:\mathcal E\to\mathcal E$ is the swap function
defined by $\tau(e_1,e_2) = (e_2,e_1)$ for all $(e_1,e_2)\in\mathcal E$.

Define the product measure $\bar P := P_{\bW}\otimes\lambda_2$ on $(\mathcal
W\times\mathcal E,\mathcal B_{\mathcal W}\otimes\mathcal B_{\mathcal E})$, and
let $N\in\mathcal B_{\mathcal W}\otimes \mathcal B_{\mathcal E}$ be any
measurable set in $\mathcal W\times\mathcal E$ satisfying $\bar P(N) = 0$. For
all $\bw\in\mathcal W$, let $N_{\bw} := \{\be\in\mathcal E\mid(\bw,\be)\in N\}$
denote the $\bw$-section of $N$. Then, by Tonelli's theorem,
$
\int_{\mathcal W}\lambda_2(N_{\bw})P_{\bW}(\dif \bw) = \bar P(N) = 0.
$
Thus, for $P_{\bW}$-almost all $\bw$, $\lambda_2(N_{\bw})
= 0$. Hence, for $P_{\bW}$-almost all $\bw$, $\kappa(\bw,N_{\bw}) = 0$ by
Assumption \ref{assu:continuity tls} and the union bound (to handle multiple
exceptional sets), so that
$$
P(N) = \int_{\mathcal W} \int_{N_{\bw}} \kappa(\bw,\dif \be)P_{\bW}(\dif \bw) = \int_{\mathcal W}\kappa(\bw,N_{\bw})P_{\bW}(\dif \bw) = 0,
$$ 
where the first equality follows from Part II of Theorem 10.2.1 in \citet{D04}.
Thus, $P$ is absolutely continuous with respect to $\bar P$, and there exists
a measurable function $(\bw,\be)\mapsto h(\bw,\be)\in[0,\infty)$, i.e.~a
(Radon-Nikodym) density of $P$ with respect to $\bar P$, such that
\begin{align}
\int_A \kappa(\bw,B)P_{\bW}(\dif \bw)
& = P(A\times B) = \int_{A\times B} h(\bw,\be) \bar P(\dif \bw,\dif \be) \nonumber\\
& = \int_{A}\int_{B} h(\bw,\be)\lambda_2(\dif \be)P_{\bW}(\dif \bw) \label{eq: kruchu-verchu}
\end{align}
for all $A\in\mathcal B_{\mathcal W}$ and $B\in\mathcal B_{\mathcal E}$, where
the first equality follows from the definition of $\kappa$, the second from the
Radon-Nikodym theorem, and the third from Tonelli's theorem. Also,
\begin{align*}
\int_A\kappa(\bw,B)P_{\bW}(\dif \bw) 
& = \int_A\kappa(\bw,\tau(B))P_{\bW}(\dif \bw) \\
& = \int_{A}\int_{\tau(B)} h(\bw,\be)\lambda_2(\dif \be)P_{\bW}(\dif \bw)  
 = \int_{A}\int_{B} h(\bw,\tau(\be))\lambda_2(\dif \be)P_{\bW}(\dif \bw),
\end{align*}
again for all $A\in\mathcal B_{\mathcal W}$ and $B\in\mathcal B_{\mathcal E}$,
where the first equality follows from Assumption \ref{assu:Exchangeability-TLS},
the second from the same argument as that leading to \eqref{eq: kruchu-verchu},
and the third from the change of variable formula.
Therefore, setting
$$
f_{\bepsilon\mid \bw}(\be) := \frac{h(\bw,\be) + h(\bw,\tau(\be))}{2},\quad\text{for all }(\bw,\be)\in\mathcal W\times\mathcal E,
$$
we thus obtain
$$
\int_A \kappa(\bw,B)P_{\bW}(\dif \bw) = \int_{A}\int_{B} f_{\bepsilon\mid\bw}(\be)\lambda_2(\dif \be)P_{\bW}(\dif \bw)
$$
for all $A\in\mathcal B_{\mathcal W}$ and $B\in\mathcal B_{\mathcal E}$. We
claim that
\begin{equation}\label{eq: dudley almost there}
\kappa(\bw,B) = \int_B f_{\bepsilon\mid\bw}(\be)\lambda_2(\dif \be),\quad\text{for all }B\in\mathcal B_{\mathcal E}\text{ for $P_{\bW}$-almost all $\bw$.}
\end{equation}
To prove this, fix $B\in\cB_\cE$ and consider the sets 
\[
  A:=A\left(B\right):=\left\{\bw\in\cW\middle|\kappa(\bw,B) \neq \int_B f_{\bepsilon\mid\bw}(\be)\lambda_2(\dif \be)\right\},
\]
and
\[
  A_m:=A_m\left(B\right):=\left\{\bw\in\cW\middle|\left|\kappa(\bw,B) - \int_B f_{\bepsilon\mid\bw}(\be)\lambda_2(\dif \be)\right|>\frac{1}{m}\right\},\quad m\in\N,
\]
so that $A=\cup_{m=1}^{\infty}A_m$. Decompose $A_m$ into two disjoint sets,
\begin{align*}
  A_m^{>0}&:=A_m^{>0}\left(B\right):=\left\{\bw\in\cW\middle|\kappa(\bw,B) > \int_B f_{\bepsilon\mid\bw}(\be)\lambda_2(\dif \be) + \frac{1}{m}\right\},\\
  A_m^{<0}&:=A_m^{<0}\left(B\right):=\left\{\bw\in\cW\middle|\kappa(\bw,B) < \int_B f_{\bepsilon\mid\bw}(\be)\lambda_2(\dif \be) - \frac{1}{m}\right\}.
\end{align*}
Since $B\in\cB_{\cE}$ and $A_m^{>0}\in\cB_{\cW}$, we know that
\[
  0 = \int_{A_m^{>0}}\left[\kappa(\bw,B) - \int_B f_{\bepsilon\mid\bw}(\be)\lambda_2(\dif \be)\right]P_{\bW}(\dif \bw) \geqslant \frac{1}{m}P_{\bW}(A_m^{>0}),
\]
which implies that $P_{\bW}(A_m^{>0}) = 0$ for all $m\in\N$. Similar reasoning
shows $P_{\bW}(A_m^{<0}) = 0$ for all $m\in\N$, and thus $P_{\bW}(A_m) = 0$ for
all $m\in\N$. Since $A\in\cB_{\cW}$ and $A=\bigcup_{m=1}^\infty A_m$, the union
bound implies that $P_{\bW}(A) = 0$. Hence, for any $B\in\cB_{\cE}$,
$\kappa(\cdot,B)$ and $\bw\mapsto\int_B f_{\bepsilon\mid\bw}(\be)\lambda_2(\dif
\be)$ can differ on at most a $P_{\bW}$-null set. We next argue that the
exceptional set can be chosen independently of $B\in\cB_{\cE}$. To this end,
consider the sets $\cC:=\cR\cup\{\cE\}$, where $\cR$ is short for the rational
rectangles
$
 \cR := \left\{\left(p,q\right] \times \left(r,s\right] \middle| p,q,r,s \in \mathbb{Q},p<q,r<s\right\}.
$
Since $\cC$ is closed under (non-empty) intersection, it is a $\pi$-system.
Since $\cC$ is countable, our previous calculation and the union bound combine
to show that $P_{\bW}(\cup_{B\in\cC} A(B))  = 0$. Fix
$\bw\in\cW\backslash\cup_{B\in\cC} A(B)$. Then the probability measure $B\mapsto
\kappa(\bw,B)$ and the (non-negative) measure $B\mapsto \int_B
f_{\bepsilon\mid\bw}(\be)\lambda_2(\dif \be)$ agree on the $\pi$-system $\cC$.
Since $\cE\in\cC$ implies
$
1 = \kappa(\bw,\cE) = \int_{\cE} f_{\bepsilon\mid\bw}(\be)\lambda_2(\dif \be),
$
we see that $B\mapsto \int_B f_{\bepsilon\mid\bw}(\be)\lambda_2(\dif \be)$ is
in fact a probability measure on $(\cE,\cB_{\cE})$ with $\be \mapsto
f_{\bepsilon\mid\bw}(\be)$ being the associated (Lebesgue) PDF. As these two
probability measures agree on the $\pi$-system $\cC$, and the rational
rectangles $\cR(\subset\cC)$ generate the Borel $\sigma$-algebra on $\cE$, it
follows that the two probability measures agree on all of $\cB_{\cE}$
\citep[Theorem 3.3]{billingsley1995probability}. We have thus established the
claim in \eqref{eq: dudley almost there}.

The function $(\bw,\be)\mapsto f_{\bepsilon\mid \bw}(\be)$ defined above
is measurable and satisfies $f_{\bepsilon\mid \bw}(e_1,e_2) = f_{\bepsilon\mid
\bw}(e_2,e_1)$ for all $\bw\in\mathcal W$ and $(e_1,e_2)\in\mathcal E$ by
construction. Also, for any $P$-integrable function $(\bw,\be)\mapsto
g(\bw,\be)$, \eqref{eq: i am exhausted} holds by \eqref{eq: dudley almost there}
and Theorem 10.2.1 in \citet{D04}. In addition,
$f_{\bepsilon\mid\bw}(\be)\in[0,\infty)$ for all $(\bw,\be)\in\mathcal
W\times\mathcal E$ by construction. 

For the third and final claim, observe that, as part of the argument
establishing \eqref{eq: dudley almost there}, we showed that $\be\mapsto
f_{\bepsilon\mid\bw}(\be)$ is a (Lebesgue) PDF for each
$\bw\in\cW\backslash\cup_{B\in\cC} A(B)$ with $\cup_{B\in\cC} A(B)$ being a
$P_{\bW}$-null set. For (the exceptional) $\bw\in\cup_{B\in\cC} A(B)$, if any,
we redefine $\be\mapsto f_{\bepsilon\mid\bw}(\be)$ as
$\be\mapsto\varphi(e_1)\varphi(e_2)$ with $\varphi$ being the standard normal
PDF. This modification completes the proof.
\end{proof}

\end{document}